\pdfoutput=1
\documentclass[english]{article}
\usepackage[table]{xcolor}
\usepackage{amsmath,amsthm} 
\usepackage{graphicx}
\usepackage{amssymb}
\usepackage{enumitem}
\usepackage{natbib} 
\setcitestyle{round}
\usepackage[letterpaper]{geometry} 
\usepackage{appendix}
\usepackage{cancel}
\usepackage{mathtools}

\definecolor{medgreen}{rgb}{0.0, 0.75, 0.0}

\usepackage{thmtools,thm-restate}

\usepackage{tikz}
\usepackage{pgf} 
\usetikzlibrary{patterns,automata,arrows,shapes,snakes,topaths,trees,backgrounds,positioning,through,calc,tikzmark}
\usepackage{setspace}
\usepackage{multicol}
\usepackage{graphicx}
\usepackage{xcolor}

\theoremstyle{definition}
\newtheorem{theorem}{Theorem}[section]

\newtheorem{proposition}[theorem]{Proposition}

\newtheorem{definition}[theorem]{Definition}

\newtheorem{lemma}[theorem]{Lemma}

\newtheorem{example}[theorem]{Example}

\newtheorem{remark}[theorem]{Remark}

\makeatletter  

\usepackage{stmaryrd}  
\usepackage{comment}
\usepackage{hyperref}
\hypersetup{colorlinks=true, citecolor=blue, linkcolor=blue}

\newcommand{\M}{\mathcal{M}}
\newcommand{\uph}{\!\!\upharpoonright\!}
\newcommand{\margin}{\mathrm{Margin}}
\newcommand{\splitn}{\mathrm{Split\#}}
\newcommand{\strength}{\mathrm{Strength}}
\newcommand{\weakness}{\mathrm{Weakness}}
\newcommand{\stris}{\mathrm{Strength}^{\mathit{is}}}
\newcommand{\tolerantpi}{Tolerant Positive Involvement}
\newcommand{\pfp}{\mathbf{P}}
\newcommand{\pfq}{\mathbf{Q}}

\makeatother

\usepackage{babel} 
\onehalfspace
 
\begin{document} 

\title{An Axiomatic Characterization of Split Cycle}

\date{{\normalsize Forthcoming in \textit{Social Choice and Welfare}.}}

\author{Yifeng Ding$^*$, Wesley H. Holliday$^\dagger$, and Eric Pacuit$^\ddagger$\\ \\
$*$ Peking University {\normalsize (\href{mailto:yf.ding@pku.edu.cn}{yf.ding@pku.edu.cn})} \\
 $\dagger$ University of California, Berkeley {\normalsize (\href{mailto:wesholliday@berkeley.edu}{wesholliday@berkeley.edu})} \\
 $\ddagger$ University of Maryland {\normalsize (\href{mailto:epacuit@umd.edu}{epacuit@umd.edu})}}

\maketitle

\begin{abstract}
    A number of rules for resolving majority cycles in elections have been proposed in the literature. Recently, Holliday and Pacuit (Journal of Theoretical Politics 33 (2021) 475-524) axiomatically characterized the class of rules refined by one such cycle-resolving rule, dubbed Split~Cycle: in each majority cycle, discard the majority preferences with the smallest majority margin. They showed that any rule satisfying five standard axioms plus a weakening of Arrow's Independence of Irrelevant Alternatives~(IIA), called Coherent IIA, is refined by Split Cycle. In this paper, we go further and show that Split Cycle is the \textit{only} rule satisfying the axioms of Holliday and Pacuit together with two additional axioms, which characterize the class of rules that refine Split Cycle: Coherent Defeat and Positive Involvement in Defeat. Coherent Defeat states that any majority preference not occurring in a cycle is retained, while Positive Involvement in Defeat is closely related to the well-known axiom of Positive Involvement (as in J. P\'{e}rez, Social Choice and Welfare 18 (2001) 601-616). We characterize Split~Cycle not only as a collective choice rule but also as a social choice correspondence, over both profiles of linear ballots and profiles of ballots  allowing ties.
\end{abstract}

\tableofcontents

\section{Introduction}\label{Sec:Intro}

The possibility of cycles in the majority relation of an election---wherein for candidates $c_1,\dots,c_n$, a majority of voters prefer  $c_1$ to $c_2$, a majority prefer $c_2$ to $c_3$, and so on, while a majority prefer $c_n$ to $c_1$---has been taken to show that  ``majority rule is fatally flawed by an internal inconsistency'' (\citealt[p.~59]{Wolf1970}). Yet voting theorists have studied many collective choice rules based on pairwise majority comparisons of candidates designed to resolve majority cycles. These collective choice rules output a binary relation between candidates, which we will call the relation of \textit{defeat}, that is guaranteed to be free of cycles. One prominent family of such rules resolves cycles by paying attention to the size of majority victories, e.g., as measured by the majority \textit{margin} between candidates $x$ and $y$, defined as the number of voters who prefer $x$ to $y$ minus the number who prefer $y$ to $x$. Examples include the Ranked Pairs (\citealt{Tideman1987}, \citealt{ZavistTideman1989}), River (\citealt{Heitzig2004b}), Beat Path (\citealt{Schulze2011,Schulze2018}), Kemeny (\citealt{kemeny1959}),\footnote{Though Kemeny is not usually defined in terms of pairwise margins of victory, it can be so defined as shown in \citealt[p.~87]{Fischer2016}. Another rule usually defined in terms of individual preferences but also definable in terms of pairwise margins is the Borda count (see \citealt[p.~28]{Zwicker2016}).}  Weighted Covering (\citealt{Dutta1999}, \citealt{Fernandez2018}), and Split Cycle (\citealt{HP2021,HP2020}) rules. For instance, according to  Split Cycle,\footnote{Eppley's \citeyearpar{Eppley2000} ``Beatpath Criterion Method'' can be defined in the same way but using \textit{winning votes} instead of \textit{margins} as the measure of strength of majority preference. This distinction does not matter for our axiomatization in the context of linear ballots, but it does matter for our axiomatization in the context of ballots allowing ties (see Section \ref{Sec:Ties}).} majority cycles are resolved as follows:
\begin{enumerate}
    \item For each majority cycle, identify the pairwise majority victories with the smallest margin in that cycle.
    \item A candidate $a$ defeats a candidate $b$ according to Split Cycle if and only if $a$ has a pairwise majority victory over $b$ that was not identified in step 1.
\end{enumerate}
In other words, a majority victory of $a$ vs.~$b$ counts as a defeat of $b$ if and only if in \emph{each} majority cycle in which that majority victory appears, it does not have the smallest margin in the cycle. The resulting defeat relation contains no cycles. See Figure \ref{FirstEx} for an example.

Faced with such a rule for resolving cycles, the question becomes: why this rule and not something else? The question can be partially answered by identifying axioms that distinguish between known rules. A deeper answer comes from a complete \textit{axiomatic characterization} of a rule as the unique rule satisfying some list of natural axioms. For example, such axiomatic characterizations exist for the collective choice rules that rank candidates by Copeland score (\citealt{Rubinstein1980}) and Borda scores (\citealt{Nitzan1981}, \citealt{Mihara2017}).\footnote{A conjectured axiomatization of Ranked Pairs can be found in \citealt{Tideman1987}.}

\begin{figure}

\begin{center}
\begin{tikzpicture}
\node[circle,draw,minimum width=0.25in] at (0,0)      (a) {$a$}; 
\node[circle,draw,minimum width=0.25in] at (3,0)      (b) {$c$}; 
\node[circle,draw,minimum width=0.25in] at (1.5,1.5)  (c) {$b$}; 
\node[circle,draw,minimum width=0.25in] at (1.5,-1.5) (d) {$d$};
\path[->,draw,thick] (b) to[pos=.7] node[fill=white] {$8$} (a);
\path[->,draw,thick] (c) to node[fill=white] {$2$} (a);
\path[->,draw,thick] (a) to node[fill=white] {$10$} (d);
\path[->,draw,thick] (c) to node[fill=white] {$4$} (b);
\path[->,draw,thick] (d) to node[fill=white] {$12$} (b);
\path[->,draw,thick] (d) to[pos=.7]  node[fill=white] {$6$} (c);
\end{tikzpicture}\hspace{1in}
\begin{tikzpicture}
\node[circle,draw,minimum width=0.25in] at (0,0)      (a) {$a$}; 
\node[circle,draw,minimum width=0.25in] at (3,0)      (b) {$c$}; 
\node[circle,draw,minimum width=0.25in] at (1.5,1.5)  (c) {$b$}; 
\node[circle,draw,minimum width=0.25in] at (1.5,-1.5) (d) {$d$};

\path[->,draw,thick] (a) to node[fill=white] {$D$} (d);
 
\path[->,draw,thick] (d) to node[fill=white] {$D$} (b);
\path[->,draw,thick] (d) to[pos=.7]  node[fill=white] {$D$} (c);
\end{tikzpicture}
\end{center}
\caption{A margin graph (left) with three cycles: $(a,d,b,a)$, $(a,d,c,a)$, and $(a,d,b,c,a)$. Deleting the weakest edge in each cycle (namely $(b,a)$ in the first, $(c,a)$ in the second, and $(b,c)$ in the third) results in the Split Cycle defeat graph (right).}\label{FirstEx}
\end{figure}
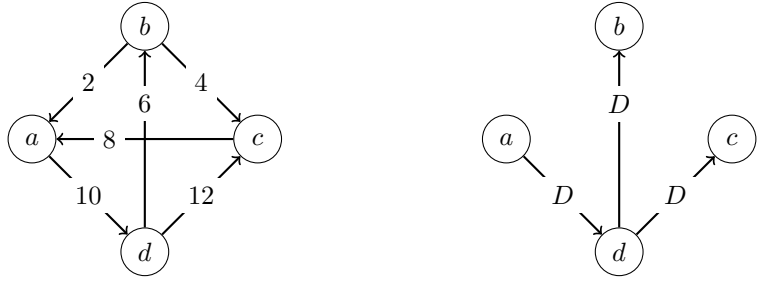

Recently Holliday and Pacuit \citeyearpar{HP2021} have characterized the class of rules refined by the Split Cycle  rule\footnote{To say that a rule is \textit{refined by} Split Cycle means that if a candidate $x$ defeats a candidate $y$ according to the rule, then $x$ also defeats $y$ according to Split Cycle.} using six axioms, five of which are standard (see Section \ref{Sec:StandAxVCCR} below) and the sixth of which is a weakening of Arrow's \citeyearpar{Arrow1963} axiom of Independence of Irrelevant Alternatives (IIA). They call their new axiom Coherent~IIA. Recall that IIA states that for any two profiles $\mathbf{P}$ and $\mathbf{P}'$ of voter preferences, if $\mathbf{P}$ and $\mathbf{P}'$ are the same with respect to how each voter ranks $x$ vs.~$y$, then if $x$ defeats $y$ in $\mathbf{P}$, $x$ must also defeat $y$ in $\mathbf{P}'$. The motivation for weakening IIA to   Coherent IIA can be seen in a simple example from \citealt{HP2021}. In the profile $\mathbf{P}$ in Figure~\ref{Fig1}, displayed alongside its corresponding margin graph, arguably any sensible collective choice rule should judge that candidate $a$ defeats candidate $b$. However, in the profile $\mathbf{P}'$ in Figure~\ref{Fig1}, no sensible collective choice rule---technically, no collective choice rule that is anonymous, neutral, and guarantees that some undefeated candidate exists---can judge that $a$ defeats $b$, despite $a$ still beating $b$ head-to-head by a margin of $n$. This is a counterexample to IIA as a normative requirement. Holliday and Pacuit argue that the ``Fallacy of IIA'' is to ignore how the context of a full election can force us to suspend judgment on some relations of defeat that we could coherently accept in a different context. 

\begin{figure}[h]
\begin{center}
$\mathbf{P}$\qquad
\begin{tabular}{ccc}
$n$ & $n$ & $n$   \\\hline
$\boldsymbol{a}$ & $\boldsymbol{b}$ &  $c$ \\
$\boldsymbol{b}$ &  $\boldsymbol{a}$ & $\boldsymbol{a}$ \\
$c$ &  $c$ &  $\boldsymbol{b}$ \\
\end{tabular}\hspace{1.32in}
\begin{tikzpicture}[baseline=(current bounding box.center)]
\node[circle,draw, minimum width=0.25in] at (0,0) (a) {$a$}; 
\node[circle,draw,minimum width=0.25in] at (3,0) (c) {$c$}; 
\node[circle,draw,minimum width=0.25in] at (1.5,1.5) (b) {$b$}; 

\path[->,draw,thick] (a) to node[fill=white] {$n$} (b);
\path[->,draw,thick] (b) to node[fill=white] {$n$} (c);
\path[->,draw,thick] (a) to node[fill=white] {$n$} (c);
\end{tikzpicture}
\end{center}

\begin{center}
$\mathbf{P}'$\qquad
\begin{tabular}{ccc}
$n$ & $n$ & $n$   \\\hline
$\boldsymbol{a}$ & $\boldsymbol{b}$ &  $c$ \\
$\boldsymbol{b}$ &  $c$ & $\boldsymbol{a}$ \\
$c$ &  $\boldsymbol{a}$ &  $\boldsymbol{b}$ \\
\end{tabular}\hspace{1.32in}
\begin{tikzpicture}[baseline=(current bounding box.center)]

\node[circle,draw, minimum width=0.25in] at (0,0) (a) {$a$}; 
\node[circle,draw,minimum width=0.25in] at (3,0) (c) {$c$}; 
\node[circle,draw,minimum width=0.25in] at (1.5,1.5) (b) {$b$}; 

\path[->,draw,thick] (a) to node[fill=white] {$n$} (b);
\path[->,draw,thick] (b) to node[fill=white] {$n$} (c);
\path[->,draw,thick] (c) to node[fill=white] {$n$} (a);
\end{tikzpicture}
\end{center}
\caption{Profiles (left) and their margin graphs (right) illustrating the ``Fallacy of IIA.''}\label{Fig1}
\end{figure}
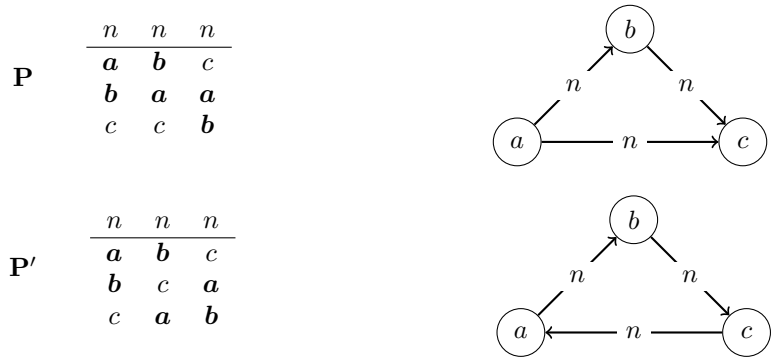

In the context of a majority cycle in which  $x$ is majority preferred to $y$, the electorate is \textit{incoherent with respect to $x$ and $y$} in the sense that while there is an argument that $x$ should defeat $y$, in virtue of the majority preference for $x$ over $y$, there is also an opposing argument that $x$ should \textit{not} defeat $y$, in virtue of the path of majority preferences from $y$ to $x$; e.g., $y$ is majority preferred to $z$, who is majority preferred to~$x$. One natural measure of the degree of this incoherence is the strength of the opposing argument, which is some monotonic function of the \textit{margins} of the relevant majority preferences; e.g., the larger the margin of $y$ over $z$ or of $z$ over $x$, the more incoherent with  the majority preference for $x$ over $y$---and the smaller those margins, the less incoherent with the majority preference for $x$ over $y$. On this view, the following is sufficient  for $\mathbf{P}'$ to be \textit{no more incoherent than $\mathbf{P}$ with respect to $x$ and $y$}: the margin graph of $\mathbf{P}'$ is obtained from that of $\mathbf{P}$ by deleting or reducing the margins on zero or more edges not connecting $x$ and $y$  or by deleting zero or more candidates other than $x$ and $y$. Adopting this view about incoherence, Holliday and Pacuit accept the core intuition behind IIA whenever contextual incoherence does not interfere, leading to their axiom of Coherent IIA, which can be stated informally as follows: 
\begin{itemize}
    \item Coherent IIA (informally): if $\mathbf{P}$ and $\mathbf{P}'$ are the same with respect to how each voter ranks $x$ vs.~$y$, $x$~defeats $y$ in $\mathbf{P}$, and  \textit{$\mathbf{P}'$ is not more incoherent than $\mathbf{P}$ with respect to $x$ and $y$}, then $x$ must also defeat $y$ in $\mathbf{P}'$.
    \end{itemize}
    Holliday and Pacuit then prove that Split Cycle is the most resolute collective choice rule\footnote{Technically, they characterize Split Cycle as what they call a \textit{variable-election} collective choice rule, whose domain contains profiles with different sets of candidates and voters (see Section \ref{Sec:VCCRVSCC} below). Until the end of this section, we use `collective choice rule' to refer to the variable-election variety.} satisfying the five standard axioms plus Coherent IIA. Here ``most resolute'' means that for any other rule $f$ that satisfies the six axioms, if $x$ defeats $y$ according to $f$, then $x$ defeats $y$ according to Split Cycle.

Holliday and Pacuit's theorem may be viewed as characterizing Split Cycle as the unique collective choice rule satisfying their six axioms plus a seventh axiom stating that the rule should be the most resolute rule satisfying the first six axioms. In a sense, however, this characterization using the notion of resoluteness is only half of a characterization of Split Cycle.\footnote{As Holliday and Pacuit \citeyearpar[p.~501]{HP2021} write, ``A natural next step would be to obtain another axiomatic characterization of Split Cycle as the only VCCR satisfying some axioms without reference to resoluteness.''} We would like axioms on a collective choice rule such that for any rule $f$ satisfying the axioms, $x$ defeats $y$ according to $f$ \textit{if and only if} $x$ defeats $y$ according to Split Cycle. Such an axiomatic characterization is the main result of the present paper.

Our new characterization of Split Cycle involves two natural axioms, which we prove are satisfied by exactly the collective choice rules that refine Split Cycle:
\begin{itemize}
    \item Coherent Defeat: if a majority of voters prefer $x$ to $y$, and there is no majority cycle involving $x$ and $y$, then $x$ \textit{defeats} $y$.
        \item Positive Involvement in Defeat: if $y$ does \textit{not} defeat $x$ in profile $\mathbf{P}$, and $\mathbf{P}'$ is obtained from $\mathbf{P}$ by adding one new voter who ranks $x$ above $y$, then $y$ still does not defeat $x$ in $\mathbf{P}'$.
\end{itemize}
Coherent Defeat is a point of common ground between Split Cycle, Ranked Pairs, Beat Path, and GOCHA (\citealt{Schwartz1986}): in the absence of cyclic incoherence, majority preference is sufficient for defeat. We will discuss other ways of motivating Coherent Defeat below, drawing on \citealt{Heitzig2002}.\footnote{Heitzig's \citeyearpar{Heitzig2004} definition of an ``immune'' candidate is the same as that of a Split Cycle winner if we replace ``stronger'' with ``at least as strong'' in his definition.} The intuition behind Positive Involvement in Defeat, a variable-electorate axiom, is similar to the intuition behind fixed-electorate  \textit{monotonicity} axioms for collective choice rules, as stated in \citealt{Blair1982}: ``additional support for a (pairwise) winning alternative leaves it winning'' (p.~935). Perhaps surprisingly, many collective choice rules violate Positive Involvement in Defeat, including Ranked Pairs, River, Beat Path, and Covering (see \citealt{Duggan2013} for many versions) viewed as collective choice rules. As we will explain, the name is borrowed from the related axiom of ``Positive Involvement'' for functions that output winners instead of a defeat relation (see \citealt{Saari1995}, \citealt{Perez2001}, \citealt{HP2021PI}). Our main result in this paper is that Split Cycle is the unique collective choice rule satisfying the five standard axioms, Coherent IIA, Coherent Defeat, and Positive Involvement in Defeat. The proof uses the well-known concept of minimal cuts in graph theory.

So far we have discussed Split Cycle as a \textit{collective choice rule}  that outputs for a given profile a binary relation of defeat on the set of candidates---or more precisely, what we call a \textit{variable-election} collective choice rule (VCCR), whose domain includes profiles with different set of candidates and voters. In this paper, we also characterize the associated variable-election \textit{social choice correspondence} (VSCC) that outputs for a given profile the set of undefeated candidates. This involves a translation from VSCCs to VCCRs which induces a translation in the reverse direction from IIA (resp.~Coherent IIA) for VCCRs to IIA (resp.~Coherent IIA) for VSCCs. The additional axioms of Coherent Defeat and Positive Involvement in Defeat also have natural analogues for VSCCs. Interestingly, the VSCC analogue of the VCCR axiom of Positive Involvement in Defeat, which we call Tolerant Positive Involvement, strengthens the axiom of Positive Involvement from the prior literature mentioned above. We prove that the Split Cycle VSCC is the unique VSCC satisfying five standard axioms, Coherent IIA, Coherent Defeat, and Tolerant Positive Involvement.

The use in our axiomatizations of Coherent IIA, on the one hand, and variants of Positive Involvement, on the other, nicely matches the two main normative points made about the Split Cycle VSCC in \citealt{HP2020}, which stresses Split Cycle's ability to handle both (i) the ``Problem of Spoilers'' and (ii) the ``Strong No Show Paradox'' (\citealt{Perez2001}). Problem (i) is a \textit{variable-candidate} problem, wherein adding a new candidate to an election who loses overall and loses head-to-head to a candidate $x$ causes $x$ to go from winning to losing, while (ii) is a \textit{variable-voter} problem, wherein adding new voters to an election who rank a candidate $x$ in  first place causes $x$ to go from winning to losing. In essence, Coherent IIA is a strengthening (assuming other uncontroversial axioms) of anti-spoiler axioms that mitigate (i), while Tolerant Positive Involvement is a strengthening of the Positive Involvement axiom that prevents (ii).

The rest of the paper is organized as follows. In Section \ref{Sec:Prelim}, we formally define VCCRs and VSCCs in general and the Split Cycle VCCR and VSCC in particular. In Section \ref{Sec:AxVCCR}, we review the axioms from \citealt{HP2021} (Sections \ref{Sec:StandAxVCCR}-\ref{Sec:CohIIAVCCR}) and discuss in more depth the additional axioms of Coherent Defeat (Section \ref{Sec:CohDefVCCR}) and Positive Involvement in Defeat (Section \ref{Sec:PosInvVCCR}). Using these axioms, we prove our  characterization result for the Split Cycle VCCR in Section \ref{Sec:CharVCCR}. We then turn to the Split Cycle VSCC. In Section \ref{Sec:AxVSCC}, we define analogues for VSCCs of the axioms for VCCRs in Section \ref{Sec:AxVCCR}. Using these analogous axioms, we prove our characterization result for the Split Cycle VSCC in Section \ref{Sec:CharVSCC}. In Section \ref{Sec:Ties}, we adapt our characterization results to the setting in which ties are permitted in voters' ballots. In Section \ref{sec:necessity}, we show that the three special axioms---Coherent IIA, Positive Involvement in Defeat (or Tolerant Positive Involvement for VSCCs), and Coherent Defeat---are all necessary in our axiomatization. Finally, we conclude with some suggestions for related axiomatization problems in voting theory in Section \ref{Sec:Conclusion}.

\newpage

\section{Preliminaries}\label{Sec:Prelim}
\subsection{VCCRs and VSCCs}\label{Sec:VCCRVSCC}

Fix infinite sets $\mathcal{V}$ and $\mathcal{X}$ of \textit{voters} and \textit{candidates}, respectively.  We begin by assuming that voters submit linear orders of the candidates. For $X\subseteq\mathcal{X}$, let $\mathcal{L}(X)$ be the set of all strict linear orders on $X$. In Section~\ref{Sec:Ties}, we will adapt our results to the setting in which voters may submit strict weak orders, allowing ties.

 \begin{definition}\label{ProfileDef}
 A (\textit{linear}) \textit{profile} is a function $\mathbf{P}: V\to
 \mathcal{L}(X)$ for some nonempty finite $V\subset \mathcal{V}$ and nonempty
 finite $X\subset \mathcal{X}$, which we denote by $V(\mathbf{P})$ (called the
 set of \textit{voters in $\mathbf{P}$})  and $X(\mathbf{P})$ (called the set of
 \textit{candidates in $\mathbf{P}$}), respectively. We call $\mathbf{P}(i)$
 voter $i$'s \textit{ballot}. When $\mathbf{P}$ is clear from context, we write
 `$x\succ_iy$' for $(x,y)\in\mathbf{P}(i)$, and when $i$ is also clear from
 context, we write $x\succ y$ for $ x\succ_iy$.
 
 Given a nonempty set $Y\subseteq X(\mathbf{P})$ of candidates, the restriction $\mathbf{P}_{|Y}$ of $\mathbf{P}$ to $Y$ is the profile such that for each $i\in V(\mathbf{P})$,  $\mathbf{P}_{|Y}(i)$ is the restriction of the ballot $\mathbf{P}(i)$ to $Y$.
\end{definition}

 Sometimes we will display profiles in their \textit{anonymized} form, as in Figure \ref{Fig1}, meaning that we only display how many voters have each type of ballot, rather than indicating which voters have which ballots. The finiteness of the sets of voters and candidates in profiles will be used implicitly in many constructions in later proofs.

Our main objects of study are (i) functions that assign to each profile a binary relation on the set of candidates and (ii) functions that assign to each profile a subset of the candidates.

\begin{definition} A \textit{variable-election collective choice rule} (VCCR) is a function $f$ on the domain of all profiles such that for any profile $\mathbf{P}$, $f(\mathbf{P})$ is an asymmetric binary relation on $X(\mathbf{P})$, which we call \textit{the defeat relation for $\mathbf{P}$ according to $f$}. For $x,y\in X(\mathbf{P})$, we say that \textit{$x$ defeats $y$ in $\mathbf{P}$ according to $f$} when $(x,y)\in f(\mathbf{P})$.
\end{definition}

\begin{definition} A \textit{variable-election social choice correspondence} (VSCC) is a function $F$ on the domain of all profiles such that for any profile $\mathbf{P}$, we have $\varnothing\neq F(\mathbf{P})\subseteq X(\mathbf{P})$.
\end{definition}

The point of the adjective `variable-election' is that different profiles in the domain of a function may have different sets of voters and candidates. By contrast, in the literature, collective choice rules (CCRs) and social choice correspondences (SCCs) are often defined such that all profiles in the domain of a given function have the same sets of voters and candidates, respectively.

Recall that a binary relation $R$ is \emph{acyclic} if there is no sequence $(x_0, x_1, \dots, x_{n+1})$ (with $n \ge 0$) such that $x_0 = x_{n+1}$ and for any $i = 0 \dots n$, $(x_i, x_{i+1}) \in R$. A VCCR $f$ is said to be \textit{acyclic} if for any profile $\mathbf{P}$, the binary relation $f(\mathbf{P})$ is acyclic. Such a VCCR induces a VSCC $\overline{f}$ that returns for a given profile $\mathbf{P}$ the maximal (undefeated) elements of $f(\mathbf{P})$.

\begin{definition}\label{InducedVSCC}
 For any VCCR $f$, let $\overline{f}$ be the function defined on all profiles such that for any profile~$\mathbf{P}$, 
 \[
 \overline{f}(\mathbf{P})=\{x\in X(\mathbf{P})\mid \mbox{there is no }y\in X(\mathbf{P})\colon y\mbox{ defeats }x\mbox{ in $\mathbf{P}$ according to }f\}.
 \]
 If $\overline{f}$ is a VSCC, we say that $\overline{f}$ is the VSCC \emph{defeat-rationalized} by $f$.
\end{definition}

\begin{lemma}\label{VCCRtoVSCC} Given any acyclic VCCR $f$, the function $\overline{f}$ defined above is a VSCC, the VSCC defeat-rationalized by $f$, since $\varnothing\neq \overline{f}(\mathbf{P})\subseteq X(\mathbf{P})$ for any profile $\mathbf{P}$.
\end{lemma}

Our axiomatization proofs rely heavily on establishing that one VCCR (resp.~VSCC) refines another in the following sense.

\begin{definition}
  Let $f$ and $f'$ be VCCRs. We say that $f$ \emph{refines} $f'$ if for any profile $\mathbf{P}$, $f(\mathbf{P}) \supseteq
  f'(\mathbf{P})$; that is, $f$ outputs all the defeats that $f'$ does and possibly more.

  Let $F$ and $F'$ be VSCCs. We say that $F$ \emph{refines} $F'$ if for any profile $\mathbf{P}$, $F(\mathbf{P}) \subseteq
  F'(\mathbf{P})$; that is, $F$ always selects a subset of the candidates selected by $F'$.
\end{definition}

\subsection{Split Cycle}

In this section, we define the Split Cycle VCCR and VSCC. The definition starts with the notion of the margin graph of a profile and of majority paths and cycles in the margin graph.\footnote{Since Split Cycle needs only the information in the margin graph of a profile, which is an example of a \textit{weighted} (weak) \textit{tournament}, Split Cycle may also be regarded as a \textit{weighted tournament solution} (see \citealt{Fischer2016}).}

\begin{definition}
  Let $\mathbf{P}$ be a profile and $x, y \in X(\mathbf{P})$. The \textit{margin
    of $x$ over $y$} in $\mathbf{P}$, written $\margin_{\mathbf{P}}(x,
  y)$, is
\[
  |\{i \in V(\mathbf{P}) \mid (x,y)\in\mathbf{P}(i)\}| - |\{i \in V(\mathbf{P}) \mid
  (y,x)\in \mathbf{P}(i)\}|.
\]
If $\margin_{\mathbf{P}}(x,y)>0$, we say that $x$ is \textit{majority preferred} to $y$ in $\mathbf{P}$.

The \textit{margin graph} of $\mathbf{P}$, denoted $\mathcal{M}(\mathbf{P})$, is
the directed graph with weighted edges such that its set of nodes is $X(\mathbf{P})$ and there is an edge from $x$ to $y$ iff $\margin_\mathbf{P}(x, y) > 0$, the
weight of which is  $\margin_\mathbf{P}(x, y)$.

A \textit{majority path} in $\mathbf{P}$ is  a path in
$\mathcal{M}(\mathbf{P})$, i.e., a sequence $\rho = ( x_1, x_2, \cdots,
x_n )$ of nodes of $\mathcal{M}(\mathbf{P})$ such that for each $i = 1
\ldots n-1$ there is an edge from $x_i$ to $x_{i+1}$, i.e., $\margin_{\mathbf{P}}(x_i, x_{i+1}) > 0$. Given such a majority path
$\rho$, we define its \emph{strength}, written $\strength_{\mathbf{P}}(\rho)$, 
as
\[
  \min\{\margin_{\mathbf{P}}(x_i, x_{i+1}) \mid i = 1 \ldots n-1\}.
\]
Such a $\rho$ is also called a majority path \textit{from} $x_1$ \textit{to} $x_n$. When $x_n = x_1$, we call such a $\rho$ a \emph{majority cycle}, and in this case $\strength_{\mathbf{P}}(\rho)$ is also called the \emph{splitting number} of $\rho$, denoted $\splitn_{\mathbf{P}}(\rho)$. A \textit{simple majority path} is a majority path in which no candidate is repeated, while a \emph{simple majority cycle} is a majority cycle in which no candidate is repeated except the first and the last. 
\end{definition}

There are many equivalent ways to define the Split Cycle VCCR, four of which are based on following lemma (for proofs, see \citealt{HP2020} and \citealt{HNP2021}).

\begin{lemma}\label{EquivDef} For any profile $\mathbf{P}$ and $x, y \in X(\mathbf{P})$, the following are equivalent:
\begin{enumerate}
    \item $\margin_{\mathbf{P}}(x, y) > 0$ and $\margin_{\mathbf{P}}(x, y) > \splitn_{\mathbf{P}}(\rho)$ for every majority cycle $\rho$ containing $x$ and $y$;
        \item $\margin_{\mathbf{P}}(x, y) > 0$ and $\margin_{\mathbf{P}}(x, y) > \splitn_{\mathbf{P}}(\rho)$ for every simple majority cycle $\rho$ containing $x$ and~$y$;
        \item $\margin_{\mathbf{P}}(x, y) > 0$ and $\margin_{\mathbf{P}}(x, y) > \splitn_{\mathbf{P}}(\rho)$ for every simple majority cycle $\rho$ in which $y$ directly follows $x$;
    \item $\margin_{\mathbf{P}}(x, y) > 0$ and $\margin_{\mathbf{P}}(x, y) > \strength_{\mathbf{P}}(\rho)$ for every simple majority path $\rho$ from $y$ to $x$.
\end{enumerate}
\end{lemma}

\begin{definition}
  The Split Cycle VCCR, denoted $sc$, is defined as follows: for any profile
  $\mathbf{P}$ and  $x, y \in X(\mathbf{P})$, $(x, y) \in sc(\mathbf{P})$ iff the condition on $x,y$ in Lemma \ref{EquivDef} holds. As is shown in \citealt{HP2020}, $sc$ is an acyclic VCCR. Thus, we define the Split Cycle VSCC $SC$ as $\overline{sc}$.
\end{definition}

\noindent The two-step algorithm in Section \ref{Sec:Intro} provides one way of
calculating $sc(\mathbf{P})$. 

\begin{example}Consider the following anonymized profile with its margin graph:
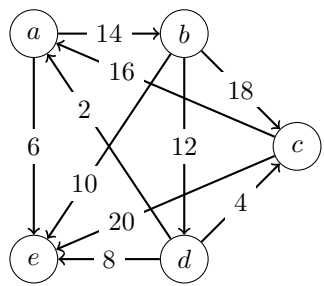
\begin{center}

\setlength{\tabcolsep}{3pt}
\begin{tabular}{ccccccccccc}
$9$ & $5$ & $3$ & $1$ & $8$ & $4$ & $3$ & $7$ & $4$ & $3$ & $1$\\\hline 
$b$ & $b$ & $a$ & $a$ & $a$ & $a$ & $d$ & $d$ & $c$ & $c$ & $c$\\ 
$d$ & $c$ & $b$ & $b$ & $e$ & $d$ & $b$ & $c$ & $a$ & $e$ & $d$\\ 
$c$ & $e$ & $e$ & $c$ & $b$ & $b$ & $e$ & $a$ & $b$ & $a$ & $e$\\ 
$e$ & $d$ & $d$ & $e$ & $c$ & $c$ & $c$ & $e$ & $d$ & $d$ & $a$\\ 
$a$ & $a$ & $c$ & $d$ & $d$ & $e$ & $a$ & $b$ & $e$ & $b$ & $b$
\end{tabular}
\hspace{1in}
\begin{tikzpicture}[baseline=(current bounding box.center)]
\node[circle,draw,minimum width=0.25in] at (2,1.5)  (a) {$b$}; 
\node[circle,draw,minimum width=0.25in] at (0,1.5)  (b) {$a$}; 
\node[circle,draw,minimum width=0.25in] at (0,-1.5) (c) {$e$}; 
\node[circle,draw,minimum width=0.25in] at (2,-1.5) (d) {$d$}; 
\node[circle,draw,minimum width=0.25in] at (3.5,0)  (e) {$c$};
\path[->,draw,thick] (b) to node[fill=white] {$14$} (a);
\path[->,draw,thick] (a) to node[fill=white] {$12$} (d);
\path[->,draw,thick] (a) to node[fill=white] {$18$} (e);
\path[->,draw,thick] (b) to node[fill=white] {$6$} (c);
\path[->,draw,thick] (d) to node[fill=white] {$8$} (c);
\path[->,draw,thick] (d) to node[fill=white] {$4$} (e);
\path[->,draw,thick] (a) to[pos=.7] node[fill=white] {$10$} (c);
\path[->,draw,thick] (d) to[pos=.7] node[fill=white] {$2$} (b);
\path[->,draw,thick] (e) to[pos=.7] node[fill=white] {$16$} (b);
\path[->,draw,thick] (e) to[pos=.7] node[fill=white] {$20$} (c);
\end{tikzpicture}
\end{center}
\noindent The simple cycles of the above margin graph are highlighted with thickened arrows below:
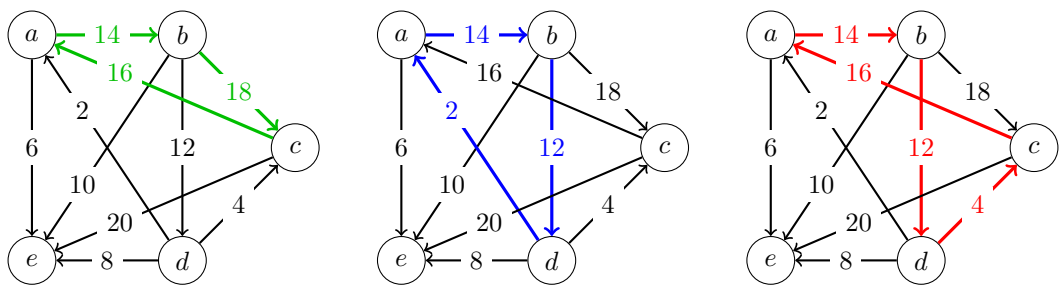
\begin{center}
\begin{tikzpicture}[baseline=(current bounding box.center)]
\node[circle,draw,minimum width=0.25in] at (2,1.5)  (a) {$b$}; 
\node[circle,draw,minimum width=0.25in] at (0,1.5)  (b) {$a$}; 
\node[circle,draw,minimum width=0.25in] at (0,-1.5) (c) {$e$}; 
\node[circle,draw,minimum width=0.25in] at (2,-1.5) (d) {$d$}; 
\node[circle,draw,minimum width=0.25in] at (3.5,0)  (e) {$c$};
\path[->,draw,very thick,medgreen] (b) to node[fill=white] {$14$} (a);
\path[->,draw,thick] (a) to node[fill=white] {$12$} (d);
\path[->,draw,very thick, medgreen] (a) to node[fill=white] {$18$} (e);
\path[->,draw,thick] (b) to node[fill=white] {$6$} (c);
\path[->,draw,thick] (d) to node[fill=white] {$8$} (c);
\path[->,draw,thick] (d) to node[fill=white] {$4$} (e);
\path[->,draw,thick] (a) to[pos=.7] node[fill=white] {$10$} (c);
\path[->,draw,thick] (d) to[pos=.7] node[fill=white] {$2$} (b);
\path[->,draw,very thick, medgreen] (e) to[pos=.7] node[fill=white] {$16$} (b);
\path[->,draw,thick] (e) to[pos=.7] node[fill=white] {$20$} (c);
\end{tikzpicture}\hspace{.3in}\begin{tikzpicture}[baseline=(current bounding box.center)]
\node[circle,draw,minimum width=0.25in] at (2,1.5)  (a) {$b$}; 
\node[circle,draw,minimum width=0.25in] at (0,1.5)  (b) {$a$}; 
\node[circle,draw,minimum width=0.25in] at (0,-1.5) (c) {$e$}; 
\node[circle,draw,minimum width=0.25in] at (2,-1.5) (d) {$d$}; 
\node[circle,draw,minimum width=0.25in] at (3.5,0)  (e) {$c$};
\path[->,draw,very thick, blue] (b) to node[fill=white] {$14$} (a);
\path[->,draw,very thick, blue] (a) to node[fill=white] {$12$} (d);
\path[->,draw,thick] (a) to node[fill=white] {$18$} (e);
\path[->,draw,thick] (b) to node[fill=white] {$6$} (c);
\path[->,draw,thick] (d) to node[fill=white] {$8$} (c);
\path[->,draw,thick] (d) to node[fill=white] {$4$} (e);
\path[->,draw,thick] (a) to[pos=.7] node[fill=white] {$10$} (c);
\path[->,draw,very thick, blue] (d) to[pos=.7] node[fill=white] {$2$} (b);
\path[->,draw,thick] (e) to[pos=.7] node[fill=white] {$16$} (b);
\path[->,draw,thick] (e) to[pos=.7] node[fill=white] {$20$} (c);
\end{tikzpicture}\hspace{.3in}\begin{tikzpicture}[baseline=(current bounding box.center)]
\node[circle,draw,minimum width=0.25in] at (2,1.5)  (a) {$b$}; 
\node[circle,draw,minimum width=0.25in] at (0,1.5)  (b) {$a$}; 
\node[circle,draw,minimum width=0.25in] at (0,-1.5) (c) {$e$}; 
\node[circle,draw,minimum width=0.25in] at (2,-1.5) (d) {$d$}; 
\node[circle,draw,minimum width=0.25in] at (3.5,0)  (e) {$c$};
\path[->,draw,very thick,red] (b) to node[fill=white] {$14$} (a);
\path[->,draw,very thick,red] (a) to node[fill=white] {$12$} (d);
\path[->,draw,thick] (a) to node[fill=white] {$18$} (e);
\path[->,draw,thick] (b) to node[fill=white] {$6$} (c);
\path[->,draw,thick] (d) to node[fill=white] {$8$} (c);
\path[->,draw,very thick,red] (d) to node[fill=white] {$4$} (e);
\path[->,draw,thick] (a) to[pos=.7] node[fill=white] {$10$} (c);
\path[->,draw,thick] (d) to[pos=.7] node[fill=white] {$2$} (b);
\path[->,draw,very thick, red] (e) to[pos=.7] node[fill=white] {$16$} (b);
\path[->,draw,thick] (e) to[pos=.7] node[fill=white] {$20$} (c);
\end{tikzpicture}
\end{center}
\noindent After removing the weakest edge in each simple cycle, what remains is the
defeat graph:
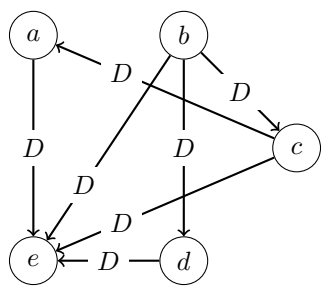
\begin{center}
\begin{tikzpicture}[baseline=(current bounding box.center)]
\node[circle,draw,minimum width=0.25in] at (2,1.5)  (a) {$b$}; 
\node[circle,draw,minimum width=0.25in] at (0,1.5)  (b) {$a$}; 
\node[circle,draw,minimum width=0.25in] at (0,-1.5) (c) {$e$}; 
\node[circle,draw,minimum width=0.25in] at (2,-1.5) (d) {$d$}; 
\node[circle,draw,minimum width=0.25in] at (3.5,0)  (e) {$c$};

\path[->,draw,thick] (a) to node[fill=white] {$D$} (d);
\path[->,draw,thick] (a) to node[fill=white] {$D$} (e);
\path[->,draw,thick] (b) to node[fill=white] {$D$} (c);
\path[->,draw,thick] (d) to node[fill=white] {$D$} (c);

\path[->,draw,thick] (a) to[pos=.7] node[fill=white] {$D$} (c);

\path[->,draw,thick] (e) to[pos=.7] node[fill=white] {$D$} (b);
\path[->,draw,thick] (e) to[pos=.7] node[fill=white] {$D$} (c);
\end{tikzpicture}
\end{center}
The only undefeated candidate is $b$, and thus $b$ is the only Split Cycle winner for the profile.
\end{example}

Finding all cycles may be computationally costly; for
more efficient algorithms, see
\citealt[Footnote 20]{HP2020} or the Preferential Voting Tools library (\href{https://pref-voting.readthedocs.io/}{https://pref-voting.readthedocs.io/}). 

The definition of Split Cycle leads directly to the following observation used
in our later proofs.
\begin{definition}
  Let $\mathcal{M}$ be a margin graph (the margin graph of some profile) and $k
  \in \mathbb{N}$. Then $\mathcal{M}\uph k$ is the result of keeping all and
  only the edges in $\mathcal{M}$ with weight at least $k$.
\end{definition}
\begin{lemma}\label{lem:restrict-to-k}
  For any profile $\mathbf{P}$ and $x, y\in X(\mathbf{P})$, $(x, y) \in
  sc(\mathbf{P})$ iff $\margin_{\mathbf{P}}(x, y) > 0$ and there is no (simple) majority
  path from $y$ to $x$ in $\mathcal{M}(\mathbf{P})\uph \margin_{\mathbf{P}}(x,
  y)$.
\end{lemma}

Finally, since Split Cycle only cares about the margin graph of a profile (formally, if $\mathcal{M}(\mathbf{P})=\mathcal{M}(\mathbf{P}')$, then $sc(\mathbf{P})=sc(\mathbf{P}')$), the only way that the assumption that voters submit \textit{linear} orders can matter is by affecting the types of margin graphs that can arise. This happens with the following parity constraint.

\begin{lemma}\label{ParityLem} If $\mathbf{P}$ is a linear profile, then either all margins between distinct candidates are even or all margins between distinct candidates are odd.
  \end{lemma}
  
Indeed, this parity constraint is the only consequence of assuming linear ballots.

\begin{proposition}[\citealt{Debord1987}] If $\mathcal{M}$ is an asymmetric weighted directed graph with positive integer weights all having the same parity, and all weights are even if there are two vertices with no edge between them (representing a zero margin), then there is a linear profile $\mathbf{P}$ such that $\mathcal{M}=\mathcal{M}(\mathbf{P})$.
\end{proposition}

\noindent In Section \ref{Sec:Ties}, we will drop the assumption of linear ballots and hence the parity constraint on margins.

\section{Axioms on VCCRs}\label{Sec:AxVCCR}

In this section, we present the axioms used in our characterization of the Split Cycle VCCR. 

\subsection{Standard axioms}\label{Sec:StandAxVCCR}
First, we recall what Holliday and Pacuit \citeyearpar{HP2021} consider the ``standard axioms'' in their characterization. 
\begin{definition}\label{Def:StandAxVCCR}
  Let $f$ be a VCCR.
  \begin{enumerate}
  \item $f$ satisfies \textit{Anonymity} if for any profiles $\mathbf{P}$ and $\mathbf{P}'$, if $V(\pfp) = V(\pfp')$ and there is a bijection $\pi$ from $V(\mathbf{P})$ to $V(\mathbf{P}')$ such that for any $i \in V(\mathbf{P}')$, $\mathbf{P}'(i) = \mathbf{P}(\pi(i))$, then  $f(\mathbf{P}) = f(\mathbf{P}')$; and $f$ satisfies \textit{Neutrality} if for any profiles $\mathbf{P}$ and $\mathbf{P}'$, if $V(\mathbf{P}) = V(\mathbf{P}')$, $X(\pfp) = X(\pfp')$, and there is a bijection $\pi$ from $X(\mathbf{P})$ to $X(\mathbf{P}')$ such that for any $i \in V(\mathbf{P})$ and  $x, y \in X(\mathbf{P})$, $(x, y) \in \mathbf{P}(i)$ iff $(\pi(x), \pi(y)) \in \mathbf{P}'(i)$, then for any $x, y \in X(\mathbf{P})$, $(x, y) \in f(\mathbf{P})$ iff $(\pi(x), \pi(y)) \in f(\mathbf{P}')$.
  \item $f$ satisfies \textit{Availability} if for any profile $\mathbf{P}$,
    $\overline{f}(\mathbf{P})$ is nonempty.
  \item $f$ satisfies \textit{Homogeneity} (resp.~\textit{Upward Homogeneity})
    if for any profile $\mathbf{P}$ and $2\mathbf{P}$, where $2\mathbf{P}$ is
    the result of replacing each voter in $\mathbf{P}$ by $2$ copies of that
    voter, $f(\mathbf{P}) = f(2\mathbf{P})$ (resp.~$f(\mathbf{P}) \subseteq
    f(2\mathbf{P})$).
  \item $f$ satisfies \textit{Monotonicity} (resp.~\textit{Monotonicity for two-candidate
      profiles}) if for any profile (resp.~two-candidate profile)
    $\mathbf{P}$ and $\mathbf{P}'$ obtained from $\mathbf{P}$ by moving $x \in X(\mathbf{P})$ up one place in some voter $i$'s ballot in the sense that for all $y, z \not= x$, $(y, z) \in \mathbf{P}'(i)$ iff $(y, z) \in \mathbf{P}(i)$, and $|\{z \in X(\mathbf{P}') \mid (x, z) \in \mathbf{P}'(i)\}| = |\{z \in X(\mathbf{P}) \mid (x, z) \in \mathbf{P}(i)\}| + 1$, we have $(x, y) \in f(\mathbf{P})$
    only if $(x, y) \in f(\mathbf{P}')$.
  \item $f$ satisfies \textit{Neutral Reversal} if for any profile $\mathbf{P}$
    and $\mathbf{P}'$ obtained from $\mathbf{P}$ by adding two voters whose
    ballots are converses of each other, we have $f(\mathbf{P}) = f(\mathbf{P}')$.
  \end{enumerate}
\end{definition}
\noindent Axioms 1-4 are widely satisfied by VCCRs in the literature. While axiom 5 (from \citealt{Saari2003}) is violated by some prominent VCCRs (e.g., the Plurality VCCR according to which $x$ defeats $y$ iff $x$ receives more first-place votes than $y$, or the Pareto VCCR according to which $x$ defeats $y$ iff every voter strictly prefers $x$ to $y$), it is  satisfied by all VCCRs that depend only on the majority margins between candidates (including, e.g., the Borda VCCR, which can be defined in terms of majority margins as in \citealt[p.~28]{Zwicker2016}).

\subsection{Coherent IIA}\label{Sec:CohIIAVCCR}

The key axiom in \citealt{HP2021} is the axiom of Coherent IIA, already discussed informally in Section \ref{Sec:Intro}. First we recall Arrow's \citeyearpar{Arrow1963} IIA, of which Coherent IIA is a weakening. 

\begin{definition} A VCCR $f$ satisfies Independence of Irrelevant Alternatives (IIA)\footnote{Holliday and Pacuit \citeyearpar{HP2021} call this principle \textit{variable-election} IIA (VIIA), since it allows $\mathbf{P}$ and $\mathbf{P}'$ to have different sets of candidates, as opposed to \textit{fixed-election} IIA (FIIA), which requires that $\mathbf{P}$ and $\mathbf{P}'$ have the same set of candidates.} if for any profiles $\mathbf{P}$ and $\mathbf{P}'$, if $x$ defeats $y$ in $\mathbf{P}$ according to $f$, and $\mathbf{P}_{\mid\{x,y\}}=\mathbf{P}_{\mid\{x,y\}}'$, then $x$ defeats $y$ in $\mathbf{P}'$ according to $f$.\end{definition}

Coherent IIA strengthens the assumption of IIA on the relation of $\mathbf{P}$ and $\mathbf{P}'$ so that not only $\mathbf{P}_{\mid\{x,y\}}=\mathbf{P}_{\mid\{x,y\}}'$ but also $\mathbf{P}'$ is related to $\mathbf{P}$ in such a way that $\mathbf{P}'$ cannot be more incoherent with respect to $x,y$, in terms of majority cycles and their strengths, than $\mathbf{P}$ is with respect to $x,y$.

\begin{definition}
  \label{def:lc-accessibility}
  For any two profiles $\mathbf{P}$ and $\mathbf{P}'$ with $x,y\in
  X(\mathbf{P})\cap X(\mathbf{P}')$,
  let \[\mathbf{P}\rightsquigarrow_{x,y}\mathbf{P}'\] if $\mathbf{P}_{|\{x,y\}}
  = \mathbf{P}'_{|\{x,y\}}$ and $\mathcal{M}(\mathbf{P}')$ can be obtained from
  $\mathcal{M}(\mathbf{P})$ by deleting zero or more candidates other than $x$
  and $y$ and deleting or reducing the margins on zero or more edges not connecting $x$ and $y$.
\end{definition}

Note that trivially $\mathbf{P} \rightsquigarrow_{x, y} \mathbf{P}'$ only if
$V(\mathbf{P}') = V(\mathbf{P})$ and $X(\mathbf{P}') \subseteq X(\mathbf{P})$;
that is, the profiles must have the same set of voters and $\mathbf{P}'$ cannot have
additional candidates.
\begin{example} Consider the following profile $\mathbf{P}$ with three
voters $i, j, k$ and four candidates $x, y, z, u$, whose margin graph is shown on the right:
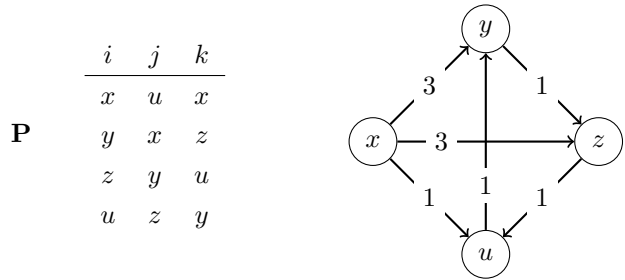
\begin{center}
$\mathbf{P}$\qquad
\begin{tabular}[]{ccc}
   $i$ & $j$ & $k$ \\
  \hline 
  $x$ & $u$ & $x$ \\
  $y$ & $x$ & $z$ \\
  $z$ & $y$ & $u$ \\
  $u$ & $z$ & $y$
\end{tabular}
\hspace{4em}
  \begin{tikzpicture}[baseline=(current bounding box.center)]
    \node[circle,draw,minimum width=0.25in] at (0,0) (x) {$x$}; 
    \node[circle,draw,minimum width=0.25in] at (3,0) (z) {$z$}; 
    \node[circle,draw,minimum width=0.25in] at (1.5,1.5) (y) {$y$}; 
    \node[circle,draw,minimum width=0.25in] at (1.5,-1.5) (u) {$u$}; 
    \path[->,draw,thick] (x) to node[fill=white] {$3$} (y);
    \path[->,draw,thick] (x) to node[near start, fill=white] {$3$} (z);
    \path[->,draw,thick] (x) to node[fill=white] {$1$} (u);
    \path[->,draw,thick] (y) to node[fill=white] {$1$} (z);
    \path[->,draw,thick] (u) to node[near start, fill=white] {$1$} (y);
    \path[->,draw,thick] (z) to node[fill=white] {$1$} (u);
  \end{tikzpicture}
\end{center}
If $\mathbf{Q}$ is the restriction of $\mathbf{P}$ to just
the candidates $x$ and $y$, which is the profile where $i$, $j$, and $k$ vote
unanimously that $x$ is better than $y$, then trivially $\mathbf{P}
\rightsquigarrow_{x, y} \mathbf{Q}$. That is, 
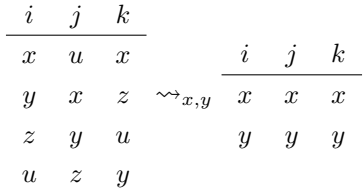
\begin{center}
\begin{tabular}[]{ccc}
   $i$ & $j$ & $k$ \\
  \hline 
  $x$ & $u$ & $x$ \\
  $y$ & $x$ & $z$ \\
  $z$ & $y$ & $u$ \\
  $u$ & $z$ & $y$
\end{tabular}
$\rightsquigarrow_{x, y}$
\begin{tabular}[]{ccc}
   $i$ & $j$ & $k$ \\
  \hline 
  $x$ & $x$ & $x$ \\
  $y$ & $y$ & $y$ \\
\end{tabular}
\end{center}
If we consider only changing the relative position of one pair of candidates for
one voter, then the only way to produce a different profile $\mathbf{Q}$ such
that $\mathbf{P} \rightsquigarrow_{x, y} \mathbf{Q}$ is to switch $x$ and $z$ in
the ballot of $k$. That is, we have
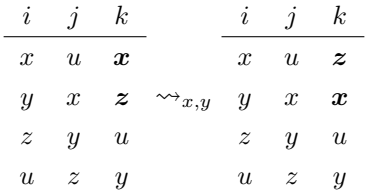
\begin{center} 
  \begin{tabular}[]{ccc}
   $i$ & $j$ & $k$ \\
  \hline 
  $x$ & $u$ & $\boldsymbol{x}$ \\
  $y$ & $x$ & $\boldsymbol{z}$ \\
  $z$ & $y$ & $u$ \\
  $u$ & $z$ & $y$
\end{tabular}
$\rightsquigarrow_{x, y}$
\begin{tabular}[]{ccc}
   $i$ & $j$ & $k$ \\
  \hline 
  $x$ & $u$ & $\boldsymbol{z}$ \\
  $y$ & $x$ & $\boldsymbol{x}$ \\
  $z$ & $y$ & $u$ \\
  $u$ & $z$ & $y$
\end{tabular}
\end{center}
with their corresponding margin graphs:
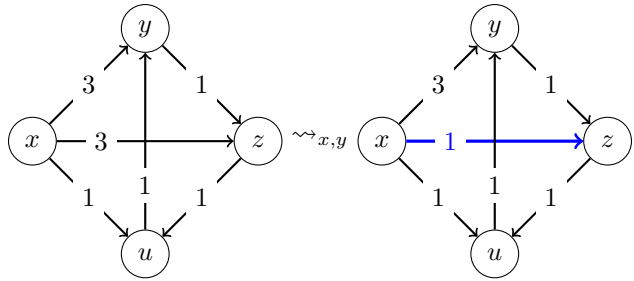
\begin{center}
  \begin{tikzpicture}[baseline]
    \node[circle,draw,minimum width=0.25in] at (0,0) (x) {$x$}; 
    \node[circle,draw,minimum width=0.25in] at (3,0) (z) {$z$}; 
    \node[circle,draw,minimum width=0.25in] at (1.5,1.5) (y) {$y$}; 
    \node[circle,draw,minimum width=0.25in] at (1.5,-1.5) (u) {$u$}; 
    \path[->,draw,thick] (x) to node[fill=white] {$3$} (y);
    \path[->,draw,thick] (x) to node[near start, fill=white] {$3$} (z);
    \path[->,draw,thick] (x) to node[fill=white] {$1$} (u);
    \path[->,draw,thick] (y) to node[fill=white] {$1$} (z);
    \path[->,draw,thick] (u) to node[near start, fill=white] {$1$} (y);
    \path[->,draw,thick] (z) to node[fill=white] {$1$} (u);
  \end{tikzpicture}
  $\rightsquigarrow_{x, y}$
  \begin{tikzpicture}[baseline]
    \node[circle,draw,minimum width=0.25in] at (0,0) (x) {$x$}; 
    \node[circle,draw,minimum width=0.25in] at (3,0) (z) {$z$}; 
    \node[circle,draw,minimum width=0.25in] at (1.5,1.5) (y) {$y$}; 
    \node[circle,draw,minimum width=0.25in] at (1.5,-1.5) (u) {$u$}; 
    \path[->,draw,thick] (x) to node[fill=white] {$3$} (y);
    \path[->,draw,very thick, blue] (x) to node[near start, fill=white] {$1$} (z);
    \path[->,draw,thick] (x) to node[fill=white] {$1$} (u);
    \path[->,draw,thick] (y) to node[fill=white] {$1$} (z);
    \path[->,draw,thick] (u) to node[near start, fill=white] {$1$} (y);
    \path[->,draw,thick] (z) to node[fill=white] {$1$} (u);
  \end{tikzpicture}
\end{center}
Note that from left to right, no new edges are created, and the only change is
that the margin of the edge from $x$ to $z$ is reduced from $3$ to $1$. The
reason this swap of $x$ and $z$ in the ballot of $k$ is the only allowed
swap if we want a profile $\mathbf{Q}$ such that  $\mathbf{P}\rightsquigarrow_{x, y}\mathbf{Q}$ is
that (1) obviously we cannot swap $x$ and $y$, and (2) for any other pair of
candidates, the margin between them in the original margin graph is $1$ so that
a swap would flip the edge. The edges are directed, so flipping edges means
creating new majority edges, which is not allowed in the definition of
$\rightsquigarrow_{x, y}$. 
\end{example}

We now define Coherent IIA by replacing $\mathbf{P}_{\mid\{x,y\}}=\mathbf{P}_{\mid\{x,y\}}'$ in the definition of IIA with  $\mathbf{P}\rightsquigarrow_{x, y}\mathbf{P}'$.

\begin{definition} A VCCR $f$ satisfies \textit{Coherent IIA} if for any profiles $\mathbf{P}$ and $\mathbf{P}'$, if $x$ defeats $y$ in $\mathbf{P}$ according to $f$,  and  $\mathbf{P}\rightsquigarrow_{x,y}\mathbf{P}'$, then $x$ defeats $y$ in $\mathbf{P}'$ according to $f$.\end{definition}

\noindent Note a simple consequence of Coherent IIA, known as Weak IIA (cf.~\citealt{Baigent1987}): for any profiles $\mathbf{P}$ and $\mathbf{P}'$, if $x$ defeats $y$ in $\mathbf{P}$ according to $f$ and $\mathbf{P}_{\mid\{x,y\}}=\mathbf{P}_{\mid\{x,y\}}'$, then it is not the case that $y$ defeats $x$ in $\mathbf{P}'$ according to $f$. To see that Weak IIA follows from Coherent IIA, note that according to any VCCR satisfying Coherent IIA, if $x$ defeats $y$ in $\mathbf{P}$, $\mathbf{P}_{|\{x, y\}} = \mathbf{P}'_{|\{x, y\}} = \mathbf{Q}$, and $y$ defeats $x$ in $\mathbf{P}'$, then by Coherent IIA and the observation that $\mathbf{P} \rightsquigarrow_{x, y} \mathbf{P}_{|\{x, y\}} = \mathbf{Q}$ and $\mathbf{P}' \rightsquigarrow_{y, x} \mathbf{P}'_{|\{x, y\}} = \mathbf{Q}$, we have that both $x$ defeats $y$ and $y$ defeats $x$ in $\mathbf{Q}$, contradicting the fact that the defeat relation produced by the VCCR $f$ must be asymmetric. For extensive discussion of Coherent IIA and its consequences, see \citealt[Section  4.3]{HP2021}. In this paper we focus instead on the two new axioms used in our main result.

An axiom closely related to Coherent IIA that is also discussed in \citealt{HP2021} is the following axiom of Majority Defeat.
\begin{definition}\label{MajDefeat}
  A VCCR $f$ satisfies Majority Defeat iff for any profile $\mathbf{P}$ and $x, y \in X(\mathbf{P})$, if $(x, y) \in f(\mathbf{P})$, then $\margin_{\mathbf{P}}(x,y) > 0$, i.e., $x$ defeats $y$ only if $x$ is majority preferred to $y$.
\end{definition}
\noindent Majority Defeat will be used a few times later. 

\subsection{Coherent Defeat}\label{Sec:CohDefVCCR}

We come now to the first of our two new axioms.

\begin{definition}
  A VCCR $f$ satisfies \textit{Coherent Defeat} if for any profile $\mathbf{P}$
  and $x, y \in X(\mathbf{P})$, if ${\margin_{\mathbf{P}}(x, y) > 0}$ and there is
  no majority cycle containing $x$ and $y$ in $\mathbf{P}$ (or equivalently,
  there is no majority path from $y$ to~$x$), then $(x, y) \in f(\mathbf{P})$.
\end{definition}
\noindent The key idea behind Coherent Defeat is simple: when there is no incoherence due to majority cycles, majority preference is sufficient for defeat. In other words, majority cycles are the only reason we deviate from majority preference for deciding defeat. One caveat is that we understand incoherence locally: when deciding whether $x$ defeats $y$, only majority cycles involving $x$ and $y$ matter; thus, regardless of whether there are majority cycles involving other candidates, if there are no majority cycles involving $x$ and $y$, then a majority preference for $x$ over $y$ is sufficient for  $x$ to defeat $y$. As mentioned in Section \ref{Sec:Intro}, a number of VCCRs (such as Ranked Pairs, Beat Path, and GOCHA, as proved below) together with Split Cycle share this core commitment and thus can all be seen as ways of resolving incoherence due to relevant majority cycles. To put the point in terms of the Advantage-Standard Model of \citealt{HK2020b}, according to which $x$ defeats $y$ once the \textit{advantage} of $x$ over $y$ (which depends only on how voters ranks $x$ versus $y$) exceeds the \textit{standard} for $x$ to defeat $y$ (which may depend on other information in the profile, including how voters rank $x$ against non-$y$ candidates and $y$ against non-$x$ candidates), Coherent Defeat follows assuming the advantage and standard functions satisfy the following weak constraints:
\begin{itemize}
  \item the advantage of $x$ over $y$ is greater than $0$ if $x$ is majority preferred to $y$;
  \item the standard for $x$ to defeat $y$ is $0$ if there is no majority cycle involving $x$ and $y$.
\end{itemize}
The second constraint is clearly necessary if we take the standard for $x$ to defeat $y$ to measure in some way local incoherence due to majority cycles involving $x$ and $y$.

Coherent Defeat can also be viewed as a principle of ``unchallenged defeat'' that follows from three principles:
\begin{itemize}
\item for $x$ to defeat $y$, it is sufficient to find one reason for $x$ to
  defeat $y$ and  make sure that no reasons for $x$ to defeat $y$ can be
  challenged;
\item a majority preference for $x$ over $y$ is a reason for $x$ to defeat $y$;
\item any challenge to any reason for $x$ to defeat $y$ must be based on a
  majority preference path from $y$ to $x$.
\end{itemize}
For a more concrete model, we can view the margin graph as providing \textit{arguments} for and against propositions of the form ``$x$ defeats $y$'' and their negations.  There are two types of arguments. First, there is an argument for ``$x$ defeats $y$'' whenever there is a majority preference for $x$ over $y$. Second, when there is a majority path from $x$ to $y$,  there is an argument for ``$y$ does not defeat $x$'', since for each link $a, b$ in the path we have an argument for  ``$a$ defeats $b$'', but the defeat relation must be acyclic, so taken together the steps along a majority path constitute an argument for ``$y$ does not defeat $x$''. It should be noted here that an argument based on a majority path from $x$ to $y$ to the conclusion ``$x$ defeats $y$'' is not necessarily a good argument, since defeat relations are not obviously transitive. Now if we grant that arguments for and against candidates defeating each other can only be generated in the above two ways, then once we have an $x$ who is majority preferred to $y$ and there is no majority path from $y$ to $x$, we have an argument for ``$x$ defeats $y$'' but no counterargument to the contrary. Thus, we should accept that $x$ defeats $y$. Similar ideas of treating majority preferences and paths as arguments appear in \citealt{Heitzig2002}, and our Coherent Defeat is essentially his Immunity to Binary Arguments ($\mathrm{Im}_A$) applied to VCCRs, with $A$ as the majority preference relation.

\begin{proposition}
  Ranked Pairs, Beat Path, and GOCHA as VCCRs satisfy Coherent Defeat.
\end{proposition}
\begin{proof}
  Ranked Pairs is standardly defined as a VSCC. We first define Ranked Pairs as a VSCC, and then consider a natural VCCR version for this proposition. For any profile $\pfp$, let $Pairs(\pfp) = \{(x, y) \in X(\pfp)^2 \mid x \not= y \text{ and } \margin_{\mathbf{P}}(x, y) \ge 0\}$. Say a \emph{tie-breaker for $\pfp$} is a linear order $L \in \mathcal{L}(Pairs(\pfp))$. Then we define the $L$-ordering $\succ_{\pfp, L}$ of $Pairs(\pfp)$ as follows: $(x, y) \succ_{\pfp, L} (x', y')$ iff $\margin_\pfp(x, y) > \margin_\pfp(x', y')$ or $\margin_\pfp(x, y) = \margin_\pfp(x', y')$ and $(x, y) L (x', y')$. Then define $rp(\pfp, L)$, a linear order on $X(\pfp)$, as follows: 
  \begin{itemize}
    \item List $Pairs(\pfp)$ according to $\succ_{\pfp, L}$ as $(x_1, y_1) \succ_{\pfp, L} (x_2, y_2) \succ_{\pfp, L} \dots \succ_{\pfp, L} (x_m, y_m)$.
    \item Let $D_0 = \varnothing$, and for each $i = 1, \dots, m$, if $D_{i-1} \cup \{(x_i, y_i)\}$ is acyclic, let $D_i = D_{i-1} \cup \{(x_i, y_i)\}$, and otherwise let $D_i = D_{i-1} \cup \{(y_i, x_i)\}$.
    \item Let $rp(\pfp, L) = D_m$.
  \end{itemize}
  Intuitively, $rp(\pfp, L)$ is a defeat relation relative to the tie-breaker $L$. For any tie-breaker $L$, $rp(\pfp, L)$ is a linear ordering of $X(\pfp)$. The Ranked Pairs VSCC is defined by letting $RP(\pfp)$  be the set of all $x \in X(\pfp)$ such that there is a tie-breaker $L \in \mathcal{L}(Pairs(X(\pfp)))$ and $x$ is the top element of $rp(\pfp, L)$. We define the VCCR version of Ranked Pairs by letting $rp(\pfp) = \bigcap\{rp(\pfp, L) \mid L \in \mathcal{L}(Pairs(X(\pfp)))\}$. That is, $x$ defeats $y$ according to the Ranked Pairs iff for all tie-breakers $L$, $x$ defeats $y$ relative to $L$ according to Ranked Pairs.

  Now we show that $rp$ satisfies Coherent Defeat. Let $\pfp$ be a profile and $x, y \in X(\pfp)$ such that $\margin_\pfp(x, y) > 0$. Assume there is no majority path from $y$ to $x$. Let $L$ be any tie-breaker. Since $\margin_\pfp(x, y) > 0$, for any $(x', y') \succ_{\pfp, L} (x, y)$, $\margin_\pfp(x', y') > 0$. Using the notation above, let $k$ be the position of the pair $(x, y)$ in the ordering $\succ_{\pfp, L}$. Then for any $i = 1, \dots, k$, $\margin_{\pfp}(x_i, y_i) > 0$, and by an easy induction up to $k$, $D_i$ is acyclic and the transitive closure of $D_i$ must also be a subset of the transitive closure of $\{(x_1, y_1), \dots, (x_i, y_i)\}$. Since there is no majority path from $y$ to $x$, $(y, x)$ is not in the transitive closure of $\{(x_1, y_1), \dots, (x_{k-1}, y_{k-1})\}$. This means $D_{k-1} \cup \{(x, y)\}$ is acyclic, and thus $(x, y) \in D_k$ and $(x, y) \in rp(\pfp, L)$. Since $L$ was arbitrary, $(x, y) \in rp(\pfp)$.

For Beat Path, recall that for any majority path $\rho$, the strength $\strength(\rho)$ of $\rho$ is the minimal weight in the majority edges of $\rho$. For any $x, y \in X(\pfp)$, let 
\[
\strength_{\pfp}(x, y) = \max\{\strength_\pfp(\rho) \mid \rho \text{ is a majority path from } x \text{ to } y \text{ in } \pfp\}
\]
where the max of the empty set is stipulated to be $0$. Then Beat Path as a VCCR is defined by setting $(x, y) \in bp(\pfp)$ iff $\strength_\pfp(x, y) > \strength_\pfp(y, x)$. Now if $\margin_{\pfp}(x, y) > 0$ while there is no majority path from $y$ to $x$, then $\strength_\pfp(x, y) \ge \margin_{\pfp}(x, y) > 0 = \strength_\pfp(y, x)$. Thus, $(x, y) \in bp(\pfp)$.

  GOCHA is also usually defined as a VSCC, but the underlying idea also defines a VCCR; it can be seen as a purely qualitative version of Beat Path, not taking the weights into consideration. For any profile $\pfp$, the GOCHA VCCR $gocha$ is defined by setting $(x, y) \in gocha(\pfp)$ iff there is a majority path from $x$ to $y$ but not from $y$ to $x$ in $\pfp$.\footnote{This definition is due to Corollary 6.2.2 in \cite{Schwartz1986}. The VSCC $GOCHA$ is simply $\overline{gocha}$.} It is then immediate  that $gocha$ satisfies Coherent Defeat.\end{proof}

\subsection{Positive Involvement in Defeat}\label{Sec:PosInvVCCR}
Our second new axiom for VCCRs is based on the standard axiom of Positive Involvement for VSCCs (\citealt{Saari1995}, \citealt{Perez2001}, \citealt{HP2021PI}).

\begin{definition} A VSCC $F$ satisfies \textit{Positive Involvement} if for any
  profile $\mathbf{P}$, if $y\in F(\mathbf{P})$ and $\mathbf{P}'$ is obtained
  from $\mathbf{P}$ by adding one new voter who ranks $y$ in first place, then
  $y\in F(\mathbf{P}')$.
\end{definition}
\noindent Thus, Positive Involvement says that a candidate $y$'s winning should be preserved under the
addition of a voter who gives $y$  maximum support by ranking $y$ at the top
of her ballot. As Perez \citeyearpar[p.~605]{Perez2001} remarks, it can be seen as ``the minimum to
require concerning the coherence in the winning set when new voters are added.''

\begin{remark} Another well-known axiom for VSCCs concerning the addition of voters to a profile is the Reinforcement axiom: if $\mathbf{P}$ and $\mathbf{P}'$ are profiles with the same set of candidates but disjoint sets of voters, then $F(\mathbf{P})\cap F(\mathbf{P}')\neq\varnothing$ implies  $F(\mathbf{P}+\mathbf{P}')=F(\mathbf{P})\cap F(\mathbf{P}')$, where $\mathbf{P}+\mathbf{P}'$ is the profile combining the voters from $\mathbf{P}$ and $\mathbf{P}'$. Positive Involvement clearly follows from Reinforcement\footnote{Even in the weaker version requiring only  $F(\mathbf{P})\cap F(\mathbf{P}')\subseteq F(\mathbf{P}+\mathbf{P}')$. Also note that the Homogeneity axiom of Definition~\ref{Def:StandAxVCCR} follows from Reinforcement.} together with the axiom of Faithfulness (\citealt{Young1974}): if $\mathbf{P}$ has only one voter who ranks  $x$ uniquely first, then $F(\mathbf{P})=\{x\}$. However, Reinforcement is inconsistent with Condorcet Consistency (if some candidate $x$ has positive margins over all other candidates in $\mathbf{P}$, then $F(\mathbf{P})=\{x\}$) (see \citealt[Proposition 2.5]{Zwicker2016}) and hence is not satisfied by Split Cycle, which is Condorcet consistent. It is sometimes said that the Kemeny VSCC (\citealt{kemeny1959}) satisfies Reinforcement and Consistency, but this is not the case; it is only when Kemeny is regarded as a \textit{social preference function} (SPF), which assigns to each profile a \textit{set} of binary relations, that the Kemeny SPF can be regarded as satisfying Reinforcement for SPFs (\citealt{Young1978}), which is a substantially weaker axiom than Reinforcement for VSCCs (see \citealt[p.~45]{Zwicker2016}). For doubts about the normative plausibility of Reinforcement for VSCCs, see \citealt[Remark 4.26]{HP2020}.\end{remark}

To generalize Positive Involvement to VCCRs, we need to ask what is the ``minimum to require
concerning the coherence'' in the defeat relation when new voters are added. First, there is  a trivial way to generalize Positive Involvement to VCCRs.
\begin{definition}
  An acyclic VCCR $f$ satisfies \textit{Positive Involvement} if the VSCC $\overline{f}$
  (recall Lemma \ref{VCCRtoVSCC}) satisfies Positive Involvement: for any profile $\mathbf{P}$ and $y\in X(\mathbf{P})$, if  $y$ is
  undefeated in $\mathbf{P}$ according to $f$, and 
  $\mathbf{P}'$ is obtained from $\mathbf{P}$ by adding one new voter who ranks $y$
  in first place, then $y$ is still undefeated in $\mathbf{P}'$ according to~$f$.
\end{definition}
The problem with this generalization is not that it is unreasonable for a VCCR to satisfy it but rather that we find it reasonable to ask for more. Given a VCCR $f$, we have information not only about who is defeated or undefeated but also about who is defeated \textit{by whom}. Thus, there is  the question of how the defeat relation between two candidates $x$ and $y$ should react to the addition of a voter. A natural answer is this: adding a voter who ranks $y$ above $x$ should not lead to $y$'s defeat by $x$. 
\begin{definition}
  A VCCR $f$ satisfies \textit{Positive Involvement in Defeat} if for any
  profile $\mathbf{P}$ and $x,y\in X(\mathbf{P})$, if $y$ is not defeated by $x$
  in $\mathbf{P}$ according to $f$, and $\mathbf{P}'$ is obtained from
  $\mathbf{P}$ by adding one new voter  who ranks $y$ above $x$, then $y$ is not
  defeated by $x$ in $\mathbf{P}'$ according to $f$.
\end{definition}
\noindent This is of course not the only possible answer. For example, we could require not only that $y$ is ranked above $x$ in the new ballot but also that $y$ is ranked at the top. 
\begin{definition}
  A VCCR $f$ satisfies \textit{First-place Involvement in Defeat} if for any
  profile $\mathbf{P}$ and $x,y \in X(\mathbf{P})$, if $y$ is not defeated by
  $x$ in $\mathbf{P}$ according to $f$, and $\mathbf{P}'$ is obtained from
  $\mathbf{P}$ by adding one new voter who ranks $y$ \emph{in first place},
  then $y$ is not defeated by $x$ in $\mathbf{P}'$ according to $f$.
\end{definition}
\noindent Since the requirement on $x$ and $y$ is stronger in First-place Involvement in Defeat, it is entailed by Positive Involvement in Defeat: any VCCR satisfying Positive Involvement in Defeat also satisfies First-place Involvement in Defeat. Moreover, it is easy to see that First-place Involvement in Defeat entails Positive Involvement for VCCRs.

It turns out, as we show in Section \ref{sec:necessity}, that the stronger axiom of Positive Involvement in Defeat is necessary for our characterization of Split Cycle (compared to First-place Involvement in Defeat and hence Positive Involvement), fixing the other axioms we use. However, we do not see this as a problem for Split Cycle as an appealing VCCR, since we find the intuitive appeal of Positive Involvement in Defeat rather clear and perhaps even clearer than First-place Involvement in Defeat. If the only change to a profile in which $y$ is already undefeated by $x$ is the addition of a ballot ranking $y$ above $x$, thereby lending support to $y$ against $x$, then there is no reason to now say that $x$ defeats $y$.

\begin{proposition}\label{SCPID}
    The Split Cycle VCCR satisfies Positive Involvement in Defeat.
\end{proposition}
\begin{proof}
  Let $\mathbf{P}$ be a profile and $x, y \in X(\mathbf{P})$.  Suppose $y$ is not defeated by $x$. Then either (1) $\margin_{\mathbf{P}}(x, y) \le 1$ or (2) $\margin_{\mathbf{P}}(x, y) \ge 2$ and there is a majority path $\rho$ from $y$ to $x$ such that $\margin_{\mathbf{P}}(x, y) \le \strength_{\mathbf{P}}(\rho)$. 

  Now let $\pfp'$ be the result of adding to $\pfp$ one ballot where $y$ is above $x$. In case (1), $\margin_{\pfp'}(x, y) = \margin_{\pfp}(x, y) - 1 \le 0$, so $x$ still does not defeat $y$ in $\pfp'$. In case (2), note that $\strength_{\pfp}(\rho) \ge 2$, so every edge in $\rho$ has margin at least $2$, and  from $\pfp$ to $\pfp'$ only one ballot is added, so the weakest edges in $\rho$ have their margins decrease by at most $1$ from $\pfp$ to $\pfp'$. Thus, $\rho$ is still a majority path in $\pfp'$ from $y$ to $x$ and $\strength_{\pfp'}(\rho) \ge \strength_{\pfp}(\rho) - 1 \ge \margin_{\pfp}(x, y) - 1 = \margin_{\pfp'}(x, y)$, so $x$ does not defeat $y$ in $\pfp'$.
\end{proof}

As another example of working with {Positive Involvement in Defeat}, we mention  that Weighted Covering (\citealt{Dutta1999}, \citealt{Fernandez2018}) satisfies it. Weighted Covering will also be used in Section \ref{sec:necessity}. 
\begin{definition}\label{WCdef}
  The Weighted Covering VCCR, denoted $wc$, is defined as follows: for any
  profile $\mathbf{P}$ and $x, y \in X(\mathbf{P})$, $(x, y) \in wc(\mathbf{P})$
  iff $\margin_{\mathbf{P}}(x, y) > 0$ and for all $z \in X(\mathbf{P})$,
  $\margin_{\mathbf{P}}(x, z) \ge \margin_{\mathbf{P}}(y, z)$. 
\end{definition}

\begin{proposition}\label{WCPID}
  Weighted Covering satisfies Positive Involvement in Defeat.
\end{proposition}
\begin{proof}
  Let $\mathbf{P}$ be a profile and $x, y \in X(\mathbf{P})$.  Suppose $y$ is
  not defeated by $x$. Then either (1) $\margin_{\mathbf{P}}(x, y) \le 0$ or
  (2) there is $z \in X(\mathbf{P})$ such that $\margin_{\mathbf{P}}(x, z) <
  \margin_{\mathbf{P}}(y, z)$. 

  Now let $\mathbf{P}'$ be the result of adding to $\mathbf{P}$ one ballot where
  $y$ is above $x$. In case (1), $\margin_{\mathbf{P}'}(x, y) < 0$ since the new
  ballot decreases the margin from $x$ to $y$ by $1$, and thus $y$ is still
  undefeated by $x$. In case (2), since in the new ballot, $y \succ x$, there are
  only three possible places for $z$ relative to $y$ and $x$: $z \succ y \succ x$, $y \succ z
  \succ x$, and $y \succ x \succ z$. In the first and the third case, $\Delta_{x, z} =
  \margin_{\mathbf{P}'}(x, z) - \margin_{\mathbf{P}}(x, z)$ and $\Delta_{y, z} =
  \margin_{\mathbf{P}'}(y, z) - \margin_{\mathbf{P}}(y, z)$ are the same ($-1$
  and $1$ in the first and third cases, respectively), and in the second case, $\Delta_{x, z} = -1 <
  \Delta_{y, z} = 1$. Thus, in all three cases, since $\Delta_{x, z} \le
  \Delta_{y, z}$, we still have $\margin_{\mathbf{P}'}(x, z) <
  \margin_{\mathbf{P}'}(y, z)$, and $y$ is still not defeated by $x$ in
  $\mathbf{P}'$.
\end{proof}
On the other hand, the many non-weighted versions of Covering (see
\citealt{Gillies1959}, \citealt{Fishburn1977}, \citealt{Miller1980}, \citealt{Duggan2013}) do not
satisfy Positive Involvement in Defeat. For example, for any profile
$\mathbf{P}$ and $x, y \in X(\mathbf{P})$, we say that $x$ defeats $y$ according
to the Right Covering VCCR when $\margin_{\mathbf{P}}(x, y) > 0$ and for all $z
\in X(\mathbf{P})$, if $\margin_{\mathbf{P}}(y, z) > 0$ then
$\margin_{\mathbf{P}}(x, z) > 0$. To see that the Right Covering VCCR
fails Positive Involvement in Defeat, consider any profile $\mathbf{P}$ whose
margin graph is
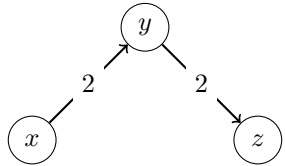
\begin{center}
  \begin{tikzpicture}[baseline]
    \node[circle,draw,minimum width=0.25in] at (0,0) (x) {$x$}; 
    \node[circle,draw,minimum width=0.25in] at (3,0) (z) {$z$}; 
    \node[circle,draw,minimum width=0.25in] at (1.5,1.5) (y) {$y$}; 
    \path[->,draw,thick] (x) to node[fill=white] {$2$} (y);
    \path[->,draw,thick] (y) to node[fill=white] {$2$} (z);
  \end{tikzpicture}
\end{center}
By definition, $x$ does not defeat $y$ in $\mathbf{P}$ according to Right
Covering because of $z$. However, once we add a ballot $y \succ x \succ z$ to
$\mathbf{P}$, the margin graph becomes 
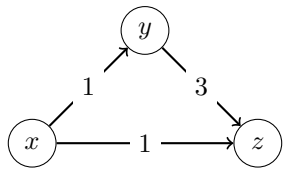
\begin{center}
  \begin{tikzpicture}[baseline]
    \node[circle,draw,minimum width=0.25in] at (0,0) (x) {$x$}; 
    \node[circle,draw,minimum width=0.25in] at (3,0) (z) {$z$}; 
    \node[circle,draw,minimum width=0.25in] at (1.5,1.5) (y) {$y$}; 
    \path[->,draw,thick] (x) to node[fill=white] {$1$} (y);
    \path[->,draw,thick] (x) to node[fill=white] {$1$} (z);
    \path[->,draw,thick] (y) to node[fill=white] {$3$} (z);
  \end{tikzpicture}
\end{center}
and now $x$ defeats $y$ according to Right Covering, since $x$ now also has a
positive margin over $z$.

For another example of a VCCR violating Positive Involvement in Defeat, consider the following VCCR that defeat-rationalizes the Minimax VSCC (\citealt{Simpson1969}, \citealt{Kramer1977}). 
  For any profile $\pfp$ and $x \in X(\pfp)$, define \[\weakness_\pfp(x) = \max\{\margin_\pfp(z, x) \mid z \in X(\pfp)\}.\] Then the VCCR $mm$ is defined by setting $(x, y) \in mm(\pfp)$ iff $\weakness_\pfp(x) < \weakness_\pfp(y)$. Now consider any profile $\pfp$ whose margin graph is shown on the left below, and suppose $\pfp'$ is obtained by adding to $\pfp$ a ballot $b \succ y \succ x \succ a$, resulting in the margin graph on the right:
\begin{center}
  \begin{tikzpicture}[baseline]
    \node[circle,draw,minimum width=0.25in] at (0,0) (x) {$x$}; 
    \node[circle,draw,minimum width=0.25in] at (3,0) (b) {$b$}; 
    \node[circle,draw,minimum width=0.25in] at (1.5,1.5) (a) {$a$}; 
    \node[circle,draw,minimum width=0.25in] at (1.5,-1.5) (y) {$y$}; 
    \path[->,draw,thick] (a) to node[fill=white] {$2$} (x);
    \path[->,draw,thick] (b) to node[fill=white] {$2$} (y);
  \end{tikzpicture}
  \hspace{4em}
  \begin{tikzpicture}[baseline]
    \node[circle,draw,minimum width=0.25in] at (0,0) (x) {$x$}; 
    \node[circle,draw,minimum width=0.25in] at (3,0) (b) {$b$}; 
    \node[circle,draw,minimum width=0.25in] at (1.5,1.5) (a) {$a$}; 
    \node[circle,draw,minimum width=0.25in] at (1.5,-1.5) (y) {$y$}; 
    \path[->,draw,thick] (a) to node[fill=white] {$1$} (x);
    \path[->,draw,thick] (y) to node[fill=white] {$1$} (x);
    \path[->,draw,thick] (b) to node[fill=white] {$1$} (a);
    \path[->,draw,thick] (b) to node[fill=white, near start] {$1$} (x);
    \path[->,draw,thick] (y) to node[fill=white, near start] {$1$} (a);
    \path[->,draw,thick] (b) to node[fill=white] {$3$} (y);
  \end{tikzpicture}
\end{center}
It is clear that $\weakness_{\pfp}(x) = 2 = \weakness_{\pfp}(y)$, but $\weakness_{\pfp'}(x) = 1 < 3 = \weakness_{\pfp'}(y)$. In other words, $(x, y) \not\in mm(\pfp)$, but $(x, y) \in mm(\pfp')$, while the only change from $\pfp$ to $\pfp'$ is that a new voter ranking $y$ above $x$ joined the election. What is notable about this failure of Positive Involvement in Defeat is that Minimax as a VSCC does satisfy  the axiom of Tolerant Positive Involvement that we will introduce in Section~\ref{Sec:PosInvVSCC} for VSCCs that corresponds to Positive Involvement in Defeat.

One may naturally wonder why we focus solely on Positive Involvement, given that
there is also the axiom of Negative Involvement in the literature (again see
\citealt{Perez2001}).
\begin{definition}
  A VSCC $F$ satisfies \textit{Negative Involvement} if for any profile
  $\mathbf{P}$ and $y\in X(\mathbf{P})$, if $y\not\in F(\mathbf{P})$ and $\mathbf{P}'$ is obtained by
  adding one new voter who ranks $y$ in last place, then $y \not\in F(\mathbf{P}')$.

  An acyclic VCCR $f$ satisfies \textit{Negative Involvement} if $\overline{f}$ satisfies
  Negative Involvement. We say that a VCCR $f$ satisfies \textit{Negative Involvement in Defeat} if
  for any profile $\mathbf{P}$ and $x,y\in X(\mathbf{P})$, if $(x, y) \in
  f(\mathbf{P})$, and $\mathbf{P}'$ is obtained from $\mathbf{P}$ by adding one
  new voter who ranks $x$ above $y$, then $(x, y) \in f(\mathbf{P}')$.
\end{definition}
The reason we can focus solely on Positive Involvement is given by the following proposition.
\begin{proposition}
  An acyclic VCCR (resp.~VCCR) $f$ satisfying Neutral Reversal satisfies Positive Involvement (resp.
  Positive Involvement in Defeat) iff it satisfies Negative Involvement (resp.
  Negative Involvement in Defeat).
\end{proposition}
\begin{proof}
  Let $f$ be an acyclic VCCR $f$ satisfying Neutral Reversal. Suppose $f$ fails Positive
  Involvement. Then we have a profile $\mathbf{P}$, $x \in
  \overline{f}(\mathbf{P})$, and $\mathbf{P}'$ that adds to
  $\mathbf{P}$ a ballot $L$ ranking $x$ in first place and yet $x \not\in
  \overline{f}(\mathbf{P}')$. Let $\mathbf{P}''$ be any profile that 
  adds to $\mathbf{P}'$ the converse of the ballot $L$, which ranks $x$ in last
  place. Then $\mathbf{P}''$ is the result of adding to $\mathbf{P}$ a pair of
  ballots in full reversal. So by Neutral Reversal, $x \in
  \overline{f}(\mathbf{P}'')$. But then the pair $\mathbf{P}'$ and $\mathbf{P}''$
  witnesses the failure of Negative Involvement for $\overline{f}$ and also $f$.
  Clearly, the same strategy of adding the converse of $L$ works for the other
  direction and also works for showing that a VCCR $f$ satisfying Neutral Reversal satisfies Positive Involvement
  in Defeat iff it satisfies Negative Involvement in Defeat.
\end{proof}
Thus, as long as we are focusing on VCCRs satisfying Neutral Reversal, and in particular on margin-based VCCRs such as Split Cycle, there is no loss of generality in focusing just on the positive axioms.

\section{Characterization of the Split Cycle VCCR}\label{Sec:CharVCCR}

We are now ready to prove our first main result, an axiomatic characterization of the
Split Cycle VCCR.

\begin{theorem}\label{Thm:Main} The Split Cycle VCCR  is the unique VCCR
  satisfying Anonymity, Neutrality, Availability, (Upward) Homogeneity,
  Monotonicity (for two-candidate profiles), Neutral Reversal, Coherent IIA,
  Coherent Defeat, and Positive Involvement in Defeat.
\end{theorem}

Holliday and Pacuit have already proved half of what we need.

\begin{theorem}[\citealt{HP2021}]\label{Thm:HP} If $f$ is a VCCR satisfying Anonymity, Neutrality, Availability, (Upward) Homogeneity, Monotonicity (for two-candidate profiles), Neutral Reversal, and Coherent IIA, then the Split Cycle VCCR refines $f$: for any profile $\mathbf{P}$, $sc(\mathbf{P}) \supseteq f(\mathbf{P})$.
\end{theorem}

It remains to prove that the axioms in Theorem \ref{Thm:Main} force $f$ to refine Split Cycle, in which case $f = sc$. In fact, we will prove
that just Coherent Defeat and Positive Involvement in Defeat together do so.
First we recall the following well-known extension lemma.

\begin{lemma}\label{lemma:extension}
  Any acyclic binary relation can be extended to a strict linear order.
\end{lemma}
\begin{proof}
  If $R$ is acyclic, then the reflexive and transitive closure $R^*$ of $R$ is a partial order. Applying the Szpilrajn extension theorem (\citealt{Szpilrajn1930}) to $R^*$, we obtain a total order $S$ extending $R^*$ and hence $R$. The asymmetric or equivalently the irreflexive part of $S$ is a strict linear order extending $R$. (Note that being acyclic, $R$ is automatically asymmetric and irreflexive.)
\end{proof}

Next comes the key lemma. The formulation is slightly cumbersome, and it actually
proves more than what we need for Theorem \ref{Thm:Main}. However, it is
precisely what we need for characterizing Split Cycle as a VSCC in Section \ref{Sec:CharVSCC}.

\begin{lemma}
  \label{lem:add-ballot}
  For any profile $\mathbf{P}$ and $(x, y) \in
  sc(\mathbf{P})$, if $\margin_{\mathbf{P}}(x, y) > 2$, then there is a
  ballot $L \in \mathcal{L}(X(\mathbf{P}))$ such that
  \begin{itemize}
  \item for any $z \in X(\mathbf{P}) \setminus \{y\}$ with $\margin_{\mathbf{P}}(y, z) \le
    0$, we have $(y, z) \in L$ (in particular, $(y, x) \in L$), and
  \item $(x, y) \in sc(\mathbf{P} + L)$.
  \end{itemize}
\end{lemma}
  \noindent In other words, if $x$ defeats $y$ according to Split Cycle by a sufficient margin
  (more than 2), then $x$ can still defeat $y$ after the addition of a
  specifically chosen ballot in which $y$ is ranked very high in the sense that
  $y$ is ranked above all candidates that it does not beat head-to-head,
  including $x$ since $y$ is defeated by $x$ according to Split Cycle.

For the proof of Lemma \ref{lem:add-ballot}, recall that in a graph containing vertices $y$ and $x$, a \textit{cut from $y$ to $x$} is a set of edges such
    that every path from $y$ to $x$ contains an edge from the set. Equivalently,
    a cut from $y$ to $x$ is a set of edges whose removal results in the loss of
    reachability of $x$ from $y$. The main idea of the proof is illustrated in Figure \ref{fig:all-ballot-example}.

\begin{figure}[h]
  \centering
  \begin{tikzpicture}[baseline]
    \node[circle,draw,very thick,minimum width=0.25in,blue] at (0,0) (x) {$x$}; 
    \node[circle,draw,minimum width=0.25in] at (3,0) (y) {$y$}; 
    \node[circle,draw,very thick,minimum width=0.25in,blue] at (1.5,1.5) (a) {$a$}; 
    \node[circle,draw,minimum width=0.25in] at (1.5,-1.5) (b) {$b$}; 
    \node[circle,draw,minimum width=0.25in] at (1.5,-3) (c) {$c$}; 
    \path[->,draw,thick] (x) to node[fill=white] {$4$} (y);
    \path[->,draw,very thick,medgreen] (y) to node[fill=white] {$2$} (b);
    \path[->,draw,thick] (b) to node[fill=white] {$2$} (x);
    \path[->,draw,thick] (a) to node[fill=white] {$4$} (x);
    \path[->,draw,thick] (y) to node[fill=white] {$4$} (c);
    \path[->,draw,very thick,medgreen] (c) to node[fill=white] {$2$} (x);
  \end{tikzpicture}
  $\xrightarrow[B:\, (y, v) \text{ with } \margin(y, v) \le 0]{C^{-1}: \text{ an
    inverted minimal cut}}$
  \begin{tikzpicture}[baseline]
    \node[circle,draw, minimum width=0.25in] at (0,0) (x) {$x$}; 
    \node[circle,draw,minimum width=0.25in] at (3,0) (y) {$y$}; 
    \node[circle,draw,minimum width=0.25in] at (1.5,1.5) (a) {$a$}; 
    \node[circle,draw,minimum width=0.25in] at (1.5,-1.5) (b) {$b$}; 
    \node[circle,draw,minimum width=0.25in] at (1.5,-3) (c) {$c$}; 
    \path[->,draw,thick] (y) to node[fill=white] {$B$} (a);
    \path[->,draw,thick] (y) to node[fill=white] {$B$} (x);
    \path[->,draw,thick] (b) to node[fill=white] {$C^{-1}$} (y);
    \path[->,draw,thick] (x) to node[near start,fill=white] {$C^{-1}$} (c);
  \end{tikzpicture}
    $\xrightarrow{\text{linearize}}$
  \begin{tikzpicture}[baseline]
    \node[circle,draw,minimum width=0.25in] at (0,1.5) (b) {$b$}; 
    \node[circle,draw,minimum width=0.25in] at (0,0.375) (y) {$y$}; 
    \node[circle,draw,minimum width=0.25in] at (0,-0.75) (a) {$a$}; 
    \node[circle,draw,minimum width=0.25in] at (0,-1.875) (x) {$x$}; 
    \node[circle,draw,minimum width=0.25in] at (0,-3) (c) {$c$}; 
    \path[->,draw,thick] (b) to  (y);
    \path[->,draw,thick] (y) to  (a);
    \path[->,draw,thick] (a) to  (x);
    \path[->,draw,thick] (x) to  (c);
  \end{tikzpicture}
  \caption{An example of generating a linear ballot $L$ as in Lemma
    \ref{lem:add-ballot}. Given $x, y$, we first identify (a) a minimal cut in
    the margin graph using only edges with margin smaller than $\margin(x, y)$
    and (b) the nodes that are not reachable from $y$ in one step (all
    highlighted on the left). Then we form the graph with the cut inverted and
    with arrows from $y$ to those highlighted nodes. This graph must be acyclic.
    Then we linearize this acyclic graph to obtain $L$.}
  \label{fig:all-ballot-example}
\end{figure}
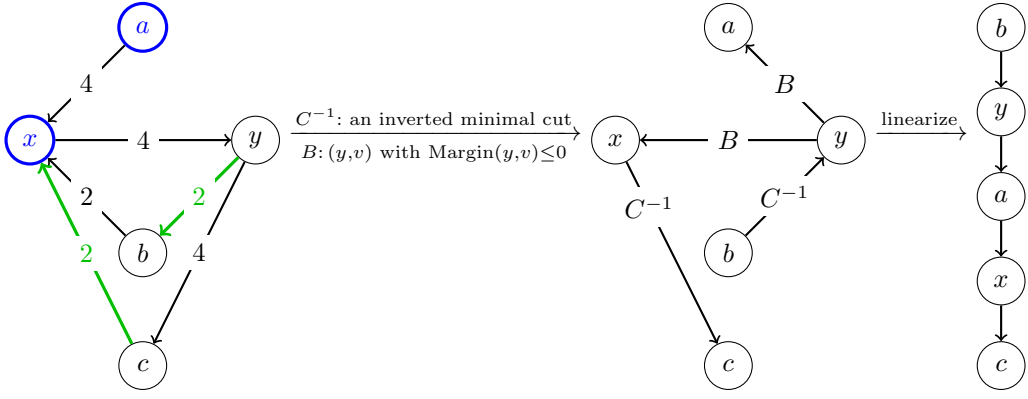

\begin{proof}
  Let $\mathcal{M}$ be the margin graph of $\mathbf{P}$ and $k =
  \margin_{\mathbf{P}}(x, y)$. Since $(x, y)
  \in sc(\mathbf{P})$, by Lemma \ref{lem:restrict-to-k} we know that $x$ is not reachable from $y$ in $\M
  \uph k$. This means that the set of edges in $\M$ of weight at most $k-1$ is
  a cut from $y$ to $x$. Thus, there is a minimal cut $C$ (minimal in
  the subset ordering) from $y$ to $x$ consisting only of edges of weight at
  most $k-1$, since the graph is finite. By the minimality of $C$, in the graph
  $\M \setminus C$ resulting from removing the edges in $C$ from $\M$, adding
  back any edge of $C$ reestablishes the reachability from $y$ to $x$. In other
  words,
  \begin{quote}
    for any edge $(u, v) \in C$, there is a path from $y$ to $u$ disjoint
    from $C$, and there is a path from $v$ to $x$ disjoint from $C$.
  \end{quote}
  From this, it follows that there are no connecting pairs of edges from $C$:
  there are no $u, v, w$ such that $(u, v)$ and $(v, w)$ are both in $C$, since
  otherwise there is a path from $v$ to $x$ and a path from $y$ to $v$, both
  disjoint from $C$, forming a path from $y$ to $x$ disjoint from $C$,
  contradicting that $C$ is a cut from $y$ to $x$.

  Now let $B = \{(y, z) \in \{y\} \times (X \setminus \{y\}) \mid
  \margin_{\mathbf{P}}(y, z) \le 0\}$ (for the particular $y$ given above). Note
  that $B$ is also the set $\{(y, z) \in \{y\} \times (X \setminus \{y\}) \mid
  (y, z) \not\in \mathcal{M}(\mathbf{P})\}$. Now we show that $C^{-1} \cup B$ is
  acyclic. Clearly there is no reflexive loop, and also there cannot be any
  cycle completely inside $C^{-1}$ since we have shown that $C$ and hence also
  $C^{-1}$ do not have connecting pairs of edges. Thus, if there is a cycle
  $\rho$, it must contain an edge $(y, z)$ from $B$. Let the next edge in $\rho$
  be $(z, u)$. Since $z$ is not $y$, this $(z, u)$ is not in $B$ and must be in
  $C^{-1}$. Now there are two cases. First, if $u = y$, then we have $(z, y) \in
  C^{-1}$ and thus $(y, z) \in C$. But $(y, z)$ is also in $B$, and by
  definition, $(y, z) \not \in \mathcal{M}(\mathbf{P})$. All edges in $C$ are in
  $\mathcal{M}(\mathbf{P})$, however, so we have a contradiction in this case.
  Second, if $u \not= y$, then consider the next edge $(u, v)$ in $\rho$. Since
  $u \not= y$, $(u, v) \not\in B$, so $(u, v) \in C^{-1}$. But then we have both
  $(z, u)$ and $(u, v)$ in $C^{-1}$, which is impossible since $C$ does not have
  connecting pairs of edges. Thus, $C^{-1} \cup B$ is acyclic. 
  
  Let $L$ be any strict linear order in $\mathcal{L}(X(\mathbf{P}))$ extending $C^{-1}
  \cup B$ by Lemma \ref{lemma:extension}. This is the ballot required by the
  lemma. Since $L$ extends $B$, it satisfies the first requirement, and in
  particular, $(y, x) \in L$. Let $\mathbf{P}' = \mathbf{P} + L$ and $\M' =
  \M(\mathbf{P}')$. Now we only need to show that $(x, y) \in sc(\mathbf{P}')$.
  Since $(y, x) \in L$, $\margin_{\mathbf{P}'}(x, y) = k-1$. Since for each $(u,
  v) \in C$, $\margin_{\mathbf{P}}(u, v) \le k-1$ and $(v, u) \in L$,
  $\margin_{\mathbf{P}'}(u, v) \le k-2 $. Let $E_{\le k-2}$ be the set of edges
  in $\M'$ of weight at most $k-2$. Then $C \subseteq E_{\le k-2}$. Note also
  that there may be edges in $\M'$ but not $\M$. However, if $(u, v)$ is in
  $\M'$ but not in $\M$, then since from $\mathbf{P}$ to $\mathbf{P}'$ only one
  ballot is added, $\margin_{\mathbf{P}'}(u, v) = 1$. Hence every edge in $\M'$
  but not in $\M$ is also in $E_{\le k-2}$ since $k > 2$ and hence $k - 2 \ge 1$. Now
  it is easy to see that $E_{\le k-2}$ is a cut from $y$ to $x$ in $\M'$: for
  any path $\rho$ from $y$ to $x$ in $\M'$, if it uses only edges in $\M$, then
  $\rho$ intersects $C$ and hence $E_{\le k-2}$; if $\rho$ uses a new edge in
  $\M'$ but not $\M$, then since the new edges are all in $E_{\le k-2}$, $\rho$
  intersects $E_{\le k-2}$. Thus, there is no path from $y$ to $x$ in $\M' \uph
  k-1$, and so $(x, y) \in sc(\mathbf{P}')$.
\end{proof}

Using Lemma \ref{lem:add-ballot}, we can now prove that just two of our axioms
force $f$ to refine Split Cycle. The idea is to use Lemma
\ref{lem:add-ballot} repeatedly so that defeat can be decided merely by Coherent
Defeat. An example is shown in Figure \ref{fig:reduction-example}.

\begin{theorem}\label{thm:at-least-as-resolute-vccr}
  If $f$ is a VCCR satisfying Coherent Defeat and Positive Involvement in
  Defeat, then $f$ refines the Split Cycle VCCR: for any
  profile $\mathbf{P}$, $f(\mathbf{P}) \supseteq sc(\mathbf{P})$.
\end{theorem}
\begin{proof}
  Pick an arbitrary profile $\mathbf{P}$ and $x, y
  \in X(\mathbf{P})$ such that $(x, y) \in sc(\mathbf{P})$. We only need to show
  that $(x, y) \in f(\mathbf{P})$. Let $k = \margin_{\mathbf{P}}(x, y)$ and $\M
  = \M(\mathbf{P})$. Now if $k \le 2$, then it is easy to see that there
  cannot be a majority path from $y$ to $x$, as then by the parity constraint from Lemma \ref{ParityLem},
  any majority path has strength at least $k$ and hence $x$ cannot defeat $y$ according to
  Split Cycle. So if $k \le 2$, we already have $(x, y) \in f(\mathbf{P})$
  by Coherent Defeat. If $k > 2$, then we inductively define $\mathbf{P}_0,
  \mathbf{P}_1, \mathbf{P}_2, \dots, \mathbf{P}_{k-2}$ and $L_1, L_2, \cdots,
  L_{k-2}$ where $\mathbf{P}_0 = \mathbf{P}$, $\mathbf{P}_{i+1} = \mathbf{P}_i +
  L_{i+1}$, and $L_{i+1}$ is obtained by applying Lemma \ref{lem:add-ballot} to
  $\mathbf{P}_i$. By the lemma, (1) $(y, x)$ is in each $L_i$, and (2) in each
  $\mathbf{P}_i$, $(x, y) \in sc(\mathbf{P}_i)$. As the margin from $x$ to
  $y$ decreases by $k-2$ times in this sequence, $(x, y) \in
  sc(\mathbf{P}_{k-2})$ but $\margin_{\mathbf{P}_{k-2}}(x, y) = 2$. This
  means that there is no majority path from $y$ to $x$. Thus, by Coherent
  Defeat, $(x, y) \in f(\mathbf{P}_{k-2})$. Finally, by Positive Involvement in
  Defeat in its contrapositive form, if $(x, y) \in f(\mathbf{P}_{i+1})$, then $(x, y)
  \in f(\mathbf{P}_i)$, since $x$ is ranked below $y$ in $L_{i+1}$. Thus, by an
  induction from ${k-2}$ back to $0$, $(x, y) \in f(\mathbf{P}_0) = f(\mathbf{P})$.
\end{proof}

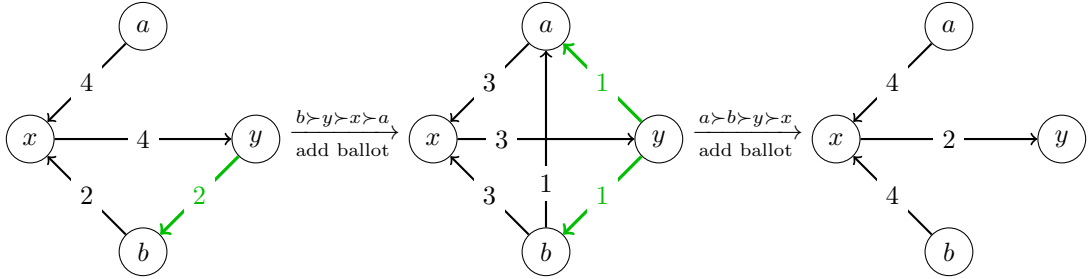
\begin{figure}[h]
  \centering
  \begin{tikzpicture}[baseline]
    \node[circle,draw, minimum width=0.25in] at (0,0) (x) {$x$}; 
    \node[circle,draw,minimum width=0.25in] at (3,0) (y) {$y$}; 
    \node[circle,draw,minimum width=0.25in] at (1.5,1.5) (a) {$a$}; 
    \node[circle,draw,minimum width=0.25in] at (1.5,-1.5) (b) {$b$}; 
    \path[->,draw,thick] (x) to node[fill=white] {$4$} (y);
    \path[->,draw,very thick,medgreen] (y) to node[fill=white] {$2$} (b);
    \path[->,draw,thick] (b) to node[fill=white] {$2$} (x);
    \path[->,draw,thick] (a) to node[fill=white] {$4$} (x);
  \end{tikzpicture}
  $\xrightarrow[\text{add ballot}]{b \succ y \succ x \succ a}$
  \begin{tikzpicture}[baseline]
    \node[circle,draw, minimum width=0.25in] at (0,0) (x) {$x$}; 
    \node[circle,draw,minimum width=0.25in] at (3,0) (y) {$y$}; 
    \node[circle,draw,minimum width=0.25in] at (1.5,1.5) (a) {$a$}; 
    \node[circle,draw,minimum width=0.25in] at (1.5,-1.5) (b) {$b$}; 
    \path[->,draw,thick] (x) to node[near start,fill=white] {$3$} (y);
    \path[->,draw,very thick,medgreen] (y) to node[fill=white] {$1$} (b);
    \path[->,draw,very thick,medgreen] (y) to node[fill=white] {$1$} (a);
    \path[->,draw,thick] (b) to node[fill=white] {$3$} (x);
    \path[->,draw,thick] (b) to node[near start,fill=white] {$1$} (a);
    \path[->,draw,thick] (a) to node[fill=white] {$3$} (x);
  \end{tikzpicture}
  $\xrightarrow[\text{add ballot}]{a \succ b \succ y \succ x}$
  \begin{tikzpicture}[baseline]
    \node[circle,draw, minimum width=0.25in] at (0,0) (x) {$x$}; 
    \node[circle,draw,minimum width=0.25in] at (3,0) (y) {$y$}; 
    \node[circle,draw,minimum width=0.25in] at (1.5,1.5) (a) {$a$}; 
    \node[circle,draw,minimum width=0.25in] at (1.5,-1.5) (b) {$b$}; 
    \path[->,draw,thick] (x) to node[fill=white] {$2$} (y);
    \path[->,draw,thick] (b) to node[fill=white] {$4$} (x);
    \path[->,draw,thick] (a) to node[fill=white] {$4$} (x);
  \end{tikzpicture}
  \caption{An example of eliminating all paths from $y$ to $x$ by adding ballots
    that put $y$ in a sufficiently high position. For each of the first two
    margin graphs, a minimal cut from $y$ to $x$ is highlighted by thickened
    arrows. There are many ways to eliminate all paths from $y$ to $x$, and the
    key is to do this quickly before the positive margin from $x$ to $y$ runs 
    out.}
  \label{fig:reduction-example}
\end{figure}

Combining Theorems \ref{thm:at-least-as-resolute-vccr} and \ref{Thm:HP}, any
VCCR satisfying the axioms in Theorem \ref{Thm:Main} refines and is refined by Split Cycle, so it is equal to Split Cycle. This completes the
proof of Theorem \ref{Thm:Main}.

\section{Axioms on VSCCs}\label{Sec:AxVSCC}
In this section, we characterize Split Cycle as a VSCC. As before, we first introduce the standard axioms (\ref{Sec:VSCCstandard}). Then we define Coherent IIA for VSCCs (\ref{Sec:CohIIAVSCC}). Care must be taken here as existing IIA-like axioms for SCCs have implicit commitments irrelevant to the spirit of IIA. Next we define Coherent Defeat for VSCCs (\ref{Sec:CohDefVSCC}), the obvious analogue of Coherent Defeat for VCCRs, and finally a new axiom of Tolerant Positive Involvement (\ref{Sec:PosInvVSCC}), a proper strengthening of Positive Involvement necessary for our characterization.

\subsection{Standard axioms}\label{Sec:VSCCstandard}

Adapting the standard axioms for VCCRs from Section \ref{Sec:StandAxVCCR} to VSCCs, we obtain the following.

\begin{definition}
  Let $F$ be a VSCC.
  \begin{enumerate}
  \item $F$ satisfies \textit{Anonymity} if for any profiles $\mathbf{P}$ and $\mathbf{P}'$, if $V(\pfp) = V(\pfp')$ and there is a bijection $\pi$ from $V(\mathbf{P})$ to $V(\mathbf{P}')$ such that for any $i \in V(\mathbf{P}')$, $\mathbf{P}'(i) = \mathbf{P}(\pi(i))$, then  $F(\mathbf{P}) = F(\mathbf{P}')$; and $F$ satisfies \textit{Neutrality} if for any profiles $\mathbf{P}$ and $\mathbf{P}'$, if $V(\mathbf{P}) = V(\mathbf{P}')$, $X(\pfp) = X(\pfp')$, and there is a bijection $\pi$ from $X(\mathbf{P})$ to $X(\mathbf{P}')$ such that for any $i \in V(\mathbf{P})$ and  $x, y \in X(\mathbf{P})$, $(x, y) \in \mathbf{P}(i)$ iff $(\pi(x), \pi(y)) \in \mathbf{P}'(i)$, then for any $x \in X(\mathbf{P})$, $x \in F(\mathbf{P})$ iff $\pi(x) \in F(\mathbf{P}')$.

  \item $F$ satisfies \textit{Homogeneity} (resp.~\textit{Upward Homogeneity})
    if for any profile $\mathbf{P}$ and $2\mathbf{P}$, where $2\mathbf{P}$ is
    the result of replacing each voter in $\mathbf{P}$ by $2$ copies of that
    voter, we have $F(\mathbf{P}) = F(2\mathbf{P})$ (resp. $F(\mathbf{P})
    \supseteq F(2\mathbf{P})$\footnote{Equivalently: $X(\mathbf{P}) \setminus
      F(\mathbf{P}) \subseteq X(\mathbf{P}) \setminus F(2\mathbf{P})$. In short,
      Upward Homogeneity for VSCCs means losing is preserved under doubling the
      profile.}).
  \item $F$ satisfies \textit{Monotonicity} (resp.~\textit{Monotonicity for two-candidate profiles}) if for any profile (resp. two-candidate profile) $\mathbf{P}$ and $\mathbf{P}'$ obtained from $\mathbf{P}$ by moving $x \in X(\mathbf{P})$ up one place in some voter's ballot (see Definition \ref{Def:StandAxVCCR} for a precise formulation), we have $x \in F(\mathbf{P})$ only if $x \in F(\mathbf{P}')$.

  \item $F$ satisfies \textit{Neutral Reversal} if for any profile $\mathbf{P}$ and $\mathbf{P}'$ obtained from $\mathbf{P}$ by adding two voters whose ballots are converses of each other, we have $F(\mathbf{P}) = F(\mathbf{P}')$.
  \end{enumerate}
\end{definition}

Observe that we no longer have the Availability axiom for VSCCs, since by definition they must return a non-empty set of winners.

\subsection{Coherent IIA}\label{Sec:CohIIAVSCC}

We now turn to the crucial question of how to formulate the analogue for VSCCs of Coherent IIA for VCCRs. First, we must ask: what is the analogue for VSCCs of Arrow's IIA for VCCRs? One  answer is the following, adopting terminology of Denicolo \citeyearpar{Denicolo2000}, based on \citealt{Hansson1969}.

\begin{definition} A VSCC $F$ satisfies \textit{Hansson's Pairwise Independence} (HPI) if for any profiles $\mathbf{P}$ and $\mathbf{P}'$ with $x,y\in X(\mathbf{P})$, if $x\in F(\mathbf{P})$, $y\not\in F(\mathbf{P})$, and $\mathbf{P}_{\mid\{x,y\}}=\mathbf{P}_{\mid\{x,y\}}'$, then $y\not\in F(\mathbf{P}')$.
\end{definition}
\noindent The problem with this proposal, simply put, is this: who said $x$ is the candidate who defeats $y$ in $\mathbf{P}$? If we knew on the basis of $x\in F(\mathbf{P})$ and $y\not\in F(\mathbf{P})$ that $x$ defeats $y$ in $\mathbf{P}$, then we could indeed conclude that $x$ still defeats $y$ in $\mathbf{P}'$, so $y\not\in F(\mathbf{P}')$. But it does not follow (without assumptions beyond IIA) from $x\in F(\mathbf{P})$ and $y\not\in F(\mathbf{P})$ that $x$ in particular is the candidate who defeats $y$ in $\mathbf{P}$; all that follows it that $x$ is undefeated and $y$ is defeated by someone or other. HPI in effect assumes that $F$ is rationalized by a variable-election \textit{social welfare function} (see \citealt[Theorem 1]{Denicolo1993}), i.e., a VCCR $f$ for which $f(\mathbf{P})$ is always a strict weak order, in which case $x\in F(\mathbf{P})$ and $y\not\in F(\mathbf{P})$ do imply that $x$ defeats $y$. Thus, HPI smuggles in an additional  ``social rationality'' assumption, which should not be part of a pure independence condition. But we can fix this problem with the following weaker definition.

\begin{definition} A VSCC $F$ satisfies \textit{Pure IIA} if for any profile $\mathbf{P}$ and $y\in X(\mathbf{P})$, if $y\not\in F(\mathbf{P})$, then there is an $x\in X(\mathbf{P})$ such that for any profile $\mathbf{P}'$ with $\mathbf{P}_{\mid\{x,y\}}=\mathbf{P}_{\mid\{x,y\}}'$, we have $y\not\in F(\mathbf{P}')$.
\end{definition}

The following proposition verifies that Pure IIA is the correct analogue of IIA
for VSCCs.

\begin{proposition} For any VSCC $F$, the following are equivalent:
\begin{enumerate}
\item\label{PureIIAProp1} $F$ satisfies Pure IIA;
\item\label{PureIIAProp2} $F$ is defeat-rationalized by a VCCR satisfying IIA.
\end{enumerate}
\end{proposition}

\begin{proof} From \ref{PureIIAProp1} to \ref{PureIIAProp2}, suppose $F$ satisfies Pure IIA. Define $f$ as follows: for $x,y\in X(\mathbf{P})$, $x$ defeats $y$ in $\mathbf{P}$ according to $f$ if for every $\mathbf{P}'$ with $x\in X(\mathbf{P}')$ and $\mathbf{P}_{\mid\{x,y\}}=\mathbf{P}_{\mid\{x,y\}}'$, we have $y\not\in F(\mathbf{P}')$. Thus, $y$ is undefeated in $\mathbf{P}$ according to $f$ if and only if for every $x\in X(\mathbf{P})$, there is a $\mathbf{P}'$ with $x\in X(\mathbf{P}')$,  $\mathbf{P}_{\mid\{x,y\}}=\mathbf{P}_{\mid\{x,y\}}'$, and  $y\in F(\mathbf{P}')$; this implies, by Pure IIA, that $y\in F(\mathbf{P})$. Conversely, $y\in F(\mathbf{P})$ implies that $y$ is undefeated in $\mathbf{P}$ according to $f$ by taking $\mathbf{P}'=\mathbf{P}$. Thus, $F$ is deafeat-rationalized by $f$.

From \ref{PureIIAProp2} to \ref{PureIIAProp1}, suppose $F$ is defeat-rationalized by a VCCR $f$ satisfying IIA. To show that $F$ satisfies Pure IIA, suppose $\mathbf{P}$ is a profile with $y\not\in F(\mathbf{P})$. Then since $f$ defeat-rationalizes $F$, there is an $x\in X(\mathbf{P})$ such that $x$ defeats $y$ in $\mathbf{P}$ according to $f$. Then since $f$ satisfies IIA, for any profile $\mathbf{P}'$ with $\mathbf{P}_{\mid\{x,y\}}=\mathbf{P}_{\mid\{x,y\}}'$, we have that $x$ defeats $y$ in $\mathbf{P}'$ according to $f$, which by defeat-rationalization implies that $y\not\in F(\mathbf{P}')$. This shows that $F$ satisfies Pure IIA.\end{proof}

Now, just as we translated IIA for VCCRs to Pure IIA for VSCCs, we
define Coherent IIA for VSCCs as follows, using the notation
$\rightsquigarrow _{x, y}$ from Definition \ref{def:lc-accessibility}.

\begin{definition}\label{def:coherent-IIA-VSCC}
  A VSCC $F$ satisfies \textit{Coherent IIA} if for any profile $\mathbf{P}$ and
  $y \in X(\mathbf{P})$, if $y \not\in F(\mathbf{P})$, then there is $x \in
  X(\mathbf{P})$ such that for any profile $\mathbf{P}'$ with $\mathbf{P}
  \rightsquigarrow_{x, y}\mathbf{P}'$, we have  $y \not\in F(\mathbf{P}')$.
\end{definition}

 \begin{remark}Dong \citeyearpar[Theorem 12]{Dong2021} arrives at a similar axiom (dubbed \textit{crucial defeat}) for VSCCs inspired by Coherent~IIA for VCCRs, which Dong uses to characterize Split Cycle as the VSCC satisfying several axioms and refining all VSCCs satisfying those axioms, as Theorem \ref{Thm:HP} did for VCCRs. We will give our analogue of Theorem \ref{Thm:HP} for VSCCs in Theorem \ref{thm:no-more-resolute-vscc} below. \end{remark}

Due to Definition \ref{def:coherent-IIA-VSCC} being largely a translation, we immediately have the following.
\begin{proposition}\label{prop:ciia-vccr-to-vscc}
  For any VCCR $f$ satisfying Coherent IIA, $\overline{f}$ also satisfies Coherent IIA.
\end{proposition}
The converse is not true, and this is naturally due to the weaker expressivity of VSCCs: many VCCRs can defeat-rationalize one VSCC. For example, we can define a VCCR $f$ by $(x, y) \in f(\pfp)$ iff $y \not\in SC(\pfp)$. Then $SC = \overline{f}$, but $f$ does not satisfy Coherent IIA, since we may have $\margin_\pfp(x, y) < 0$ and yet $(x, y) \in f(\pfp)$, while we must have $(x, y) \not\in f(\pfp_{|\{x,y\}})$.

\subsection{Coherent Defeat}\label{Sec:CohDefVSCC}
Translating Coherent Defeat from VCCRs to VSCCs is straightforward.
\begin{definition}
  A VSCC $F$ satisfies \textit{Coherent Defeat} if for any profile $\mathbf{P}$
  and $x, y \in X(\mathbf{P})$, if $\margin_{\mathbf{P}}(x,y) > 0$ and there is
  no majority path from $y$ to $x$, then $y \not\in F(\mathbf{P})$.
\end{definition}
\noindent In fact, this is precisely Immunity to Binary Arguments ($\mathrm{Im}_A$) in
\citealt{Heitzig2002} with $A$ being the majority preference relation, only phrased
in our variable-election setting.

\subsection{Tolerant Positive Involvement}\label{Sec:PosInvVSCC}
As we noted when we first introduced Positive Involvement in Section \ref{Sec:PosInvVCCR}, a strengthening is required for characterizing Split Cycle. For the Split Cycle VCCR, we used the strengthening of Positive Involvement in Defeat. For the Split Cycle VSCC, we use the following.
\begin{definition}
  A VSCC $F$ satisfies \textit{\tolerantpi} if for any profile $\mathbf{P}$, if $x \in F(\mathbf{P})$ and $\mathbf{P}'$ is obtained
  from $\mathbf{P}$ by adding one new voter who ranks $x$ above every other candidate
  $y$ such that $x$ is not majority preferred to $y$ in $\mathbf{P}$, then $x
  \in F(\mathbf{P}')$.
\end{definition}
We use the word `tolerant' since Tolerant Positive Involvement is applicable
more broadly than Positive Involvement; it says that for a winner $x$'s winning
to be preserved under the addition of a ballot $L$, $x$ can tolerate being
ranked below some candidates as long as $x$ is ranked high enough: $x$ is above
all those $y$ to which $x$ is not majority preferred before adding $L$. In other
words, if $L$ weakens all potential weaknesses of $x$ in the sense that $L$
decreases all non-negative margins over $x$, then adding $L$ will not cause $x$
to lose. In contrast, Positive Involvement guarantees the preservation of $x$'s
winning only for ballots that weaken all margins over $x$ (by putting $x$ at the
top), not just the non-negative margins.

\begin{example}To illustrate the difference between Positive Involvement and Tolerant Positive Involvement, let us consider Instant Runoff Voting---in particular, the version (as in \citealt[p.~7]{Taylor2008}) where at each stage, all candidates with the fewest first-place votes are eliminated, unless all candidates would thereby be eliminated. Consider the following anonymized profile $\mathbf{P}$:
\begin{center}
$\mathbf{P}$\qquad
\begin{tabular}[]{ccc}
  $3$ & $4$ & $2$ \\
  \hline 
  $x$ & $y$ & $z$ \\
  $y$ & $x$ & $x$ \\
  $z$ & $z$ & $y$ \\
\end{tabular}
\end{center}
According to Instant Runoff, $x$ is the winner in
$\mathbf{P}$. It is also not hard to see that Instant Runoff satisfies
Positive Involvement, and in this particular case, if we add a voter whose
ballot puts $x$ at the top, then $x$ will remain the sole winner as $z$ and $y$
will still be eliminated consecutively. Now note that the margin graph of
$\mathbf{P}$ is 
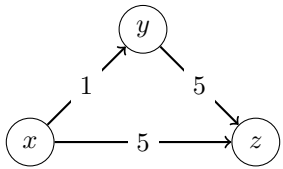
\begin{center}
  \begin{tikzpicture}[baseline]
    \node[circle,draw,minimum width=0.25in] at (0,0) (x) {$x$}; 
    \node[circle,draw,minimum width=0.25in] at (3,0) (z) {$z$}; 
    \node[circle,draw,minimum width=0.25in] at (1.5,1.5) (y) {$y$}; 
    \path[->,draw,thick] (x) to node[fill=white] {$1$} (y);
    \path[->,draw,thick] (x) to node[fill=white] {$5$} (z);
    \path[->,draw,thick] (y) to node[fill=white] {$5$} (z);
  \end{tikzpicture}
\end{center}
This shows that $x$ is the Condorcet winner, the candidate who is majority
preferred to every other candidate. In other words, there are no other
candidates to whom $x$ is not majority preferred. This makes the
precondition for Tolerant Positive Involvement trivial, and for Tolerant
Positive Involvement to hold for Instant Runoff, $x$ should remain a winner no
matter what ballot is added to $\mathbf{P}$. But this is not the case. Let 
$\mathbf{Q}$ be the result of adding to $\mathbf{P}$ a new voter whose ballot is
$z \succ x \succ y$: 
\begin{center}
\begin{tabular}[]{ccc}
  $3$ & $4$ & $3$ \\
  \hline 
  $x$ & $y$ & $z$ \\
  $y$ & $x$ & $x$ \\
  $z$ & $z$ & $y$ \\
\end{tabular}
\end{center}
According to Instant Runoff, $x$ is eliminated in the first round and therefore does not win in $\mathbf{Q}$. Thus, Instant Runoff satisfies Positive Involvement but not Tolerant Positive Involvement.
\end{example}

Since Tolerant Positive Involvement is a strengthening of Positive Involvement, which many Condorcet methods violate (see \citealt{Perez2001}), and has a majoritarian component, which scoring rules like Borda are unlikely to respect, it is perhaps not surprising that many well-known VSCCs fail Tolerant Positive Involvement. In fact, we know of only two standard VSCCs other than Split Cycle that satisfy the axiom, namely Minimax and Weighted Covering as a VSCC. 

\begin{proposition}
  The Minimax VSCC $MM$, defined by 
  \[
  MM(\pfp) = \{x \in X(\pfp) \mid \forall y \in X(\pfp), \weakness_\pfp(x) \le \weakness_\pfp(y)\},
  \]
  satisfies Tolerant Positive Involvement.
\end{proposition}
\begin{proof}
  The idea is similar to that in the proof showing that Weighted Covering satisfies Positive Involvement in Defeat (Proposition \ref{WCPID}). Let $\pfp$ be any profile, $x \in MM(\pfp)$, and $\pfp'$ the result of adding to $\pfp$ a single voter who ranks $x$ above all $y \in X(\pfp)$ such that $\margin_{\pfp}(y, x) \ge 0$. Then observe that $\weakness_{\pfp'}(x)$ is either $0$ or $\weakness_{\pfp}(x) - 1$. The weakness of any candidate from $\pfp$ to $\pfp'$ can decrease by at most $1$ and must be at least $0$. Thus $\weakness_{\pfp'}(x)$ is still the smallest possible in $\pfp'$.
\end{proof}

The following proposition relates Tolerant Positive Involvement to Positive Involvement in Defeat, which can be used to show that Split Cycle and Weighted Covering as VSCCs satisfy Tolerant Positive Involvement. For this we recall the axiom of Majority Defeat (Definition \ref{MajDefeat}): for any profile $\mathbf{P}$ and $x, y \in X(\mathbf{P})$, if $(x, y) \in f(\mathbf{P})$, then $\margin_{\mathbf{P}}(x,y) > 0$, i.e., $x$ defeats $y$ only if $x$ is majority preferred to $y$.

\begin{proposition}\label{prop:pid-to-tpi}
  If an acyclic VCCR $f$ satisfies Positive Involvement in Defeat and Majority Defeat, then $\overline{f}$ satisfies Tolerant Positive Involvement. In particular, the Split Cycle VSCC and the Weighted Covering VSCC $WC = \overline{wc}$ satisfy Tolerant Positive Involvement since both $sc$ and $wc$ satisfy Positive Involvement in Defeat and Majority Defeat.
\end{proposition}
\begin{proof}
  Let $f$ be an acyclic VCCR satisfying Positive Involvement in Defeat and Majority
  Defeat. Let $\mathbf{P}$ be a profile, $x \in
  F(\mathbf{P})$, and $\mathbf{P}'$ a profile obtained by adding to $\mathbf{P}$
  a new voter with ballot $L$. Moreover, assume that for any $y \in
  X(\mathbf{P}) \setminus \{x\}$, if $\margin_{\mathbf{P}}(y, x) \ge 0$, then $x
  L y$. Our goal is to show that $x \in \overline{f}(\mathbf{P}')$, i.e., that for any $y \in X(\mathbf{P})
  \setminus \{x\}$, we have $(y, x) \not\in f(\mathbf{P}')$. Now there are two cases:
  either $\margin_{\mathbf{P}}(y, x) \ge 0$ or $\margin_{\mathbf{P}}(y, x) < 0$.
  In the former case, by the assumption on $L$, $x L y$. Also, since $x \in
  \overline{f}(\mathbf{P})$, we have  $(y, x) \not\in f(\mathbf{P})$. So $(y, x) \not\in
  f(\mathbf{P}')$ since $f$ satisfies Positive Involvement in Defeat. In the
  latter case, since only one new voter is added, $\margin_{\mathbf{P}'}(y, x)
  \le 0$. Then $(y, x) \not\in f(\mathbf{P}')$ by Majority Defeat.
\end{proof}

\section{Characterization of the Split Cycle VSCC}\label{Sec:CharVSCC}
We are now ready to axiomatically characterize Split
Cycle as a VSCC.

\begin{theorem}
  \label{thm:axiomatization-vscc}
  The Split Cycle VSCC is the unique VSCC satisfying Anonymity,
  Neutrality, (Upward) Homogeneity, Neutral Reversal, Monotonicity
  (for two-candidate profiles), Coherent IIA, Coherent Defeat, and Tolerant
  Positive Involvement.
\end{theorem}
\noindent Unlike the situation with the Split Cycle VCCR where we can directly use
the theorem in \cite{HP2021} to finish one direction of the theorem, we need to
prove both directions specifically for the VSCC. On the face of it, this can be
explained by the fact that the standard axioms, when stated for VSCCs, are weaker than
their counterparts for VCCRs, since a VSCC carries less information than a VCCR, and the standard axioms involve preservation of information over
profiles. For example, Upward Homogeneity for VSCCs says that losing is
preserved under doubling profiles. However, one could be losing for different
reasons, i.e., by being defeated by different candidates. Upward Homogeneity for
VCCRs also preserves the defeaters of the losers by preserving the whole defeat
relation, while Upward Homogeneity for VSCCs cannot track the defeaters since a
VSCC does not carry the full information of defeat. It is an interesting
question when we can recover from a VSCC a VCCR that defeat rationalizes the
VSCC while satisfying various natural preservation axioms for the defeat relation, such as IIA, Coherent IIA, and Positive Involvement in Defeat. We leave this for future work.

We start with the direction that uses the standard axioms. First, by a standard
argument using symmetries, Anonymity, Neutrality, and Monotonicity for
two-candidate profiles together ensure that on two-candidate
profiles, the loser (if there is one) is majority dispreferred. This will be
used a few times later and is  already used in \citealt[Lemma~4.2]{HP2021}, so we formally
state it here. 

\begin{definition}
  A VCCR $f$ satisfies \emph{Majority Defeat for two-candidate profiles} if for
  any profile $\mathbf{P}$ with $X(\mathbf{P}) = \{x, y\}$ and $x \not= y$, if $x
  f(\mathbf{P}) y$ then $\margin_{\mathbf{P}}(x, y) > 0$.

  A VSCC $F$ satisfies \emph{Majority Defeat for two-candidate profiles} if for
  any profile $\mathbf{P}$ with $X(\mathbf{P}) = \{x, y\}$ and $x \not= y$, if $y \not\in
  F(\mathbf{P})$ then $\margin_{\mathbf{P}}(x, y) > 0$.
\end{definition}
\begin{lemma}\label{lem:majority-defeat}
  Any VCCR or VSCC $f$ satisfying Anonymity, Neutrality, and Monotonicity for
  two-candidate profiles satisfies Majority Defeat for
  two-candidate profiles.
\end{lemma}

Now we prove the analogue for VSCCs of Theorem \ref{Thm:HP}.

\begin{theorem}
  \label{thm:no-more-resolute-vscc}
  Let $SC$ be the Split Cycle VSCC and $F$ a VSCC satisfying Anonymity,
  Neutrality, (Upward) Homogeneity, Neutral Reversal, Monotonicity
  (for two-candidate profiles), and Coherent IIA. Then $SC$ refines $F$: for any profile $\mathbf{P}$, $SC(\mathbf{P}) \subseteq
  F(\mathbf{P})$.
\end{theorem}
\begin{proof}
  The proof is only slightly different than the corresponding one in
  \citealt{HP2021}. Let $\mathbf{P}$ be any profile and $y \in SC(\mathbf{P})$. Now
  we need to show that $y \in F(\mathbf{P})$. Using Upward Homogeneity,
  $F(\mathbf{P}) \supseteq F(2\mathbf{P})$. Hence we only need to show that $y
  \in F(2\mathbf{P})$. Since $SC(\mathbf{P}) = SC(2\mathbf{P})$, for notational
  convenience, let us assume from now on that $\mathbf{P}$ has an even number of
  voters and $y \in SC(\mathbf{P})$.

  Let $M$ be the largest margin of $\mathbf{P}$. Since $\mathbf{P}$ is a linear profile, $M$ 
   is even. Let $\mathbf{P}'$ be the result of adding to $\mathbf{P}$ enough
  pairs of voters with converse ballots so that $|V(\mathbf{P}')| \ge M \cdot
  |X(\mathbf{P})|$. Since adding these pairs of voters does not change the margins
  or the number of candidates, $M$ is also the largest margin of $\mathbf{P}'$ and 
  $|V(\mathbf{P}')| \ge M \cdot |X(\mathbf{P}')|$.

  By Neutral Reversal, it is enough to show that $y \in F(\mathbf{P}')$. Toward
  a contradiction, suppose that $y \not\in F(\mathbf{P}')$. Then by Coherent
  IIA, there is an $x \in X(\mathbf{P}') = X(\mathbf{P})$ such that for any profile
  $\mathbf{Q}'$ with $\mathbf{P}' \rightsquigarrow_{x, y} \mathbf{Q}'$, $y
  \not\in F(\mathbf{Q}')$. Consider the two-candidate profile $\mathbf{Q}' =
  \mathbf{P}'_{|\{x, y\}}$. Clearly $\mathbf{P}' \rightsquigarrow_{x, y} \mathbf{Q}'$, and
  by assumption $y \not\in F(\mathbf{Q}')$. Thus, by Lemma
  \ref{lem:majority-defeat}, $\margin_{\mathbf{P}'}(x, y) =
  \margin_{\mathbf{Q}'}(x, y) > 0$.

  Let $m = \margin_{\mathbf{P}'}(x, y) > 0$, which is even since
  $|V(\mathbf{P}')|$ is even. Since $y \in SC(\mathbf{P}')$, there must be a
  majority cycle $( x, y, z_1, z_2, \cdots, z_n, x )$ such that all
  the margins of the consecutive edges are at least $m$. Now construct
  $\mathbf{Q}$ and $\mathbf{Q}'$ as in \citealt{HP2021}: 
\renewcommand{\arraystretch}{1.7}
\begin{figure}[t]
  \begin{center}
\begin{tabular}{cc|cc|cc|c|cc}
 $m/2$ & $m/2$ & $m/2$ & $m/2$ & $m/2$ & $m/2$ & $\cdots$ & $m/2$ & $m/2$    \\\hline
$\boldsymbol{x}$ & $z_n$ & $\boldsymbol{y}$ & $x$ & $\boldsymbol{z_1}$ & $y$ & $\cdots$ & $ \boldsymbol{z_n}$ & $z_{n-1}$      \\
$\boldsymbol{y}$ &  $\vdots$ & $\boldsymbol{z_1}$ & $z_n$ &  $\boldsymbol{z_2}$ & $x$ & $\cdots$ & $\boldsymbol{x}$ & $\vdots$   \\
$z_1$ &  $\vdots$ & $z_2$ & $\vdots$ & $\vdots$ & $z_n$ & $\cdots$ & $y$ & $z_1$    \\
$\vdots$ &  $z_1$ & $\vdots$ & $z_2$ & $z_n$ & $\vdots$ & $\cdots$ & $z_1$ & $y$   \\
$\vdots$ &  $\boldsymbol{x}$ & $z_n$ & $\boldsymbol{y}$ & $x$ & $\boldsymbol{z_1}$ & $\cdots$ & $\vdots$ & $ \boldsymbol{z_n}$   \\
$z_n$ &  $\boldsymbol{y}$ & $x$ & $\boldsymbol{z_1}$ & $y$ & $\boldsymbol{z_2}$ & $\cdots$ & $z_{n-1}$ & $\boldsymbol{x}$   \\
\end{tabular}
\end{center}
\caption{the profile $\mathbf{Q}$.} \label{McGarvey}
\end{figure}
  \begin{itemize}
  \item First, take $m$ voters in $\mathbf{P}'$ who rank $x$ above $y$, and let
    half of them vote $x \succ y \succ z_1 \succ z_2 \succ \cdots \succ z_n$
    and the other half vote $z_n \succ z_{n-1} \succ \cdots \succ z_1 \succ x
    \succ y$. Since $m = \margin_{\mathbf{P}'}(x, y) = \margin_{\mathbf{P}}(x,
    y)$, the required $m$ voters can be found.
  \item The rest of the voters can then be split evenly according to whether
    they rank $x$ above $y$ or $y$ above $x$. So take $m/2$ fresh voters in
    $\mathbf{P}'$ ranking $y$ above $x$ and let them vote $y \succ z_1 \succ z_2
    \succ \cdots \succ z_n \succ x$. Then take $m/2$ fresh voters in
    $\mathbf{P}'$ ranking $x$ above $y$ and let them vote $x \succ z_n \cdots
    \succ z_2 \succ y \succ z_1$.
  \item For any $i = 1, 2, \cdots, n-1$, first take $m/2$ fresh voters in
    $\mathbf{P}'$ ranking $x$ above $y$ and let them vote $z_i \succ z_{i+1}
    \succ z_{i+2} \succ \cdots \succ z_n \succ x \succ y \succ z_1 \succ \cdots
    \succ z_{i-1}$. Then take $m/2$ fresh voters in $\mathbf{P}'$ ranking $y$
    above $x$ and let them vote $z_{i-1} \succ \cdots \succ z_1 \succ y \succ x
    \succ z_n \succ \cdots \succ z_{i+2} \succ z_i \succ z_{i+1}$.
  \item Take $m/2$ fresh voters in $\mathbf{P}'$ ranking $x$ above $y$ and let
    them vote $z_n \succ x \succ y \succ z_1 \succ \cdots \succ z_{n-1}$. Then
    take $m/2$ fresh voters in $\mathbf{P}'$ ranking $y$ above $x$ and let them
    vote $z_{n-1} \succ \cdots \succ z_1 \succ y \succ z_n \succ x$.
  \item The above uses $(n + 2) \cdot m$ many voters in $\mathbf{P}'$. Since $n
    + 2 \le X(\mathbf{P}')$, $m \le M$, and $V(\mathbf{P}') \ge M \cdot
    |X(\mathbf{P}')|$, we have enough voters for the above construction. Thus,
    let $\mathbf{Q}$ be the profile using the voters used above with their
    ballots also specified as above. Figure \ref{McGarvey} gives the anonymized
    summary of profile $\mathbf{Q}$. There may be unused voters in
    $\mathbf{P}'$. That is, $V(\mathbf{P}') \setminus V(\mathbf{Q})$ may be
    non-empty. However, they must come in pairs with respect to whether they
    rank $x$ above $y$ or $y$ above $x$, since $\margin_{\mathbf{Q}}(x, y) =
    \margin_{\mathbf{P}'}(x, y)$. Thus pick a ballot $L \in \mathcal{L}(\{x, y, z_1,
    \cdots, z_n\})$ with $x L y$, and construct $\mathbf{Q}'$ by adding to
    $\mathbf{Q}$ the voters in $X(\mathbf{P}')\setminus X(\mathbf{Q})$ and let
    them vote $L$ if they rank $x$ above $y$ in $\mathbf{P}'$ and $L^{-1}$
    otherwise.
    \end{itemize}
    Observe now that $\mathbf{P}' \rightsquigarrow_{x, y} \mathbf{Q}'$ since the
    margin graph of $\mathbf{Q}'$ is a pure cycle with $x, y, z_1, z_2, \cdots,
    z_n$ with each consecutive edge's margin being $m$, which is no greater
    than their original margin in $\mathbf{P}'$. Thus, by Coherent IIA, $y
    \not\in F(\mathbf{Q}')$. Since $\mathbf{Q}'$ is the result of adding
    (possibly zero) reversal pairs to $\mathbf{Q}$, $y \not\in F(\mathbf{Q})$.
    However, using the symmetries of the profile $\mathbf{Q}$, by Anonymity and
    Neutrality, $x$ and the $z_i$'s are also not in $F(\mathbf{Q})$, making $F(\mathbf{Q}) = \varnothing$. (This is the same argument as in \citealt{HP2021}; if we
    consider the rotation $\sigma$ of candidates with $\sigma(x) = y, \sigma(y)
    = z_1, \sigma(z_i) = z_{i+1}, $ and $\sigma(z_n) = x$, then applying
    $\sigma$ to $\mathbf{Q}$ results in the same anonymized profile. See Figure
    \ref{McGarveyPermuted}.) This contradicts the assumption that $F$ is a VSCC.    
    \end{proof}

\begin{figure}[t]
  \begin{center}
\begin{tabular}{cc|cc|cc|c|cc}
 $m/2$ & $m/2$ & $m/2$ & $m/2$ & $m/2$ & $m/2$ & $\cdots$ & $m/2$ & $m/2$    \\\hline
  $\boldsymbol{y}$ & $x$ & $\boldsymbol{z_1}$ & $y$ & $\boldsymbol{z}_2$  & $z_1$ & $\cdots$ & $\boldsymbol{x}$ & $z_n$      \\
  $\boldsymbol{z_1}$ & $z_n$ &  $\boldsymbol{z_2}$ & $x$ & $\boldsymbol{z}_3$ & $y$ & $\cdots$ & $\boldsymbol{y}$ &  $\vdots$   \\
 $z_2$ & $\vdots$ & $\vdots$ & $z_n$ & $\vdots$ & $x$  & $\cdots$  & $z_1$ &  $\vdots$    \\
  $\vdots$ & $z_2$ & $z_n$ & $\vdots$ & $z_n$ & $\vdots$  & $\cdots$  & $\vdots$ &  $z_1$  \\
  $z_n$ & $\boldsymbol{y}$ & $x$ & $\boldsymbol{z_1}$ & $y$ & $\boldsymbol{z}_2$ & $\cdots$  & $\vdots$ &  $\boldsymbol{x}$   \\
  $x$ & $\boldsymbol{z_1}$ & $y$ & $\boldsymbol{z_2}$ & $z_1$ & $\boldsymbol{z}_3$  & $\cdots$ & $z_n$ &  $\boldsymbol{y}$   \\
\end{tabular}
\end{center} 
\caption{the profile $\sigma(\mathbf{Q})$.} \label{McGarveyPermuted}
\end{figure}

Now we prove the other direction.
\begin{theorem}
  \label{thm:at-least-as-resolute-vscc}
  Let $SC$ be the Split Cycle VSCC and $F$ any VSCC satisfying Coherent Defeat 
  and Tolerant Positive Involvement. Then $F$ refines $SC$:
  for any profile $\mathbf{P}$, $F(\mathbf{P}) \subseteq SC(\mathbf{P})$.
\end{theorem}
\begin{proof}
  We prove the contrapositive. Suppose that $y \not\in SC(\mathbf{P})$. Then there is
  an $x \in X(\mathbf{P})$ such that $(x, y) \in sc(\mathbf{P})$, where $sc$ is the
  Split Cycle VCCR. Let $k = \margin_{\mathbf{P}}(x, y)$ and $\mathcal{M} =
  \mathcal{M}(\mathbf{P})$. If $k \le 2$, then since $\mathbf{P}$ is a linear profile, $k = 2$ and there cannot be a
  majority path from $y$ to $x$. Then by Coherent Defeat, $y \not\in
  F(\mathbf{P})$, and we are done. If $k > 2$, then we use Lemma
  \ref{lem:add-ballot} again. Inductively construct $\mathbf{P}_0, \mathbf{P}_1,
  \cdots, \mathbf{P}_{k-2}$ and $L_1, L_2, \cdots, L_{k-2}$ where
  $\mathbf{P}_{i+1} = \mathbf{P}_i + L_{i+1}$ and $L_{i+1}$ is obtained by Lemma
  \ref{lem:add-ballot} applied to $\mathbf{P}_i$. Then according to the lemma
  and a simple induction:
  \begin{itemize}
  \item $y$ is ranked by $L_i$ above any candidate whom $y$ does not beat
    head-to-head in $\mathbf{P}_{i-1}$;
  \item then by Tolerant Positive Involvement, if $y \not\in
    F(\mathbf{P}_{i+1})$, then $y \not\in F(\mathbf{P}_i)$ for each $i = 0,
    \ldots, k-3$;
  \item $(x, y)$ is in each $sc(\mathbf{P}_i)$, and in particular, $(x, y) \in
    sc(\mathbf{P}_{k-2})$;
  \item $\margin_{\mathbf{P}_{k-2}}(x,y) = 2$.
  \end{itemize}
  Since $\margin_{\mathbf{P}_{k-2}}(x,y) = 2$ and $(x, y) \in
  sc(\mathbf{P}_{k-2})$, there cannot be a majority path from $y$ to $x$. Then
  by Coherent Defeat, $y \not\in F(\mathbf{P}_{k-2})$. Then, by backward
  induction from $k-2$ to $0$, $y \not\in F(\mathbf{P}_0) = F(\mathbf{P})$.\end{proof}

Combining Theorems \ref{thm:no-more-resolute-vscc} and
\ref{thm:at-least-as-resolute-vscc}, the main theorem Theorem
\ref{thm:axiomatization-vscc} follows.

\section{Ballots with Ties}\label{Sec:Ties}
In many applications of voting methods, it is unreasonable to expect voters to
submit a linear ordering of all candidates, since the voter may lack information
or willingness to make strict comparison between all pairs of candidates.
Mathematically, this can be accommodated by using weak orders instead of linear
orders, where candidates are first put into equivalence classes of tied candidates, and then the equivalence classes are linearly ordered. In this
section, we show that a slight modification of the above axioms also axiomatizes
Split Cycle in this generalized framework.

For any $X \subseteq \mathcal{X}$, let $\mathcal{W}(X)$ be the set of all strict
weak orders on $X$. Recall that $\succ$ is a strict weak order on $X$ iff
$\succ$ is irreflexive, transitive, and \emph{negatively transitive}: for all
$x, y, z \in X$, if $x \not \succ y$ and $y \not \succ z$, then $x \not \succ
z$. Given a strict weak order $\succ$ on $X$, we say that $x, y \in X$ are in a
\emph{tie}, written $x \sim y$, if $x \not\succ y$ and $y \not\succ x$. The
relation $\succ$ being a strict weak order guarantees that being in a tie is an
equivalence relation, and note that the empty set $\varnothing$ is a strict
weak order where all candidates are in a big tie. Now a \emph{profile} is a
function $\mathbf{P}: V \to \mathcal{W}(X)$ for some nonempty finite $V \subset
\mathcal{V}$ and $X \subset \mathcal{X}$. Given a nonempty finite set $X$, we typically
specify a strict weak order ${\succ} \in \mathcal{W}(X)$ by the notation
\begin{align*}
  P_1 \succ P_2 \succ \cdots \succ P_m
\end{align*}
where $\{P_1, P_2, \cdots, P_m\}$ is a partition of $X$; e.g., where
$X=\{a,b,c,d,e\}$, we may write $\{c,e\} \succ \{b\} \succ \{a,d\}$. This means
that $\{P_1, P_2, \cdots, P_m\}$ is the set of all $\sim$-equivalence classes
(groups that are tied internally) and for any $i < j$, $x \in P_i$, and $y \in
P_j$, $x \succ y$; these two conditions uniquely determine a strict weak
ordering of $X$. When some $P_i$ is a singleton $\{x\}$, we may write `$\cdots
\succ x \succ \cdots$' instead of `$\cdots \succ \{x\} \succ \cdots$'.

In this section, all notions previously defined on linear profiles are now
defined relative to all profiles, including the definitions of VCCR and VSCC.
Thus, for example, a VCCR must return an asymmetric relation for any profile of
strict weak orders. In some previous definitions, we used the word `above'
applied to places in a ballot. Those uses are intended for strict preference.
Thus, when understanding previous definitions in the context of profiles of
strict weak orders, `$x$ is above $y$ in a ballot $\succ$' is understood as $x
\succ y$, even though in this section we will typically say `strictly above'.
Moreover, the concept of margin is now relative to strict weak orders and is
still counting and comparing only strict preferences; the margin of $x$ over $y$
does not depend directly on how many voters put $x$ or $y$ in a tie. As before,
we say `$x$ is majority preferred to $y$' when the margin of $x$ over $y$ is
greater than $0$, but because of the possibility of ties, the number of voters
putting $x$ strictly above $y$ may not really be a majority among all voters.
But we take this as only a minor verbal inconvenience, as it can be understood
that `majority' here is relative to those voters who express a strict preference
between $x$ and $y$.

When defining the axiom of Monotonicity, we used the notion of `moving a
candidate up one place in a ballot'. To make this precise for strict weak
orders, we use the following definition.

\begin{definition}
  Let $X$ be a finite set, $x \in X$, and $\succ, \succ'$ two strict weak orders
  on $X$ where $x$ is not the greatest element in $\succ$, i.e., there is $x'
  \in X \setminus \{x\}$ such that $x \not\succ x'$. Since $X$ is finite, we can
  find a minimal (relative to $\succ$) $x' \not= x$ such that $x \not \succ x'$.
  We say that $\succ'$ is the result of \textit{moving $x$ up one place in
    $\succ$} if
  \begin{itemize}
  \item for any $y, z \in X \setminus \{x\}$, $y \succ z$ iff $y \succ' z$;
  \item if $x$ is in a tie with $x'$ in $\succ$, then in $\succ'$, $x$ is not in
    a tie with any other element, $x \succ x'$, and for any $y \in X \setminus
    \{x\}$, if $y \succ x$ then $y \succ' x$;
  \item if $x$ is not in a tie (relative to $\succ$) with any other element,
    then $x$ is in a tie with $x'$ relative to $\succ'$.
  \end{itemize}
\end{definition}
\noindent Thus, intuitively, to move $x$ up one place, we either break $x$ from a tie so
that $x$ is immediately above those with whom $x$ previously tied, or we 
merge $x$ into the tied group that was immediately above $x$ if $x$ was not in a tie.
Then Lemma \ref{lem:majority-defeat} holds under this definition of `moving up one place'.

Still, the generalization to strict weak orders requires a new axiom. As we
noted, the empty set that puts every candidate in a single tie is a legal
ballot, and intuitively such a ballot should not affect the outcome of an
election. We codify this requirement as an axiom.
\begin{definition}
  A VCCR or VSCC $f$ satisfies \textit{Neutral Indifference} if for any profile
  $\mathbf{P}$ and $\mathbf{P}'$ obtained from $\mathbf{P}$ by adding a voter
  whose ballot is the empty set, we have $f(\mathbf{P}) = f(\mathbf{P}')$.
\end{definition} 
\noindent While Neutral Indifference is similar to Neutral Reversal, it does not follow
from Neutral Reversal without enough help from other axioms. For example, we can
easily devise VSCCs satisfying Anonymity, Neutrality, and Neutral Reversal by
running different rules depending on the parity of the number of voters.
Full Homogeneity (requiring precisely the same output for $\mathbf{P}$ and $2\mathbf{P}$, in contrast to Upward Homogeneity) is one axiom that can reduce Neutral Indifference to
Neutral Reversal.
\begin{proposition}\label{FullHomogeneity}
  Any VCCR or VSCC $f$ satisfying Neutral Reversal and Homogeneity
  satisfies Neutral Indifference.
\end{proposition}
\begin{proof}
  Let $\mathbf{P}: V \to \mathcal{W}(X)$ be any profile,
  $\mathbf{P}+\varnothing$ a profile obtained by adding a voter $v$ with
  $\varnothing$ as her ballot to $\mathbf{P}$, and $f$ a VCCR or VSCC satisfying
  Neutral Reversal and Homogeneity. Note that we can first double $\mathbf{P}$
  to $2\mathbf{P}$ (not using $v$) and obtain
  $2\mathbf{P}+\varnothing+\varnothing$ by adding $v$ and another $v'$ with
  $\varnothing$ as both of their ballots. Now by Homogeneity, $f(\mathbf{P}) =
  f(2\mathbf{P})$. By Neutral Reversal and noting that $\varnothing =
  \varnothing^{-1}$, $f(2\mathbf{P}) = f(2\mathbf{P} + \varnothing +
  \varnothing)$. But $2\mathbf{P} + \varnothing + \varnothing$ is a doubling of
  $\mathbf{P} + \varnothing$. So by Homogeneity and chaining the previous
  equalities, $f(\mathbf{P}) = f(\mathbf{P} + \varnothing)$.
\end{proof}

We are now prepared to characterize the Split Cycle VCCR over profiles of strict weak orders. We replace Upward Homogeneity from Theorem \ref{Thm:Main} with Neutral Indifference.

\begin{theorem}
  \label{thm:no-more-resolute-vccr-with-ties}
  Let $f$ be a VCCR satisfying Anonymity, Neutrality, Availability, Neutral
  Indifference, Monotonicity (for two-candidate profiles), Neutral Reversal,
  and Coherent IIA. Then Split Cycle refines $f$: for any profile $\mathbf{P}$, $sc(\mathbf{P}) \supseteq
  f(\mathbf{P})$.
\end{theorem}
\begin{proof}
  Let $f$ be a VCCR satisfying the stated axioms. Toward a contradiction, let
  $\mathbf{P}$ be a profile and $x, y \in X(\mathbf{P})$ such that $(x, y) \in
  f(\mathbf{P}) \setminus sc(\mathbf{P})$. By Lemma \ref{lem:majority-defeat}
  applied to $\mathbf{P}_{|\{x, y\}}$ and Coherent IIA, $m :=
  \margin_{\mathbf{P}}(x, y) = \margin_{\mathbf{P}_{|\{x, y\}}}(x, y) > 0$. By the
  definition of Split Cycle, there are $z_1, z_2, \cdots, z_n \in X(\mathbf{P})$
  such that $( x, y, z_1, z_2, \cdots, z_n, x )$ is a majority cycle
  and $\margin_{\mathbf{P}}(x, y)$ is smallest among the margins of the
  consecutive edges.

  Now we add voters to $\mathbf{P}$ without affecting the result of $f$. Let
  $N_{\mathbf{P}}(x \succ y)$ be the number of voters in $\mathbf{P}$ who rank
  $x$ above $y$, $N_{\mathbf{P}}(x \prec y)$ the number of voters who rank $x$
  below $y$, and $N_{\mathbf{P}}(x \sim y)$ the number of voters who put $x$ and
  $y$ in a tie. Fix an arbitrary strict weak order $L$ on $X(\mathbf{P})$ in which $x \succ y$. Then let $\mathbf{P}'$ be the result of adding to $\mathbf{P}$
  \begin{itemize}
  \item $\max(3m - N_\mathbf{P}(x \succ y), 0)$ many pairs of voters who vote $L$ and $L^{-1}$, respectively, and
  \item $\max((2n - 1)m - N_\mathbf{P}(x \sim y), 0)$ many voters whose ballots
    have only ties.
  \end{itemize}
  Note first that by Neutral Reversal and Neutral Indifference, $(x, y)$ is
  still in $f(\mathbf{P}')$. Also, by the number of pairs of voters and voters
  added, $\margin_{\mathbf{P}'}(x, y)$ is still $m$, and we have  $N_{\mathbf{P}'}(x \succ y)
  \ge 3m$, $N_{\mathbf{P}'}(x \prec y) = N_{\mathbf{P}'}(x \succ y) - m \ge 2m$,
  and $N_{\mathbf{P}'}(x\sim y) \ge (2n-1)m$.

  Next we construct a profile $\mathbf{Q}$, depicted in Figure
  \ref{McGarveyTie}, using some voters in $\mathbf{P}'$ such that
  $\mathcal{M}(\mathbf{Q})$ consists of precisely a majority cycle $( x,
  y, z_1, \cdots, z_n, x )$ with each edge's weight being precisely $m$;
  the tallying at the end of the last paragraph ensures that we have enough
  voters with desired types.
  \renewcommand{\arraystretch}{1.7}
  \begin{figure}[t]
    \begin{center}
      \begin{tabular}{cc|cc|cc|c|cc}
        $m$ & $m$ & $m$ & $m$ & $m$ & $m$ & $\cdots$ & $m$ & $m$ \\
        \hline
        $\boldsymbol{x}$ & \tikzmarknode{12}{$z_n$} & $\boldsymbol{y}$ & \tikzmarknode{14}{$x$} & $\boldsymbol{z_1}$ & \tikzmarknode{16}{$y$} & $\cdots$ & $ \boldsymbol{z_n}$ & \tikzmarknode{19}{$z_{n-1}$}      \\
        $\boldsymbol{y}$ &  $\vdots$ & $\boldsymbol{z_1}$ & $z_n$ &  $\boldsymbol{z_2}$ & $x$ & $\cdots$ & $\boldsymbol{x}$ & $\vdots$   \\
        \tikzmarknode{31}{$z_1$} & $\vdots$ & \tikzmarknode{33}{$z_2$} & $\vdots$ & \tikzmarknode{35}{\raisebox{-4pt}{$\vdots$}} & $z_n$ & $\cdots$ & \tikzmarknode{38}{$y$} & $z_1$    \\
        $\vdots$ &  \tikzmarknode{42}{$z_1$} & $\vdots$ & \tikzmarknode{44}{$z_2$} & $z_n$ & \tikzmarknode{46}{$\vdots$} & $\cdots$ & $z_1$ & \tikzmarknode{49}{$y$}   \\
        $\vdots$ &  \tikzmarknode{52}{$\boldsymbol{x}$} & $z_n$ & \tikzmarknode{54}{$\boldsymbol{y}$} & $x$ & \tikzmarknode{56}{$\boldsymbol{z_1}$} & $\cdots$ & $\vdots$ & \tikzmarknode{59}{$\boldsymbol{z_n}$}   \\
        \tikzmarknode{61}{$z_n$} & \tikzmarknode{62}{$\boldsymbol{y}$} & \tikzmarknode{63}{$x$} & \tikzmarknode{64}{$\boldsymbol{z_1}$} & \tikzmarknode{65}{$y$} & \tikzmarknode{66}{$\boldsymbol{z_2}$} & {$\cdots$} & \tikzmarknode{68}{$z_{n-1}$} & \tikzmarknode{69}{$\boldsymbol{x}$}   \\
      \end{tabular}
      \begin{tikzpicture}[overlay, remember picture]
        \draw[line width=1.4em, line cap=round, draw=gray, opacity=0.3] 
        ([yshift=0.3em] pic cs:31)--([yshift=-0.1em] pic cs:61);
        \draw[line width=1.4em, line cap=round, draw=gray, opacity=0.3] 
        ([yshift=0.3em] pic cs:12)--([yshift=-0.1em] pic cs:42);
        \draw[line width=1.4em, line cap=round, draw=gray, opacity=0.3] 
        ([yshift=0.3em] pic cs:52)--([yshift=-0.1em] pic cs:62);

        \draw[line width=1.4em, line cap=round, draw=gray, opacity=0.3] 
        ([yshift=0.3em] pic cs:33)--([yshift=-0.1em] pic cs:63);
        \draw[line width=1.4em, line cap=round, draw=gray, opacity=0.3] 
        ([yshift=0.3em] pic cs:14)--([yshift=-0.1em] pic cs:44);
        \draw[line width=1.4em, line cap=round, draw=gray, opacity=0.3] 
        ([yshift=0.3em] pic cs:54)--([yshift=-0.1em] pic cs:64);

        \draw[line width=1.4em, line cap=round, draw=gray, opacity=0.3] 
        ([yshift=0.3em] pic cs:35)--([yshift=-0.1em] pic cs:65);
        \draw[line width=1.4em, line cap=round, draw=gray, opacity=0.3] 
        ([yshift=0.3em] pic cs:16)--([yshift=-0.1em] pic cs:46);
        \draw[line width=1.4em, line cap=round, draw=gray, opacity=0.3] 
        ([yshift=0.3em] pic cs:56)--([yshift=-0.1em] pic cs:66);

        \draw[line width=2.5em, line cap=round, draw=gray, opacity=0.3] 
        ([yshift=-0.1em] pic cs:38)--([yshift=0.4em] pic cs:68);
        \draw[line width=2.5em, line cap=round, draw=gray, opacity=0.3] 
        ([yshift=-0.1em] pic cs:19)--([yshift=0.4em] pic cs:49);
        \draw[line width=1.4em, line cap=round, draw=gray, opacity=0.3] 
        ([yshift=0.3em] pic cs:59)--([yshift=-0.1em] pic cs:69);
      \end{tikzpicture}
    \end{center}
    \caption{the profile $\mathbf{Q}$ for Theorem \ref{thm:no-more-resolute-vccr-with-ties}. Candidates inside the same gray area are tied.} \label{McGarveyTie}
  \end{figure}
  \begin{itemize}
  \item Take $m$ voters in $\mathbf{P}'$ ranking $x$ above $y$, and let their
    ballots in $\mathbf{Q}$ be $x \succ y \succ \{z_1, \cdots, z_n\}$. Then take
    $m$ voters in $\mathbf{P}'$ tying $x$ with $y$, and let their ballots in
    $\mathbf{Q}$ be $\{z_1, \cdots, z_n\} \succ \{x, y\}$.
  \item Take another $m$ voters ranking $y$ above $x$ in $\mathbf{P}'$, and let
    their ballots in $\mathbf{Q}$ be $y \succ z_1 \succ \{x, z_2, \cdots,
    z_n\}$. Then take another $m$ voters with $x$ above $y$ in $\mathbf{P}'$,
    and let their ballots in $\mathbf{Q}$ be $\{z_2, \cdots, z_n, x\} \succ \{y,
    z_1\}$.
  \item For each $i = 1, \cdots, n-1$, first take another $m$ voters tying $x$
    with $y$ in $\mathbf{P}'$ and let their ballots in $\mathbf{Q}$ be $z_i
    \succ z_{i+1} \succ \{z_1, \cdots, z_{i-1}, z_{i+2}, \cdots, z_n, x, y\}$.
    Then take another $m$ voters tying $x$ with $y$ in $\mathbf{P}'$ and let
    their ballots in $\mathbf{Q}$ be $\{z_1, \cdots, z_{i-1}, z_{i+2}, \cdots,
    z_n, x, y\} \succ \{z_i, z_{i+1}\}$.
  \item Finally, take another $m$ voters ranking $x$ above $y$ in $\mathbf{P}'$
    and let their ballots in $\mathbf{Q}$ be $z_n \succ x \succ \{y, z_1,
    \cdots, z_{n-1}\}$. Then take another $m$ voters ranking $y$ above $x$ in
    $\mathbf{P}'$ and let their ballots in $\mathbf{Q}$ be $\{y, z_1, \cdots,
    z_{n-1}\} \succ \{x, z_n\}$.
  \end{itemize}
  By a standard argument with permutations, using Anonymity, Neutrality, and
  Availability, $(x, y) \not\in f(\mathbf{Q})$. Now let $L$ be an arbitrary
  strict weak order on $x, y, z_1, \cdots, z_n$ with $x L y$. Then let
  $\mathbf{Q}'$ be the result of adding the remaining voters in $V(\mathbf{P}')
  \setminus V(\mathbf{Q})$ to $\mathbf{Q}$ with their ballots given by:
  \begin{itemize}
  \item if the voter ranks $x$ above $y$ in $\mathbf{P}'$, her ballot in
    $\mathbf{Q}'$ is $L$;
  \item if the voter ranks $y$ above $x$ in $\mathbf{P}'$, her ballot in
    $\mathbf{Q}'$ is $L^{-1}$;
  \item if the voter ties $x$ with $y$, then her ballot in $\mathbf{Q}'$ is
    $\varnothing$.
  \end{itemize}
  Two immediate observations follow. First, it is easy to check that
  $\mathbf{P}'_{|\{x, y\}} = \mathbf{Q}'_{|\{x, y\}}$. Second, since
  $\margin_{\mathbf{Q}}(x, y)$ is easily seen to be $m$, the number of voters
  from $V(\mathbf{P}') \setminus V(\mathbf{Q})$ ranking $x$ above $y$ in
  $\mathbf{P}'$ must be equal to the number of voters from $V(\mathbf{P}')
  \setminus V(\mathbf{Q})$ ranking $y$ above $x$. Thus, $\mathbf{Q}'$ is the
  result of adding to $\mathbf{Q}$ some (possibly zero) pairs of voters with
  converse ballots and some (possibly zero) voters with the fully tied ballot.
  Thus:
  \begin{itemize}
  \item $\mathcal{M}(\mathbf{Q}') = \mathcal{M}(\mathbf{Q})$. Since the latter
    is just a majority cycle with each edge's weight being $m$,
    $\mathcal{M}(\mathbf{Q}')$ can be obtained from $\mathcal{M}(\mathbf{P}') =
    \mathcal{M}(\mathbf{P})$ by deleting candidates and edges and lowering
    weights not involving the edge between $x$ and $y$ (the weights of the edges
    in the cycle $x, y, z_1, \cdots, z_n, x$ are originally all at least $m$).
    So Coherent IIA applies, and $(x, y) \in f(\mathbf{Q}')$.
  \item Since $(x, y) \not\in f(\mathbf{Q})$, and $\mathbf{Q}'$ is the result
    adding pairs of converse ballots and ballots with one big tie, by Neutral
    Reversal and Neutral Indifference, $(x, y) \not\in f(\mathbf{Q}')$.
  \end{itemize}
  So we have a contradiction.
\end{proof}
Appeal to Neutral Indifference cannot be eliminated in the proof since we may
have to add an odd number of big tie ballots at some point. In particular, 
$\max((2n - 1)M - N_\mathbf{P}(x \sim y), 0)$ may be odd. 

We now prove the analogous theorem for the Split Cycle VSCC over profiles of strict weak orders.

\begin{theorem}
  \label{thm:no-more-resolute-VSCC-with-ties}
  Let $F$ be a VSCC satisfying Anonymity, Neutrality, Neutral
  Indifference, Monotonicity (for two-candidate profiles), Neutral Reversal,
  and Coherent IIA. Then  $SC$ refines $F$: for any profile $\mathbf{P}$, $SC(\mathbf{P}) \subseteq F(\mathbf{P})$.
\end{theorem}
\begin{proof}
  The proof essentially combines the strategies used for Theorems 
  \ref{thm:no-more-resolute-vscc} and
  \ref{thm:no-more-resolute-vccr-with-ties}. Let $F$ be a VSCC satisfying the
   axioms and $\mathbf{P}$ a profile (allowing ties in ballots). Since
  we are dealing with a VSCC, as in the proof for Proposition
  \ref{thm:no-more-resolute-vscc}, we need to first add voters to prepare for
  a later construction. Let $M$ be the largest margin in the margin graph of
  $\mathbf{P}$. Then fix a linear order $L^*$ on $X(\mathbf{P})$, and let
  $\mathbf{P}'$ be the result of adding to $\mathbf{P}$
  \begin{itemize}
  \item $3M$ pairs of voters in which one voter's ballot is $L^*$ and the other
    voter's ballot is $(L^*)^{-1}$ and
  \item $2|X(\mathbf{P})| \cdot M$ many voters whose ballot is the empty set.
  \end{itemize}
  Clearly $\mathbf{P}$ and $\mathbf{P}'$ have the same margin graph, and by the assumed 
  axioms, $F(\mathbf{P}') = F(\mathbf{P})$. So to show that
  $SC(\mathbf{P}) \subseteq F(\mathbf{P})$, it is enough to show that
  $SC(\mathbf{P}') \subseteq F(\mathbf{P}')$.

  Toward a contradiction, suppose that $y \in SC(\mathbf{P}')$ but $y \not\in
  F(\mathbf{P}')$. Then by Coherent IIA there is ${x \in X(\mathbf{P}')}$ such
  that for any profile $\mathbf{Q}'$ with $\mathbf{P}' \rightsquigarrow_{x, y}
  \mathbf{Q}'$, we have $y \not\in F(\mathbf{Q}')$. Using Lemma
  \ref{lem:majority-defeat} and Coherent IIA, $m := \margin_{\mathbf{P}'}(x, y)
  = \margin_{\mathbf{P}}(x, y)> 0$. Since $y \in SC(\mathbf{P})$, there is a
  majority cycle $( x, y, z_1, \cdots, z_n, x )$ with $(x, y)$ being
  an edge with the smallest margin. Moreover, by the construction of
  $\mathbf{P}'$, using the notation in the proof of Theorem
  \ref{thm:no-more-resolute-vccr-with-ties}, $N_{\mathbf{P'}}(x \succ y) \ge
  3m$, $N_{\mathbf{P}'}(y \succ x) \ge 2m$, and ${M_{\mathbf{P}'}(x \sim y) \ge
  (2n-1)m}$ since $M \ge m$ and $2|X(\mathbf{P})| \ge 2n-1$. Thus, the construction
  of $\mathbf{Q}$ and $\mathbf{Q}'$ in the proof of Theorem~\ref{thm:no-more-resolute-vccr-with-ties} can be applied in exactly the same
  way. Using the symmetries of $\mathbf{Q}$ and Anonymity and Neutrality, if $y \not\in F(\mathbf{Q})$, then $F(\mathbf{Q}) = \varnothing$, contradicting that $F$ is a VSCC. Thus, $y \in F(\mathbf{Q})$, and hence by Neutral Reversal and Neutral Indifference, $y \in F(\mathbf{Q}')$. On the other hand, by construction $\mathbf{P}' \rightsquigarrow_{x, y} \mathbf{Q}'$, so $y \not\in  F(\mathbf{Q}')$. Hence we obtain a contradiction as desired.
\end{proof}

Now we consider the other direction for axiomatization where, in the case of a
VCCR $f$, we seek to show that $(x,y) \in f(\mathbf{P})$ given that $(x, y) \in
sc(\mathbf{P})$. The strategy is to use Lemma \ref{lem:add-ballot} repeatedly
until we can resort to Coherent Defeat since $\margin_{\mathbf{P}}(x, y)$ is
still positive but there is no majority path back from $y$ to $x$. Note that
Lemma \ref{lem:add-ballot} can only be used to lower the margin from $x$ to $y$
to $2$, and to finish the argument, we used the parity constraint for the margin
graphs obtained from profiles using only linear ballots (Lemma \ref{ParityLem}):
the margins must share the same parity, and thus if $\margin_{\mathbf{P}}(x, y)
= 2$, then there cannot be a majority path from $y$ back to $x$ using only
margins smaller than $2$ since the number $1$ is not available. However, now
that we are allowing ballots with ties, the parity constraint is no longer there
to help.

Note first that Lemma \ref{lem:add-ballot} works also for profiles allowing
ties in ballots, since nothing in the proof depends on the parity constraint for
linear profiles; the proof is exactly the same, and the constructed ballot $L$
is still linear. We record the result here.
\begin{lemma}
  \label{lem:add-ballot-with-ties}
  For any profile $\mathbf{P}:V \to \mathcal{W}(X)$ and any $(x, y) \in
  sc(\mathbf{P})$, if $\margin_{\mathbf{P}}(x, y) > 2$, then there is a
  linear ballot $L \in \mathcal{L}(X)$ such that
  \begin{itemize}
  \item for any $z \in X \setminus \{y\}$ with $\margin_{\mathbf{P}}(y, z) \le
    0$, we have $(y, z) \in L$ (in particular, $(y, x) \in L$), and
  \item $(x, y) \in sc(\mathbf{P} + L)$.
  \end{itemize}
\end{lemma}
To apply Lemma \ref{lem:add-ballot-with-ties} directly to our axiomatization
for all profiles allowing ties, one may hope to relax the assumption that
$\margin_{\mathbf{P}}(x, y) > 2$ to $\margin_{\mathbf{P}}(x, y) > 1$. But this
cannot be done as shown by the following example:
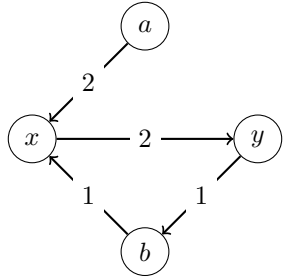
\begin{center}
  \begin{tikzpicture}

    \node[circle,draw, minimum width=0.25in] at (0,0) (x) {$x$}; 
    \node[circle,draw,minimum width=0.25in] at (3,0) (y) {$y$}; 
    \node[circle,draw,minimum width=0.25in] at (1.5,1.5) (a) {$a$}; 
    \node[circle,draw,minimum width=0.25in] at (1.5,-1.5) (b) {$b$}; 

    \path[->,draw,thick] (x) to node[fill=white] {$2$} (y);
    \path[->,draw,thick] (y) to node[fill=white] {$1$} (b);
    \path[->,draw,thick] (b) to node[fill=white] {$1$} (x);
    \path[->,draw,thick] (a) to node[fill=white] {$2$} (x);

  \end{tikzpicture}
\end{center}
For any ballot $L$ in which $y$ is ranked above $x$ and $a$, adding $L$ to 
any profile $\mathbf{P}$ generating the above margin graph would render $(x, y)$
no longer a defeat edge according to Split Cycle, since there will be a
majority cycle $(x, y, a, x)$ in which $(x, y)$ is a weakest edge with margin $1$.

For this reason, we resort to the axiom of Downward Homogeneity to restore
the parity constraint.
\begin{definition}
  A VCCR $f$ (resp.~VSCC $F$) satisfies \emph{Downward Homogeneity} if for any profile
  $\mathbf{P}$ and $2\mathbf{P}$, where $2\mathbf{P}$ is the result of replacing
  each voter in $\mathbf{P}$ by $2$ copies of that voter, $f(\mathbf{P})
  \supseteq f(2\mathbf{P})$ (resp. $F(\mathbf{P}) \subseteq F(2\mathbf{P})$).
\end{definition}
\begin{theorem}
  \label{thm:at-least-as-resolute-vccr-with-ties}
  Let $f$ be a VCCR satisfying Downward Homogeneity, Coherent Defeat, and
  Positive Involvement in Defeat. Then $f$ refines the Split Cycle VCCR: for any
  profile $\mathbf{P}$, $f(\mathbf{P}) \supseteq sc(\mathbf{P})$.
\end{theorem}
\begin{proof}
  Pick an arbitrary profile $\mathbf{P}: V \to \mathcal{W}(X)$ and  $x, y \in
  X$ such that $(x, y) \in sc(\mathbf{P})$. Using Downward Homogeneity, to show
  that $(x, y) \in f(\mathbf{P})$, it is enough to show that $(x, y) \in
  f(2\mathbf{\mathbf{P}})$. Let $k = \margin_{2\mathbf{P}}(x, y)$ (which must be
  at least $2$) and inductively define $\mathbf{P}_0, \cdots, \mathbf{P}_{k-2}$
  where $\mathbf{P}_0 = 2\mathbf{P}$ and $\mathbf{P}_{i+1} = \mathbf{P}_i + L_i$
  with $L_i$ obtained by applying Lemma \ref{lem:add-ballot-with-ties} to
  $\mathbf{P}_i$. An easy induction shows that $(x, y) \in sc(\mathbf{P}_{k-2})$
  and $\margin_{\mathbf{P}_{k-2}}(x, y) = 2$. But more importantly, since each
  $L_i$ is linear, the parity constraint is preserved, and in each
  $\mathbf{P}_i$ inductively constructed, the margins in $\mathbf{P}_i$'s margin
  graph share the same parity. This means that there is no majority edge in
  $\mathbf{P}_{k-2}$ with margin $1$. Thus, there is no majority path from $y$
  back to $x$ in $\mathbf{P}_{k-2}$. By Coherent Defeat, $(x, y) \in
  f(\mathbf{P}_{k-2})$, and by backward induction using Positive Involvement in
  Defeat, $(x, y) \in f(2\mathbf{P})$.
\end{proof}
\begin{theorem}
  \label{thm:at-least-as-resolute-vscc-with-ties}
  Let $F$ be a VSCC satisfying Downward Homogeneity, Coherent Defeat, and
  Tolerant Positive Involvement. Then $F$ refines the Split Cycle VSCC: for any
  profile $\mathbf{P}$, $F(\mathbf{P}) \subseteq SC(\mathbf{P})$.
\end{theorem}
\begin{proof}
  By combining the parity argument in the proof of Theorem
  \ref{thm:at-least-as-resolute-vccr-with-ties} and the inductive argument
  using Tolerant Positive Involvement in the proof of Theorem
  \ref{thm:at-least-as-resolute-vscc}.
\end{proof}

To sum up, now we have the following axiomatizations of Split Cycle when ties
are allowed in ballots.
\begin{theorem}
  Allowing ties in ballots, the Split Cycle VCCR is the unique VCCR satisfying
  Anonymity, Neutrality, Availability, Downward Homogeneity, Neutral
  Indifference, Neutral Reversal, Monotonicity (for two-candidate profiles),
  Coherent IIA, Coherent Defeat, and Positive Involvement in Defeat. Downward
  Homogeneity and Neutral Indifference can be replaced by Homogeneity.
\end{theorem}
\begin{theorem}
  Allowing ties in ballots, the Split Cycle VSCC is the unique VSCC satisfying
  Anonymity, Neutrality, Downward Homogeneity, Neutral
  Indifference, Neutral Reversal, Monotonicity (for two-candidate profiles),
  Coherent IIA, Coherent Defeat, and Tolerant Positive Involvement. Downward
  Homogeneity and Neutral Indifference can be replaced by Homogeneity.
\end{theorem}

\section{The Necessity of Three Special Axioms}
\label{sec:necessity}
In this section, we show the necessity of the three special axioms in the axiomatization of Split Cycle: Coherent IIA, Coherent Defeat, and Positive Involvement in Defeat (and Tolerant Positive Involvement in the context of VSCCs). For simplicity, we go back to the linear ballot setting.

\subsection{Coherent IIA}
We first exhibit a VCCR satisfying all the standard axioms, Coherent Defeat, and Positive Involvement in Defeat, but not Coherent IIA. By Theorem \ref{thm:at-least-as-resolute-vccr}, this VCCR must refine  Split Cycle.

\begin{definition}\label{scwc}
  Recall from Definition \ref{WCdef} the weighted covering VCCR $wc$ defined by $(x, y) \in wc(\mathbf{P})$ iff $\margin_{\mathbf{P}}(x, y) > 0$ and for all $z \in X(\mathbf{P})$, $\margin_{\mathbf{P}}(x, z) \ge \margin_{\mathbf{P}}(y, z)$. Define the VCCR $scwc$ as the profile-wise union of $sc$ and $wc$.\footnote{A previous version of this paper used $\margin_{\mathbf{P}}(x, z) > \margin_{\mathbf{P}}(y, z)$ in place of $\margin_{\mathbf{P}}(x, z) \ge \margin_{\mathbf{P}}(y, z)$ in Definition~\ref{scwc}. We are grateful to Dan Lepor (personal communication) for pointing out that the version of $scwc$ with $\geq$ in place of $>$ also satisfies Positive Involvement. Clearly using $\geq$ makes $scwc$ more resolute.} That is, $scwc(\mathbf{P}) = sc(\mathbf{P}) \cup wc(\mathbf{P})$. In other words, $x$ defeats $y$ according to $scwc$ iff either $x$ defeats $y$ according to $sc$ or $x$ defeats $y$ according to $wc$.
\end{definition}

It requires a delicate argument to show that $scwc$ satisfies Availability. 
\begin{proposition}\label{SCWCacyclic}
  The VCCR $scwc$ satisfies Availability. In fact, $scwc(\mathbf{P})$ is always acyclic.
\end{proposition}
\begin{proof}
  Pick any profile $\mathbf{P}$. We claim that for any $x, y, z \in X(\mathbf{P})$, if $(x, y) \in sc(\mathbf{P})$ and $(y, z) \in wc(\mathbf{P})$, then $(x, z) \in sc(\mathbf{P})$. Suppose $(x, y) \in sc(\mathbf{P})$ and $(y, z) \in wc(\mathbf{P})$. From $(y, z) \in wc(\mathbf{P})$, we have  $\margin_{\mathbf{P}}(y, x) \ge \margin_{\mathbf{P}}(z, x)$, which implies (1) $\margin_{\mathbf{P}}(x, z) \ge \margin_{\mathbf{P}}(x, y)$. Now pick any majority path $\rho$ from $z$ to $x$, and our goal is to show that $\margin_{\mathbf{P}}(x, z) > \strength_{\mathbf{P}}(\rho)$. Let $a$ be the candidate immediately after $z$ in the majority path $\rho$, and let $\rho'$ be the majority path that starts with $y$, goes to $a$ immediately, and then continues to $x$ following $\rho$. Since $(y, z) \in wc(\mathbf{P})$, $\margin_{\mathbf{P}}(y, a) \ge \margin_{\mathbf{P}}(z, a)$. This means (2) $\strength_{\mathbf{P}}(\rho') \ge \strength_{\mathbf{P}}(\rho)$. Since $(x, y) \in sc(\mathbf{P})$, $\margin_{\mathbf{P}}(x, y) > \strength_{\mathbf{P}}(\rho')$. Connecting this with inequalities (1) and (2), $\margin_{\mathbf{P}}(x, z) > \strength_{\mathbf{P}}(\rho)$, which shows that $(x, z) \in sc(\mathbf{P})$.

  Now we note that $wc(\mathbf{P})$ is transitive and irreflexive and hence acyclic. So if there is a cycle in $scwc(\mathbf{P})$, the cycle must contain at least one edge in $sc(\mathbf{P})$. But then we can repeatedly use the above claim and obtain a cycle in $sc(\mathbf{P})$, contradicting the acyclicity of $sc$.
\end{proof}

\begin{proposition}
  The VCCR $scwc$ satisfies all the standard axioms, Coherent Defeat, and Positive Involvement in Defeat. Moreover, it satisfies Majority Defeat. Thus, using Proposition \ref{prop:pid-to-tpi}, $\overline{scwc}$ satisfies all the standard axioms, Coherent Defeat, and Tolerant Positive Involvement.
\end{proposition}
\begin{proof}
  Since $scwc(\mathbf{P})$ only depends on the margin graph of $\mathbf{P}$---and indeed only on the ordering of margins by size, not the numerical values---it follows that Anonymity, Neutrality, Homogeneity, and Neutral Reversal are  satisfied. Monotonicity is also easy to see as both $sc$ and $wc$ satisfy Monotonicity. We showed in Proposition~\ref{SCWCacyclic} that $scwc$ satisfies Availability.  Next, $scwc$ satisfies Coherent Defeat since it refines $sc$, and $sc$ satisfies Coherent Defeat. Since both $sc$ and $wc$ satisfy Positive Involvement in Defeat by Propositions \ref{SCPID} and \ref{WCPID}, reasoning by cases we see that $scwc$ satisfies Positive Involvement in Defeat. Thus $scwc$ satisfies all the stated axioms. Since $sc$ and $wc$ both satisfy Majority Defeat, $scwc$ also satisfies Majority Defeat, which implies that $\overline{scwc}$ satisfies Tolerant Positive Involvement by Proposition \ref{prop:pid-to-tpi}. That $\overline{scwc}$ satisfies the other stated axioms is easy to check.
\end{proof}

\begin{proposition}
  Neither the VCCR $scwc$ nor the VSCC $\overline{scwc}$ satisfies Coherent IIA.
\end{proposition}
\begin{proof}
  Let $\pfp$ be the following profile with its margin graph displayed on the right:
  \begin{center} 
    $\pfp$ \hspace{2em}
  \begin{tabular}[]{ccccccc}
    $v_1$ & $v_{2, 3}$ & $v_{4, 5}$ & $v_{6, 7}$ & $v_8$ & $v_9$ \\
    \hline 
    $a$ & $a$ & $b$ & $y$ & $x$ & $x$ \\
    $x$ & $y$ & $a$ & $x$ & $y$ & $b$ \\
    $y$ & $x$ & $x$ & $b$ & $b$ & $a$ \\
    $b$ & $b$ & $y$ & $a$ & $a$ & $y$
  \end{tabular}
  \hspace{3em}
  \begin{tikzpicture}[baseline=(current bounding box.center)]
    \node[circle,draw, minimum width=0.25in] at (0,0) (x) {$x$}; 
    \node[circle,draw,minimum width=0.25in] at (2.15,0) (y) {$y$}; 
    \node[circle,draw,minimum width=0.25in] at (0,2.15) (a) {$a$}; 
    \node[circle,draw,minimum width=0.25in] at (2.15,2.15) (b) {$b$}; 

    \path[->,draw,thick] (x) to node[fill=white] {$1$} (y);
    \path[->,draw,thick] (y) to node[fill=white] {$3$} (b);
    \path[->,draw,thick] (b) to node[fill=white] {$3$} (a);
    \path[->,draw,thick] (a) to node[fill=white] {$1$} (x);
    \path[->,draw,thick] (a) to node[near start, fill=white] {$3$} (y);
    \path[->,draw,thick] (x) to node[near start, fill=white] {$5$} (b);
  \end{tikzpicture}
  \end{center}
  The column under $v_{i, i+1}$ means that the voters $v_i$ and $v_{i+1}$ submitted the same ballot. Observe that $sc(\pfp) = \{(x, b)\}$ and $wc(\pfp) = \{(x, y)\}$. (A quick method to check: if there is a 3-cycle, then none of the three edges are in $wc$.) Thus $scwc(\pfp) = \{(x, b), (x, y)\}$ and $\overline{scwc}(\pfp) = \{a, x\}$.

Now consider profile $\pfq$ and its margin graph:
  \begin{center} 
  $\pfq$ \hspace{2em}
  \begin{tabular}[]{ccccccc}
    $v_1$ & $v_{2, 3}$ & $v_{4, 5}$ & $v_{6, 7}$ & $v_8$ & $v_9$ \\
    \hline 
    $a$ & $a$ & $b$ & $y$ & $x$ & $x$ \\
    $x$ & $y$ & $a$ & $x$ & $y$ & $b$ \\
    $y$ & $b$ & $x$ & $b$ & $b$ & $a$ \\
    $b$ & $x$ & $y$ & $a$ & $a$ & $y$
  \end{tabular}
  \hspace{3em}
  \begin{tikzpicture}[baseline=(current bounding box.center)]
    \node[circle,draw, minimum width=0.25in] at (0,0) (x) {$x$}; 
    \node[circle,draw,minimum width=0.25in] at (2.15,0) (y) {$y$}; 
    \node[circle,draw,minimum width=0.25in] at (0,2.15) (a) {$a$}; 
    \node[circle,draw,minimum width=0.25in] at (2.15,2.15) (b) {$b$}; 

    \path[->,draw,thick] (x) to node[fill=white] {$1$} (y);
    \path[->,draw,thick] (y) to node[fill=white] {$3$} (b);
    \path[->,draw,thick] (b) to node[fill=white] {$3$} (a);
    \path[->,draw,thick] (a) to node[fill=white] {$1$} (x);
    \path[->,draw,thick] (a) to node[near start, fill=white] {$3$} (y);
    \path[->,draw,thick] (x) to node[near start, fill=white] {$1$} (b);
  \end{tikzpicture}
  \end{center}
Observe that $\pfp \rightsquigarrow_{x, y} \pfq$ (only $v_2$ and $v_3$ changed ballots). Now $(x, y) \not\in scwc(\pfq) = \varnothing$. Thus $\pfp, \pfq$ witness $scwc$ failing Coherent IIA. To show that $\overline{scwc}$ fails Coherent IIA, we must be more careful. Starting with the fact that $y \not\in \overline{scwc}(\pfp)$, we must show that for any candidate $u\not=y$, there is $\pfp'$ such that $\pfp \rightsquigarrow_{y, u} \pfp'$ and $y \in \overline{scwc}(\pfp')$. If $u = b$, we can use $\pfp' = \pfp_{|\{y, b\}}$. If $u = a$, we can use $\pfp' = \pfp_{|\{y, a, b\}}$. And finally if $u = x$, we use $\pfq$.
\end{proof}
Thus, for either the VCCR $sc$ or the VSCC $SC$, Coherent IIA cannot be dropped from our axioms.

\subsection{Positive Involvement in Defeat and Tolerant Positive Involvement}

Now we turn to the necessity of Positive Involvement in Defeat and Tolerant Positive Involvement. For the following definition, as a convention we define the minimum of the empty set to be $\infty$ (an object greater than every natural number) and the maximum of the empty set to be $0$.

\begin{definition}
  For any positively weighted directed graph $\mathcal{M}$ and simple path
  $\rho$ from $x$ to $y$ in $\mathcal{M}$, define the \emph{ignore-source
    strength $\stris_{\mathcal{M}}(\rho)$ of $\rho$} to be the minimum of the
  weights of the consecutive edges in $\rho$ that do not start from $x$. The \textit{ignore-source strength $\stris_{\mathcal{M}}(x,y)$ from $x$ to $y$} is the
  maximum of $\stris_{\mathcal{M}}(\rho)$ for all majority paths $\rho$ from $x$ to $y$.
  Then define the VCCR \emph{Ignore-source Split Cycle} $isc$ by: for any
  profile $\mathbf{P}$ and $x,y \in X(\mathbf{P})$, $(x, y) \in isc(\mathbf{P})$
  iff $\margin_{\mathbf{P}}(x, y) > \stris_{\mathcal{M}(\mathbf{P})}(y,x)$. The
  Ignore-source Split Cycle VSCC $ISC$ is defined as $\overline{isc}$ (recall
  Lemma \ref{VCCRtoVSCC}).
\end{definition}

From the definition, the following observations are immediate. 
\begin{lemma}
  \label{lem:str-is}
  Let $\mathcal{M}$ be a margin graph and $x, y$ vertices in $\mathcal{M}$.
  \begin{enumerate}
  \item $\stris_{\mathcal{M}}(x, y) = \infty$ iff $(x, y)$ is an edge in
    $\mathcal{M}$. $\stris_{\mathcal{M}}(x, y) = 0$ iff $y$ is not reachable
    from $x$.
  \item\label{lem:str-is.2} $\stris_{\mathcal{M}}(x, y) = \strength_{\mathcal{M}'}(x, y)$ where
    $\mathcal{M}'$ is the weighted directed graph obtained from $\mathcal{M}$ by increasing to $\infty$ the weight of each
    edge whose source is $x$. Here $\strength_{\mathcal{M}'}$ is calculated in the
    standard way: it is the maximum of the strengths of all paths from $x$ to $y$, 
    where the strength of a path is the minimum of the weights of the
    consecutive edges in that path. Thus `ignore source' can also be viewed as
    `infinity source'.
  \item Suppose $\stris_{\mathcal{M}}(x, y) < \infty$. Then if we delete all
    edges in $\mathcal{M}$ that (1) do not start from $x$ and (2) have weights
    no greater than $\stris_{\mathcal{M}}(x,y)$, then $y$ is no longer
    reachable from $x$. In other words, when $\stris_{\mathcal{M}}(x, y) <
    \infty$, the set of all edges not starting from $x$ and with weights no
    greater than $\stris_{\mathcal{M}}(x, y)$ is a cut from $x$ to $y$.
  \item For any cut $C$ from $x$ to $y$ that does not use any edge starting from
    $x$, $\stris_{\mathcal{M}}(x,y)$ is at most the maximal weight of the
    edges in $C$.
  \end{enumerate}
\end{lemma}

\begin{proposition}
  $isc$ satisfies Anonymity, Neutrality, Availability, Homogeneity,
  Monotonicity, Neutral Reversal, Coherent IIA, Coherent Defeat, and First-place
  Involvement in Defeat.
\end{proposition}
\begin{proof}
  We only verify First-place Involvement in Defeat here; the rest can be
  verified in the same way as they are verified for Split Cycle. Suppose $x$
  does not defeat $y$ according to $isc$ in $\mathbf{P}$, and let
  $\mathbf{P}'$ be the result of adding a ballot that ranks $y$ at the top. Then
  $\margin_{\mathbf{P}}(x,y) \le \stris_{\mathcal{M}(\mathbf{P})}(y, x)$, and we
  want to show the same inequality for $\mathbf{P}'$. Since the new ballot ranks
  $y$ at the top, $\margin_{\mathbf{P}'}(x, y) = \margin_{\mathbf{P}}(x, y) - 1$.
  So we only need to show that $\stris_{\mathcal{M}(\mathbf{P}')}(y, x) \ge
  \stris_{\mathcal{M}(\mathbf{P})}(y, x) - 1$. For the standard strength, this
  is trivial since we only added a single ballot, so the margin changes can be
  at most $1$. However, by the example in the proof of the next proposition,
  extra care must be taken for $\stris$. Since the added ballot ranks $y$ at the
  top, all the edges starting from $y$ in $\mathcal{M}(\mathbf{P})$ are
  strengthened in weight, and in particular every edge in
  $\mathcal{M}(\mathbf{P})$ starting from $y$ is still in
  $\mathcal{M}(\mathbf{P}')$. Thus, letting $\mathcal{M}(\mathbf{P})^+$ and
  $\mathcal{M}(\mathbf{P}')^+$ be the results of raising the weights of all edges
  from $y$ to $\infty$ in $\mathcal{M}(\mathbf{P})$ and
  $\mathcal{M}(\mathbf{P}')$, respectively, the weight \emph{decrease} of the
  edges from $\mathcal{M}(\mathbf{P})^+$ to $\mathcal{M}(\mathbf{P}')^+$ is at
  most $1$ (including eliminating edges with weight $1$); the weight
  \emph{increase} could be infinite due to new edges starting from $y$, but this
  is irrelevant for the direction of the inequality we are trying to show. Hence
  $\strength_{\mathcal{M}(\mathbf{P}')^+}(y, x) \ge
  \strength_{\mathcal{M}(\mathbf{P})^+}(y, x) - 1$. Then using 
  Lemma \ref{lem:str-is}.\ref{lem:str-is.2}, we~are~done.
\end{proof}
As we have seen that First-place Involvement in Defeat entails Positive
Involvement, $isc$ satisfies Positive Involvement as well.

\begin{proposition}
  $isc$ does not refine Split Cycle, and it fails Positive Involvement in
  Defeat. Moreover, $ISC$ fails Tolerant Positive Involvement.
\end{proposition}
\begin{proof}
  Consider a profile $\mathbf{P}$ whose margin graph is 
  
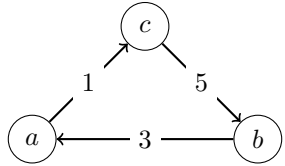
\begin{center}
    \begin{tikzpicture}
      \node[circle,draw, minimum width=0.25in] at (0,0) (a) {$a$};
      \node[circle,draw,minimum width=0.25in] at (3,0) (b) {$b$};
      \node[circle,draw,minimum width=0.25in] at (1.5,1.5) (c) {$c$};

      \path[->,draw,thick] (c) to node[fill=white] {$5$} (b);
      \path[->,draw,thick] (b) to node[fill=white] {$3$} (a);
      \path[->,draw,thick] (a) to node[fill=white] {$1$} (c);
    \end{tikzpicture}
  \end{center}
  According to Split Cycle, $b$ defeats $a$. But according to $isc$, $b$ does not defeat $a$, since $\stris(a,b) = 5$ as we need to
  ignore the edge from $a$ to $c$. Since clearly $c$ does not defeat $a$ either
  according to $isc$, $a$ is undefeated and a winner
  according to $ISC$. To see the failure of Positive Involvement in Defeat and
  Tolerant Positive Involvement, let $\mathbf{P}'$ be the result of adding a
  ballot $\succ$ such that $c \succ a \succ b$ (where $a$ is above $b$ and thus above all
  candidates to whom $a$ is not majority preferred) to $\mathbf{P}$. Then
  $\mathcal{M}(\mathbf{P}')$ is
  
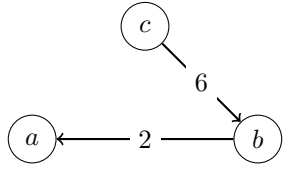
\begin{center}
    \begin{tikzpicture}
      \node[circle,draw, minimum width=0.25in] at (0,0) (a) {$a$};
      \node[circle,draw,minimum width=0.25in] at (3,0) (b) {$b$};
      \node[circle,draw,minimum width=0.25in] at (1.5,1.5) (c) {$c$};

      \path[->,draw,thick] (c) to node[fill=white] {$6$} (b);
      \path[->,draw,thick] (b) to node[fill=white] {$2$} (a);
    \end{tikzpicture}
  \end{center}
  Now $b$ defeats $a$ according to $isc$ since there is no path from $a$ to $b$
  and the ignore-source strength from $a$ to $b$ decreased from $5$ above to $0$. \end{proof}

Thus, when axiomatizing Split Cycle, we cannot replace Positive Involvement in
Defeat by First-place Involvement in Defeat for the VCCR or replace Tolerant
Positive Involvement by Positive Involvement for the VSCC. On the other hand, we
can show that Coherent Defeat and First-place Involvement in Defeat provide a
one-sided axiomatization for $isc$ in the direction opposite to what
Holliday and Pacuit \citeyearpar{HP2021} did for the Split Cycle VCCR.
\begin{theorem}
  For any VCCR $f$ that satisfies Coherent Defeat and First-place Involvement in
  Defeat, $f$ refines $isc$.
\end{theorem}
\begin{proof}
  We repeat the same cut argument in Lemma \ref{lem:add-ballot} and the inductive argument in Theorem \ref{thm:at-least-as-resolute-vccr}. The counterpart of Lemma \ref{lem:add-ballot} for $isc$ takes the following form: for any profile $\mathbf{P}: V \to \mathcal{L}(X)$ and  $(x, y) \in isc(\mathbf{P})$, if $\margin_{\mathbf{P}}(x, y) > 2$, then there is a ballot $L \in \mathcal{L}(X)$ such that
  \begin{itemize}
    \item $y$ is at the top of $L$, and
    \item $(x, y) \in isc(\mathbf{P} + L)$.
  \end{itemize}
  The proof is almost the same as that of Lemma \ref{lem:add-ballot}: since we are using $\stris_{\mathcal{M}}$, we can guarantee that the cut does not include any edge starting from the defeated candidate according to $isc$. Hence the union of the converse of the cut and the set of all pairs starting from $y$ is acyclic, and the rest of the proof is exactly the same. Then we repeatedly use this counterpart of Lemma \ref{lem:add-ballot} to show that if $(x, y) \in isc(\mathbf{P})$, then there is a sequence of ballots $L_1, \cdots, L_n$, all putting $y$ at the top, such that $(x, y) \in isc(\mathbf{P}+L_1+ \cdots + L_n)$, and in $\mathbf{P} + L_1 + \cdots + L_n$, $x$ defeats $y$ merely by Coherent Defeat. Then by First-place Involvement in Defeat, $x$ defeats $y$ in $\mathbf{P}$ according to $f$.
\end{proof}

\subsection{Coherent Defeat}
The example showing the necessity of Coherent Defeat is the following VCCR. 
\begin{definition}
  Let $oca$ ($o$ne-$c$overs-$a$ll) be the VCCR such that for any profile $\pfp$ and $x, y \in X(\pfp)$, $(x, y) \in oca(\pfp)$ iff for all $z \in X(\pfp)$, $\margin_{\pfp}(x, y) > \margin_{\pfp}(y, z)$.
\end{definition}
\begin{proposition}
  The VCCR $oca$ satisfies all the standard axioms, and thus the VSCC $\overline{oca}$ also satisfies all the standard axioms. Moreover, $oca$ satisfies Coherent IIA and Positive Involvement in Defeat, while $\overline{oca}$ satisfies Coherent IIA and Tolerant Positive Involvement.
\end{proposition}
\begin{proof}
  For the standard axioms, the only one worth commenting on  is Availability for $oca$. Indeed, we can show that $oca$ is acyclic, just as we show that $sc$ is acyclic. First, note that $oca$ satisfies Majority Defeat, since if $(x, y) \in oca(\pfp)$, then $\margin_\pfp(x, y) > \margin_\pfp(y, x)$. Also, the weakest edges of any majority cycle cannot be defeats according to $oca$. So there cannot be a defeat cycle according to $oca$.

  To see that $oca$ satisfies Coherent IIA, suppose $(x, y) \in oca(\pfp)$ and $\pfp \rightsquigarrow_{x, y} \pfq$. Then, $\margin_{\pfp}(x, y) = \margin_{\pfq}(x, y) > 0$, while for any $z \in X(\pfq)$, $\margin_{\pfq}(y, z)$ is either $< 0$ or $\le \margin_{\pfp}(y, z)$. In either case, $\margin_\pfq(x, y) > \margin_\pfq(y, z)$. So $(x, y) \in oca(\pfq)$. By Proposition \ref{prop:ciia-vccr-to-vscc}, $\overline{oca}$ also satisfies Coherent IIA. 

  To see that $oca$ satisfies Positive Involvement in Defeat, say $(x, y) \not\in oca(\pfp)$ and $\pfq$ is obtained by adding one ballot ranking $y$ above $x$. This means that there is $z \in X(\pfp)$ such that $\margin_\pfp(x, y) \le \margin_\pfp(y, z)$. Now $\margin_\pfq(x, y) = \margin_\pfp(x, y) - 1$, while $|\margin_\pfp(y, z) - \margin_\pfq(y, z)| \le 1$. Thus $\margin_\pfq(x, y) \le \margin_\pfq(y, z)$, and thus $(x, y) \not\in \pfq$. Using Proposition \ref{prop:pid-to-tpi}, $\overline{oca}$ satisfies Tolerant Positive Involvement.
\end{proof}
\begin{proposition}
  Both $oca$ and $\overline{oca}$ fail Coherent Defeat.
\end{proposition}
\begin{proof}
  Any profile $\pfp$ whose margin graph is 
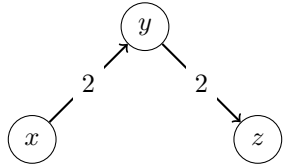
\begin{center}
  \begin{tikzpicture}[baseline]
    \node[circle,draw,minimum width=0.25in] at (0,0) (x) {$x$}; 
    \node[circle,draw,minimum width=0.25in] at (3,0) (z) {$z$}; 
    \node[circle,draw,minimum width=0.25in] at (1.5,1.5) (y) {$y$}; 
    \path[->,draw,thick] (x) to node[fill=white] {$2$} (y);
    \path[->,draw,thick] (y) to node[fill=white] {$2$} (z);
  \end{tikzpicture}
\end{center}
will do. Here $oca(\pfp) = \{(y, z)\}$ and $\overline{oca}(\pfp) = \{x, y\}$, even though $y$ is coherently defeated by $x$ as $\margin_\pfp(x, y) > 0$ and there is no majority path from $y$ to $x$.
\end{proof}

\section{Conclusion}\label{Sec:Conclusion}
We have provided axiomatizations for Split Cycle both as a VCCR and as a VSCC, both over profiles prohibiting ties and over profiles allowing ties. Most of the axioms are largely uncontroversial, while the special axioms---Coherent IIA, Coherent Defeat, Positive Involvement in Defeat, and Tolerant Positive Involvement---have direct intuitive appeal. The axiomatizations also show where and how Split Cycle diverges from similar margin-based voting methods such as Ranked Pairs and Beat Path: Coherent Defeat is a shared starting point; then Positive Involvement in Defeat (resp.~Tolerant Positive Involvement for VSCCs) forces the method to refine Split Cycle; and Coherent IIA forces the method to be refined by Split Cycle.

We conclude with several directions for possible future investigation. First, while our method for the axiomatization of Split Cycle cannot be applied directly to Beat Path and Ranked Pairs or other unaxiomatized margin-based voting methods, our results at least suggest that the axiomatization question for margin-based voting methods could be amenable to analysis of graphs by well-known graph theoretical concepts. In particular, given the similarity between the definitions of Beat Path and Split Cycle, we believe an axiomatization of Beat Path is within reach.

Second, as we mentioned at the beginning of Section \ref{Sec:CharVSCC}, in general a VSCC carries less information than a VCCR, but when we focus on VCCRs and VSCCs satisfying certain axioms, recovering canonically a rationalizing VCCR from a VSCC may be possible. In fact, the axiom of Coherent IIA itself suggests a method of recovery: given a VSCC $F$ that is rationalized by some VCCR $f$ satisfying Coherent IIA and a losing candidate $y$ in a profile $\mathbf{P}$, any candidate $x$ that defeats $y$ according to $f$ should be such that for any $\mathbf{P}'$ with $\mathbf{P} \rightsquigarrow_{x, y} \mathbf{P}'$, $x$ should still defeat $y$ in $\mathbf{P}'$ according to $f$, making $y$ still lose according to $F$ in $\mathbf{P}'$. Viewing this statement about $x$ as a property of a potential defeater of $y$, we can take as the defeaters of $y$ according to the VSCC $F$ all candidates $x \in X(\mathbf{P})$ such that for any $\mathbf{P}'$ with $\mathbf{P} \rightsquigarrow_{x, y}\mathbf{P}'$, $y \not\in F(\mathbf{P}')$. Some of these candidates may not defeat $y$ according to $f$, but we can still use this method to recover a VCCR $f'$ from $F$. It remains to be seen what axioms can be preserved by this method of constructing a VCCR from a VSCC and what other interesting properties this method may exhibit.

Third, as we showed in Section \ref{sec:necessity}, within the set of all the margin-based VCCRs satisfying Monotonicity, Coherent IIA, Coherent Defeat, and First-place Involvement in Defeat, what we call \textit{Ignore-source Split Cycle} is the least refined while Split Cycle is the most refined. It seems to be an interesting problem to classify all VCCRs in this set and also study their ordering by refinement. This could help us better understand the axioms (especially First-place Involvement in Defeat).

Finally, we believe it is important to systematically study the class of refinements of Split Cycle that are, unlike Split Cycle (see \citealt[\S~5.4.2.2]{HP2020}), \textit{asymptotically resolvable}: for any number of candidates, the probability\footnote{According to the uniform distribution on profiles with a given number of candidates and voters.} of selecting multiple winners goes to zero as the number of voters goes to infinity. Methods in this class, which include Ranked Pairs (\citealt{Tideman1987}), River (\citealt{Heitzig2004b}), Beat Path (\citealt{Schulze2011,Schulze2018}), and Stable Voting (\citealt{HP2023stable}), can be viewed as ways of deterministically breaking ties among Split Cycle winners. No known method in this class satisfies Positive Involvement or its variations studied in this paper (see \citealt{Holliday2024} for a related impossibility theorem), but no known impossibility result yet rules out the existence of such a refinement of Split Cycle.

\subsection*{Acknowledgements}

We thank the organizers and audiences at the 2022 Meeting of the Society for Social Choice and Welfare and the Conference on Voting Theory and Preference Aggregation at Karlsruhe Institute of Technology, Germany in October 2023, where this work was presented, and the two reviewers and editor for \textit{Social Choice and Welfare} for their helpful feedback.

\bibliographystyle{plainnat}
\bibliography{axioms}
\end{document}